\documentclass[twoside]{ruthesis}
\usepackage{setspace}
\usepackage{amsmath}
\usepackage{amssymb}
\usepackage[english]{babel}
\usepackage[utf8]{inputenc}
\usepackage[numbers]{natbib}
\usepackage{bibentry}
\usepackage{hyperref}
\hypersetup{
	citecolor = PineGreen,
	colorlinks=true,
	linkcolor=NavyBlue,
	filecolor=magenta,
	urlcolor=blue,
}

\usepackage{enumerate}
\usepackage{graphicx}
\usepackage{bbm}
\usepackage{mathrsfs}
\usepackage{amsthm}
\usepackage[dvipsnames]{xcolor}
\usepackage{verbatim}
\usepackage{tikz-cd}
\usepackage{float}
\usepackage{makecell}
\usepackage{diagbox}

\newcommand\numberthis{\addtocounter{equation}{1}\tag{\theequation}}
\DeclareMathOperator{\R}{\mathbb{R}}
\DeclareMathOperator{\tr}{\mathrm{tr}}
\DeclareMathOperator{\id}{\mathbbm{1}}
\DeclareMathOperator{\im}{\mathbf{i}}
\newcommand\ket[1]{\left| #1 \right\rangle}
\newcommand\bra[1]{\left\langle #1 \right|}
\newcommand\braket[2]{\left\langle #1 |#2\right\rangle}
\newcommand\qmu{q_{\textsc{MU}}}
\newcommand\rh{r_{\textsc{H}}}
\newcommand{\Hil}{\mathcal{H}}
\newcommand{\Lin}{\mathcal{L}}
\newcommand{\Den}{\mathcal{D}}
\newcommand{\Pos}{\mathcal{P}}
\newcommand{\PosS}{\mathcal{P}^*}
\newcommand{\Uni}{\mathcal{U}}
\newcommand{\C}{\mathbb{C}}
\newcommand{\Op}{\mathrm{Op}}

\newcommand{\ketbra}[2]{\ket{#1}\!\!\bra{#2}}
\newcommand{\supp}{\mathrm{supp}}
\newcommand{\CPTP}{\mathrm{CPTP}}
\newcommand{\Map}{\mathcal{M}}
\newcommand\rcp{r_{\textsc{CP}}}
\numberwithin{equation}{chapter}
\numberwithin{table}{chapter}

\usepackage{titlesec}
\usepackage{geometry}
\usepackage{hhline}
\newtheorem{theorem}{Theorem}
\numberwithin{theorem}{chapter}
\newtheorem{lemma}[theorem]{Lemma}
\newtheorem{proposition}[theorem]{Proposition}
\newtheorem{corollary}[theorem]{Corollary}
\newtheorem{definition}[theorem]{Definition}
\newtheorem{remark}[theorem]{Remark}

\title{R\'ENYI DIVERGENCE INEQUALITIES VIA INTERPOLATION, WITH APPLICATIONS TO GENERALISED ENTROPIC UNCERTAINTY RELATIONS}
\ctitle{R\'enyi divergence inequalities via interpolation, with applications to generalised entropic uncertainty relations}
\author{Alexander McKinlay}
\school{University of Technology Sydney} 
\department{Faculty of Engineering and Information Technology} 
\degree{Master of Science (Research) in Computing Sciences}
\super{Dr. Marco Tomamichel}
\cosuper{Dr. Christopher Ferrie}
\address{Sydney, Australia}
\donemonth{January} \doneyear{2021} \makeindex

\titlespacing\section{-4pt}{12pt plus 4pt minus 2pt}{0pt plus 2pt minus 2pt}
\titlespacing\subsection{0pt}{12pt plus 4pt minus 2pt}{0pt plus 2pt minus 2pt}
\titlespacing\subsubsection{0pt}{12pt plus 4pt minus 2pt}{0pt plus 2pt minus 2pt}
\titleformat{\subsubsection}
{\bfseries\itshape\normalsize}{\thesubsubsection}{1em}{}
\begin{document}
\begin{frontmatter}
\pagenumbering{roman}
\newgeometry{
	a4paper,
	left=30mm, right=30mm, top=30mm, bottom=20mm}
\maketitle
\null
\thispagestyle{empty}%
\addtocounter{page}{-1}%
\newpage
\restoregeometry
\setcounter{page}{1}
\addcontentsline
{toc}{frontmatter}{\protect\numberline {\hfil}Acknowledgements}
\begin{titlepage}
	
	\begin{center}\Large{\textbf{Acknowledgments}}\end{center}
	This thesis could not have happened without the consistent and invaluable support of Marco.\\\\
	The inspiration for the interpolation technique was supplied by Dr. \!\!Salman Beigi, whose contributions in discussions about the nature of Pisier norms were more than illuminating.\\\\
	Lucie, for her patience.
	\vspace*{\fill}
	\begin{center}
		This is a CONVENTIONAL THESIS. Published work arising from research during the candidature of the degree has been included in an appendix. 
	\end{center}
\vfill	
\end{titlepage}
\begin{titlepage}
	 
\end{titlepage}
\addcontentsline
{toc}{frontmatter}{\protect\numberline {\hfil}Declaration of publications included in the thesis}
\begin{titlepage}
	\vspace*{\fill}
	\begin{center}
		{\large \bf DECLARATION OF PUBLICATIONS INCLUDED IN THE THESIS}	
		\vskip 10pt plus 3fil
		\begin{tabular}{|rlp{11.4cm}|}
			\hline
			Title&:& Decomposition Rules for Quantum R\'enyi Mutual Information with an \quad Application to Information Exclusion Relations\\
			Parts included&:& Subsections IV.A and IV.B\\
			Authors&:& Alexander McKinlay, Marco Tomamichel \\
			Contributions&:& AM derived results and developed the method of the proofs in consultation with MT.\\
			Status&:& Published  \\
			Journal&:& Journal of Mathematical Physics, vol. 61, no. 7, p. 072202, 2020\\
			Address&:& \url{https://doi.org/10.1063/1.5143862} \\
			\hline
		\end{tabular}
	\end{center}
	\vfill
\end{titlepage}
\setcounter{tocdepth}{2}
\addcontentsline
{toc}{frontmatter}{\protect\numberline {\hfil}Contents}\tableofcontents
\listoftables
\newpage
\begin{abstract}
	\thispagestyle{plain}	
	We investigate quantum R\'enyi entropic quantities, specifically those derived from `sandwiched' divergence. This divergence is one of several proposed R\'enyi generalisations of the quantum relative entropy. We may define R\'enyi generalisations of the quantum conditional entropy and mutual information in terms of this divergence, from which they inherit many desirable properties. However, these quantities lack some of the convenient structure of their Shannon and von Neumann counterparts. We attempt to bridge this gap by establishing divergence inequalities for valid combinations of R\'enyi order which replicate the chain and decomposition rules of Shannon and von Neumann entropies. Although weaker in general, these inequalities recover equivalence when the R\'enyi parameters tend to one.
	
	To this end we present R\'enyi mutual information decomposition rules, a new approach to the R\'enyi conditional entropy tripartite chain rules and a more general bipartite comparison. The derivation of these results relies on a novel complex interpolation approach for general spaces of linear operators.
	
	 These new comparisons allow us to employ techniques that until now were only available for Shannon and von Neumann entropies. We can therefore directly apply them to the derivation of R\'enyi entropic uncertainty relations. Accordingly, we establish a family of R\'enyi information exclusion relations and provide further generalisations and improvements to this and other known relations, including the R\'enyi bipartite uncertainty relations.
	
\end{abstract}\newpage
\pagenumbering{roman}
\end{frontmatter}
\renewcommand{\arraystretch}{1.2}
\setlength{\tabcolsep}{2pt}
\setcounter{page}{1}
\null\thispagestyle{empty}
\chapter*{Introduction}
\pagenumbering{arabic}
\addcontentsline{toc}{chapter}{Introduction}
This thesis contributes some concrete results in the form of inequalities for quantum R\'enyi entropies and a direct application to determining R\'enyi entropy uncertainty relations. However, these results are a consequence of a more general desire to explore the mathematical structure and place of R\'enyi entropies in the broader theory of quantum information.
\paragraph{}
Entropy is a concept that was developed in the study of thermodynamics to quantify the disorder of a system. The statistical understanding of entropy was formalised by~\citet{boltzmann1872}, where it was viewed as a property of thermodynamic systems. The quantities we investigate stem from the information theoretic concept of entropy, characterised by the Shannon entropy $H(X)$ of a random variable~\citep{shannon48}. Although having other uses within information theory, entropy is prominently known as a measure of uncertainty (or spread) of a given random variable, quantifying the average information which would be gained (or uncertainty removed) from the observation of that variable.

Associated with the Shannon entropy of a classical random variable $X$ are several entropic quantities which each have their own uses and interpretations. The joint entropy $H(XY)$ describes the entropy of a multivariate random variable and, as for joint probability, is affected by whether or not the variables are dependent. Conditional entropy $H(X|Y)$ quantifies the average uncertainty of $X$ given that $Y$ has been observed. This quantity is well-known and understood, being a fundamental tool in cryptography, statistical analysis and as a general measure of uncertainty. Mutual information $I(X:Y)$ quantifies the level of correlation between $X$ and $Y$~--~the difference between the total of the individual entropies of $X$ and $Y$ and the joint entropy of $XY$. Most prominently,~\citet{shannon48} established that the capacity of any discrete memoryless communication channel is given by the maximal mutual information between the channel's input and output. Beyond its original use in information theory, it has found many other applications in information processing from such a wide range as machine learning~(see, e.g.,~\cite{hirche18, ML1, ML2}) and computational linguistics~(see, e.g.,~\cite{CL}).

The classical R\'enyi entropies $H_\alpha(X)$~\citep{renyi61} generalise the Shannon entropy. They form a family of entropies parametrised by their order, $\alpha>0$, that produce measurements of uncertainty that give more or less weight to events with high or low information content. More specifically, they weigh the \textit{surprisal} of outcomes differently depending on the order. We recover the Shannon entropy when $\alpha=1$ which, naturally, is one of the more useful orders of R\'enyi entropy, along with \textit{min-entropy}~--~when $\alpha = \infty$, and the \textit{collision entropy}~--~when $\alpha = 2$. Formal definitions are given in Section~\ref{sec:entropy}.

We consider also the quantum generalisation of R\'enyi entropies. It is prudent therefore to initially cover the quantum analogue of Shannon entropy, the von Neumann entropy $H(\rho)$. This entropy of a probability density matrix $\rho$ is well understood and induces quantum versions of the related entropic quantities that are well behaved and reflect our physical understanding of quantum systems.

The quantum conditional entropy has so far found uses in quantum information such as in the decoupling theorem~\citep{dupuis09} and more generally as an entanglement witness and as a measure of uncertainty~(see~\cite{Wat,wildebook13}). We also have a quantum generalisation of the mutual information which is compatible with von Neumann entropy and quantum conditional entropy. Quantum mutual information has analogous applications in quantum information, for example characterising the capacity of classical to quantum channels~\citep{holevo98,schumacher97,holevo73b} and the quantum channel capacity under entanglement assistance~\citep{Bennett,bennett02,bennett09,bertachristandl11}. It has also found applications in other areas of quantum physics, for example as an entanglement and correlation measure~(see, e.g.,~\cite{brandao13d}) and to quantify Heisenberg's uncertainty principle~(see~\cite{H95} and~\cite{coles17} for a review on related work).

From this point we will mostly consider quantum entropies and hence omit the `quantum' and rather specify that an entropy is classical if necessary.

Certain equivalences of entropic quantities are known as chain or decomposition rules. Aside from highlighting the relationships between the quantities, these equivalences can be considered as a method to define them. Indeed, given the relatively simple definitions of the von Neumann entropy and joint entropy we find the conditional entropy \emph{chain rule}, for $\rho_{AB}$ on the system $AB$ with marginal $\rho_B$ on $B$,
\begin{equation}
	H(A|B)_\rho = H(\rho_{AB}) - H(\rho_B).
\end{equation}
This concept can be extended to the tripartite conditional entropy chain rule: for a state $\rho_{ABC}$ on the system $ABC$,
\begin{equation}
	H(AB|C)_\rho = H(A|BC)_\rho + H(A|C)_\rho = H(B|AC)_\rho - H(B|C)_\rho.\label{eq:cdef-via-decomposition}
\end{equation}
Similarly, mutual information can be described as equivalent to an expression of conceptually simpler entropies, i.e. we define
\begin{equation}
	I(A:B)_\rho = H(\rho_A) + H(\rho_B) - H(\rho_{AB}) = H(\rho_A) - H(A|B)_\rho = H(\rho_B) - H(B|A)_\rho. \label{eq:def-via-decomposition}
\end{equation}
These equivalences are interpreted as \emph{decomposition rules}, expressing the mutual information in terms of the von Neumann entropies of the different marginals of the joint state on the systems $A$ and $B$. By appealing to the intuition that entropy measures uncertainty in a quantum system, they reveal that mutual information measures the uncertainty in $A$ that is due to the lack of knowledge of $B$, and vice versa.

Another important and equivalent pair of definitions for conditional entropy and mutual information is given in terms of quantum relative entropy~\citep{umegaki62}, namely as optimisations of the relative entropy between the joint state $\rho_{AB}$ and particular product states of the marginals:
\begin{align}
	H(A|B)_\rho &= -\min_{\sigma_B} D(\rho_{AB}\|\id_A\otimes \sigma_B) =  -D(\rho_{AB}\|\id_A\otimes \rho_B) \quad \text{and}\label{eq:cdef-via-divergence}\\
	I(A:B)_\rho &= \min_{\sigma_A,\sigma_B} D(\rho_{AB}\|\sigma_A\otimes \sigma_B) = \min_{\sigma_B} D(\rho_{AB}\|\rho_A\otimes \sigma_B) = D(\rho_{AB}\|\rho_A\otimes \rho_B) \,, \label{eq:def-via-divergence}
\end{align}
where in the expressions on the right we observed that the minima are achieved for the marginals $\rho_A$ and $\rho_B$ of the joint state $\rho_{AB}$.  These expressions reveal a fundamental property of these quantities that is not evident from the chain and decomposition rules, namely the \emph{data-processing inequality}. Specifically, this property entails that the conditional entropy is monotonically non-decreasing and the mutual information non-increasing under any local processing of information on $A$ and $B$. Its satisfaction directly follows from the monotonicity under quantum channels of the underlying relative entropy and the above equivalences. This property is crucial for many applications relating to the physical interpretations of these quantities, since it corresponds to our intuition that correlations cannot be created by acting on only one of the constituent parts (or by acting on them independently).

On the other hand, R\'enyi entropic quantities which maintain useful properties are not as straight-forward to derive and involve a significant increase in complexity. That said, a natural method of deriving suitable R\'enyi generalisations of the above quantities is to replace the relative entropy in~Eqs.~\eqref{eq:cdef-via-divergence} and~\eqref{eq:def-via-divergence} with a (quantum) R\'enyi relative entropy which naturally satisfies the data-processing inequality. It is important to note here that, in general, the equivalences in Eqs.~\eqref{eq:cdef-via-decomposition}-\eqref{eq:def-via-divergence} no longer hold in the case of R\'enyi entropic quantities defined via this method.

Various generalisations of the concepts of conditional entropy and mutual information to one-parameter families of operationally significant measures have been proposed. Possible candidates for a R\'enyi conditional entropy have been put forward, generalising both the quantum conditional entropy and various classical R\'enyi entropies~(see~\cite{tomamichel13, iwamoto13, teixeira12} for more detailed treatments). Similarly, for mutual information we have propositions in both the classical~(see, e.g.,~\cite{verdu15,tomamichel17c}, for recent discussions) and the quantum setting~(see, e.g.,~\cite{Hay}). We call such measures (quantum) R\'enyi conditional entropy or mutual information respectively, if they satisfy the data-processing inequality.

We are essentially looking for quantities which reproduce the relationships between the classical entropic quantities, have operational significance, reflect the physical situation and are mathematically convenient. Of these propositions for such definitions, quantities derived from the so-called `sandwiched' divergence~\citep{lennert13, wilde13} have in recent years surfaced as some of the most suitable. This parametrised quantum relative entropy, which we from this point refer to as simply `R\'enyi divergence', meets the above criteria and has found important applications within quantum information. This thesis essentially provides new results for R\'enyi divergence based quantum entropic quantities.

R\'enyi divergence generalises the quantum relative entropy from which we derive the quantum entropic quantities detailed above. It has many desirable properties as a measure of uncertainty and its structure is amenable to formulating inequalities. As an analogue to the quantum relative entropy, it induces definitions of the quantum R\'enyi conditional entropy~\citep{lennert13} and mutual information~\citep{Hay}. These definitions take the form
\begin{align}
	H_\alpha(A|B)_\rho := -\min_{\sigma_B} {D}_{\alpha}(\rho_{AB} \| \id_A \otimes \sigma_B)\quad\text{and}\quad
	I_{\alpha}(A\;;\>\!B)_\rho := \min_{\sigma_B} {D}_{\alpha}(\rho_{AB} \| \rho_A \otimes \sigma_B) \,,
	\label{eq:def-q}
\end{align}
where ${D}_{\alpha}(\cdot\|\cdot)$ denotes the R\'enyi divergence and $\alpha \in \left[\frac12,\infty\right)$. Note that, given this definition, we are lacking an obvious way to decompose the R\'enyi conditional entropy and mutual information into R\'enyi entropies of the state's marginals.

However, these quantities generalise the conditional entropy and mutual information respectively, which are recovered by setting $\alpha = 1$.  These quantities and similar constructions are widely used in the study of strong converses and the analysis of channel coding problems~(see, e.g.,~\citep{mosonyi11,mosonyiogawa13, mosonyi14-2, wilde13, cooney14, TWW}) and R\'enyi mutual information has found direct operational interpretation in classical and quantum hypothesis testing~\citep{Hay,tomamichel17c}.

Uncertainty relations are important in any field that involves signals and waves, but have found particular significance in quantum mechanics. These relations are often put in terms of commutators and standard deviations, for example the Robertson relation~(see Section~\ref{sec:ComDef}), which is a generalisation of the famous \textit{Heisenberg Uncertainty Principle}~\citep{heisenberg27}. Entropic uncertainty relations, on the other hand, take a more information theoretic approach, expressing uncertainty in terms of its entropy.

There are two canonical entropic uncertainty relations, the more well known being the Maassen-Uffink relation~\citep{maassen88} which indicates that the total uncertainty of measuring in two different bases is bounded below by a constant depending on the compatibility of the measurements. Even though it is usually expressed in terms of Shannon entropies, the general form of the theorem encompasses a family of relations in terms of classical R\'enyi entropy. There are some current well-known extensions and improvement of this relation which we cover in Section~\ref{sec:improvements-and-extensions}.

By extending the Shannon version of the Maassen-Uffink relation to conditional entropies we can then derive the Hall information exclusion principle~\citep{H95}. This result is an entropic uncertainty relation in terms of mutual information, describing the level of correlation between parts of a quantum system and some memory containing information about the system. This relation, dual to Heisenberg's uncertainty principle, gives upper bounds on the total amount of correlation between a state measured in either one of two incompatible bases and some classical memory with information about how the initial state was prepared, i.e.~if there is a high level of correlation between one measurement of the state and the memory, then there is a proportionally low level of correlation for the other measurement, depending on the compatibility of measurements. Hall's original relation has been improved upon, with bounds stronger than that found in the `Maassen-Uffink' type relations. 

The derivation of these relations and their improvements rely on the chain and decomposition rules of quantum entropies, motivating our investigation of feasible R\'enyi generalisations. We defer the details and formal statements to Section~\ref{sec:entropic-uncertainty-relations}.

The applications of entropic uncertainty relations and related concepts include quantum randomness, quantum cryptography and entanglement witnessing, see~\cite{coles17} for a more extensive review.

\paragraph{}
We now outline the core results of this thesis, i.e. establishing R\'enyi versions of the chain and decomposition rules in Eqs.~\eqref{eq:cdef-via-decomposition} and~\eqref{eq:def-via-decomposition}. So far the main example of such an extension is Dupuis' chain rules~\citep{Dup}, which generalise the tripartite chain rule to conditional entropies based on R\'enyi divergence.

In particular, we provide decomposition rules for R\'enyi mutual information which can then be applied to determine R\'enyi information exclusion relations. These rules take the form
\begin{align}
	I_{\alpha}(A\;;\>\!B)_\rho \geq H_{\beta}(\rho_A) - H_{\gamma}(A|B)_\rho  \quad \textrm{and} 	\quad
	I_{\alpha}(A\;;\>\!B)_\rho \leq H_{\bar\beta}(\rho_A) - H_{\bar\gamma}(A|B)_\rho
\end{align}
for suitable choices of R\'enyi orders $\beta, \bar\beta, \gamma$ and $\bar\gamma$. The formal result is presented in Theorem~\ref{decomp}. The two inequalities above reduce to the equality in Eq.~\eqref{eq:def-via-decomposition} when we take all the parameters to $1$.

We employ a novel interpolation approach which additionally yields a general bipartite divergence inequality and an alternative derivation of the tripartite chain rules. For completeness, we include a previously published proof of the decomposition rules in Appendix~\ref{sec:an-alternative-proof-of-the-decomposition-rule} which employs a similar method to that used by Dupuis in his demonstration of the chain rules. Interestingly, these different approaches produce slightly varied valid ranges for the R\'enyi parameters, the explanation of this discrepancy is left as an open question.

The principle mechanism in demonstrating the desired R\'enyi divergence inequalities is exploiting their connection to Schatten norms. This is done by performing complex interpolation on specific linear operator valued functions which yield our entropic quantities. The main machinery used to this end is an extension of Beigi's three-line theorem~\citep{beigi13} to spaces of general linear operators, which we detail in Chapter~\ref{sec:generalised-renyi-divergence-uncertainty-relations}.

The choices we make in these interpolations are not determined by mere luck or simply arrive out of the void, rather they are inspired and motivated by the norms on general $L_p$-spaces developed by~\citet{pisier1998}. These `Pisier norms' have already found utility in quantum information~(see~\cite{devetak06, delgosha13}) due to how they relate to R\'enyi divergence when constrained to finite spaces. In Chapter~\ref{sec:interpolation-framework} we discuss some observations of their properties in the context of R\'enyi entropies and how they influence the derivation of the divergence inequalities.

As a direct application we are able to establish R\'enyi generalisations of the Maassen-Uffink and Hall relations, and make further improvements on the bounds. Some uncertainty relations of this genre have already been treated, notably the bipartite uncertainty relation found in~\cite{mybook}, which provides a quantum R\'enyi extension to an already improved version of the Maassen-Uffink relation found in~\cite{Ber}.

We are able to improve the bipartite uncertainty relation by deriving state-dependent and state-independent bounds analogous to those determined in~\cite{coles14}. These new bounds are determined by the other R\'enyi orders present in the relation and we may recover known relations with particular choices of parameters. We are then able to establish a new family of information exclusion relations which provide a direct extension of the Hall relation to R\'enyi entropies. This family of relations can then be further improved, again with an analogous method to that used in~\cite{coles14}. Interestingly, the resulting tighter upper bounds are not order-dependent, unlike those for the improved R\'enyi bipartite uncertainty relations.
\paragraph{}
The structure of the thesis is as follows:

In Chapter~\ref{sec:BG} we cover the mathematical fundamentals, including a brief treatment of probability theory (mainly for the purposes of notation) and an overview of linear operators and Schatten norms. The remainder of this chapter is dedicated to the formal definitions of R\'enyi divergence and its related quantities, with a summary of some useful properties.

We explore Pisier norms and interpolation in Chapter~\ref{sec:interpolation-framework}. This involves introducing a family of super-operators and associated norms, which help inform some observations about certain specifications of Pisier norms. We also establish the three-line theorem for spaces of linear operators, integral to the proofs of the subsequent chapter.

Chapter~\ref{sec:renyi-entropy-divergence-inequalities} details the R\'enyi divergence inequalities and their proofs, including the generalised bipartite divergence inequalities, the decomposition rules, the tripartite chain rules and some other comparisons which directly follow. This chapter also includes some specifications of the new interpolation result which facilitate the derivation of the main inequalities.

The improved R\'enyi uncertainty relations are then established formally in Chapter~\ref{sec:generalised-renyi-divergence-uncertainty-relations}. We first cover the bipartite uncertainty relations, starting with a generalised version of the result from~\cite{mybook}, then move to a version with an improved, state-independent bound. We conclude the main results with a R\'enyi information exclusion relation and its improved version.

We then provide a brief discussion, detailing some possible future work, such as further generalisations and applications. We also examine how these results may fit into the broader theory, notably how it could aid in establishing generalisations of the chain and decomposition rules for smooth R\'enyi entropies~\citep{renner05b,mythesis} and how it could inform a clearer definition of the R\'enyi conditional mutual information~(see~\cite{bertawilde14}), both currently open questions in quantum information.
\chapter*{Notation and nomenclature}
\addcontentsline{toc}{chapter}{Notation and nomenclature}
We use standard notation as summarised in the following table:
\begin{table}[h]\caption{Overview of notation}
\begin{center}
	\begin{tabular}{|c|p{12.4cm}|}
		\hline
		\thead{Symbol}                        & \thead{Meaning}                                                                                                                                                                      \\
		\hhline{|=|=|}
		\makecell[t]{$\supp(X)$\\ $[\supp(\rho)]$}&The support of a random variable [probability density matrix], i.e. the values of $x$ such that $P(X=x)\neq 0$ [the vectors such that $\rho\ket{\psi} \neq 0$].\\
		\hline 
		\makecell{$|X|$\\$[|\supp(\rho)|]$}&\makecell[l]{ For a random variable $X$, the cardinality of the support of $P(X = x)$. \\
			{[For a probability density matrix $\rho$, the dimension of $\supp(\rho)$.]}}\\
		\hline
		$\langle X \rangle$ & The expected value of $X$, $\sum_x xP(X= x)$.
		\\\hline
		$\log$                        & The logarithm to base 2, we also take the convention that $0\log0 = 0$.                                                                                                                                                      \\
		\hline
		$A,B, AB$                       & Quantum system or subsystems                                                                                                                                                 \\
		\hline
		\makecell[t]{$\Hil_A, \Hil_B$\\$[\Hil_{AB}]$} &The Hilbert spaces of states corresponding to the quantum systems $A$ and $B$. [$\Hil_A \otimes \Hil_B$]                                                                                                   \\
		
		\hline
		$\Lin(A, B)$ $[\Lin(A)]$& Set of linear operators from $\Hil_A$ to $\Hil_B$ [$\Hil_A$].                                                                                                                                                                                                           \\
		\hline
		$\Pos(A)$ $[\Pos^*(A)]$ & The set of positive semi-definite [strictly positive] operators from $\Hil_A$ to $\Hil_A$.                
		\\
		\hline
		$\Den(A)$ $[\Den^*(A)]$& The subset of $\Pos(A)$ $[\Pos^*(A)]$ of operators with unit trace.\\
		\hline
		$\Uni(A)$ $[\Uni(A,B)]$&The set of unitary operators in $\Lin(A)$ [isometries in $\Lin(A,B)$].\\\hline
		$\CPTP(A,B)$                  & The set of completely-positive trace-preserving operator maps from $\Lin(A)$ to $\Lin(B)$.                                                                                    \\
		\hline
		$\tr_A(\cdot)$                & The partial trace over $A$, $\tr_A(K_A \otimes L_B) = \tr(K_A)L_B$.                                                                                              \\
		\hline
		$\sigma\gg\rho$               & $\sigma$ `dominates' $\rho$, i.e the kernel of $\sigma$ is contained in the kernel of $\rho$.                                                                                \\
		\hline
		\makecell[t]{$\sigma\perp\rho$\\$[\sigma\not\perp\rho]$}& $\sigma$ and $\rho$ are [not] perpendicular, i.e. the span of $\sigma$ and $\rho$ have [non-]empty intersection. \\
		\hline                                                                      
		
		$\id_A\in \Lin(A)$            & The identity map on $\Hil_A$.                                                                                                                                                 \\
		
		\hline
		 $\Op_{A\rightarrow B}(\cdot)$ & \makecell[tl]{The operator mapping from $A$ to $B$. For basis vectors $\ket{e_i}\in \Hil_A, \ket{f_j} \in \Hil_B$,\\$\Op_{A\rightarrow B}:\Hil_{AB} \rightarrow \Lin(AB),\ket{e_i}\otimes \ket{f_j}\mapsto \ketbra{f_j}{e_i}$.} \\
		
		\hline
		$\|\cdot\|_p$&The $p$-Schatten norm, $\|X\|_p=\tr\left[(X^\dagger X)^\frac{p}{2}\right]^\frac{1}{p}$. This is not a norm for $p<1$.\\$[\|\cdot\|^*_p]$&[The dual $p$-norm - the norm induced on the space of functionals on the space equipped with $\|\cdot\|_p$.]                                 \\
		\hline $\Re(z)$& The real part of the complex number $z$.\\
		\hline
		$\lambda_{\max}\left[M\right]$ & The maximum eigenvalue of $M$.\\
		\hline
		$[X,Y]$ & The commutator: for $X,Y\in \Lin(A)$, $[X,Y] = XY - YX$.\\\hline
\end{tabular}\end{center}\end{table}
\chapter{Background and current literature}\label{sec:BG}
\section{Probability theory}
Before we discuss entropy, we must establish some definitions and notations for probability.

\subsection{Events and outcomes}
Here we consider only discrete random variables and, accordingly, a set of outcomes is a discrete set which describes the observable results of an experiment. An event, on the other hand, describes a subset of possible outcomes. For example, when rolling a standard die, the numbers 1 through 6 are possible outcomes whereas the outcome being an odd number is an event which includes the outcomes 1, 3 and 5.

Given a set of outcomes $\Sigma$, we consider the probability of an event $E\subseteq\Sigma$ as a subadditive function from the power set of $\Sigma$ to the real unit interval, $P:2^\Sigma\rightarrow [0,1]$. $P(E) = 1$ means that $E$ is guaranteed to happen and $P(E) = 0$ means $E$ will not happen.

Mutually exclusive events are events that cannot occur at the same time. Heads or tails being the result of a coin toss are mutually exclusive whereas it raining on the weekend or it reaching $20^\circ$ are not mutually exclusive.
We therefore expect that, for mutually exclusive events $E$ and $F$, if $P(E) = p$ then
\begin{equation}
	P(F) \leq 1-p \implies P(E)+P(F) \leq 1.
\end{equation}
A complete set of mutually exclusive events encompasses all possible outcomes of an experiment. If $\{E_i\}$ is a complete set of mutually exclusive events we have $\bigcap_i E_i = \emptyset$ and $\bigcup_i E_i = \Sigma$. Moreover,
\begin{align}
	 P(E_k)            = 1-\sum_{i\neq k}P(E_i) 
	\implies \sum_{i}P(E_i)  = 1.\label{comp}
\end{align}
Eq.~\eqref{comp} is known as the completeness relation.

\subsection{Random variables}
A random variable $X$ associates the outcomes of an experiment with real values, so we may consider the set of outcomes $\Sigma\subseteq \R$. Accordingly, we have the probability of an outcome $x\in\Sigma$,  $P(X=x)$ and the probability of an event, $X\in E\subseteq \Sigma$, written $P(X\in E) = \sum_{x\in E} P(X=x)$. Clearly, if $x_1\neq x_2, \, X=x_1 \implies X\neq x_2$, hence the set $\{x_i\}_i$ is complete and mutually exclusive, i.e.
$
	\sum_i P(X=x_i) = 1
$.
A discrete random variable $X$ is described by its probability mass function $p(x) = P(X=x)$.

An important property of a random variable is its expected value $\langle X\rangle$. This is a measure of centre which plays a significant role in the analysis of random variables. We define the expected value as the sum of the outcome values, weighted by their probabilities:
$
	\langle X\rangle= \sum_{x\in \Sigma} xp(x)
$.
Note that in general this does not reproduce the most likely value, but rather characterises the average of the values generated by many independent observations of identical random variables. 
\subsection{Probability and related quantities}
The \textit{joint probability} of two events $E$ and $F$ is written $P(E,F)$, read ``the probability of $E$ and $F$''. It describes the chance that both $E$ and $F$ will occur simultaneously. In terms of random variables we write
$
	P(X=x, Y=y) = p(xy)
$.

The \textit{conditional probability} $P(E|F)$ is similar but instead considers the chance that $E$ occurs given that $F$ occurs. For random variables we write
$
	P(X=x|Y=y) = p(x|y)
$.

For independent random variables we have
$
	p(xy) = p(x)p(y)
$,
otherwise
\begin{equation}\label{Bayes}
	p(xy) = p(x)p(y|x) = p(y)p(x|y)\implies p(x) = \frac{p(xy)}{p(y|x)}, \quad p(x|y) = \frac{p(y|x)p(x)}{p(y)}.
\end{equation}
The equation on the far right is known as Bayes' Theorem.

This notation will be used for the remainder of this thesis, except for some cases where it may cause ambiguity.

\section{Banach spaces of operators}
The main accepted mathematical model of quantum mechanics and information is based on complex-valued linear algebra and the analysis of vector spaces. In this section we cover some of the theory and tools available to us due to this model. In particular, we detail some results which are useful in the following sections.
\subsection{Matrices as operators}
An operator, in general, is a mapping that takes an element of one space to the element of another space (possibly the same space). In our context of linear operators acting on finite-dimensional Hilbert spaces, we may represent operators with matrices and we use the terms interchangeably when the focus dictates it.
\subsubsection{Transformations on operators}
For reference we cover some basic manipulations of operators.

\begin{itemize}
\item The \textit{inverse} of a square matrix can be considered as the operator which reverses the action of the original. Not all matrices are invertible, these are known as \textit{singular} or \textit{degenerate} matrices. Given invertible $M\in \Lin(A)$ its inverse, $M^\textsc{I}$ is such that if, for $v\in \Hil_A$, 
\begin{equation}
	Mv = u \implies   v = M^\textsc{I}u = M^\textsc{I}Mv.
\end{equation}
As the above indicates, the inverse relationship is symmetric and $MM^\textsc{I} = M^\textsc{I}M = \id_A$.
 
\item The \textit{transpose} is essentially a reflection of the matrix across its diagonal. Given $M = (a_{ij})$, ${M^\top = (a_{ji})}$. Many properties and attributes are preserved under the transpose, the trace for example. If $M\in \Lin(A,B)$ then $M^\top\in \Lin(B, A)$. 

\item The \textit{conjugate} $\overline{M}$ of complex-valued matrix $M$ is just the matrix whose entries are the complex conjugates of the entries of $M$.

\item The \textit{adjoint} $M^\dagger$ of an operator is the operator such that for the inner product of the Hilbert space it is acting upon
$
	\langle u, Mv\rangle = \langle M^\dagger u, v\rangle
$.
In our context, this simplifies to the conjugate transpose of $M$, i.e. $M^\dagger = \overline{M}^\top$.
\end{itemize}
\subsubsection{Types of operators}
We have particular categories of operators each with specific properties. We work our way from the most general to the least.
\begin{itemize}
\item A \textit{square} matrix is any matrix that has the same number of rows as columns. As an operator it maps between spaces of the same dimension~--~any operator that maps to the same space must at least be square.

\item A \textit{diagonalisable} matrix is a square matrix for which there exists a basis such that it can be written as a diagonal matrix. The entries of this diagonal matrix are the eigenvalues of the original.

\item A \textit{normal} matrix is such that $M^\dagger M = M M^\dagger$. As a result all normal matrices map to and from the same space.

\item A \textit{self-adjoint} (Hermitian) matrix is such that $M^\dagger = M$. A self-adjoint matrix is thus normal. The eigenvalues of a self-adjoint matrix are real and its eigenvectors form an orthonormal basis.

\item A \textit{positive semi-definite} matrix is a self-adjoint matrix whose eigenvalues are greater than or equal to zero. A \textit{strictly positive} or positive definite matrix is similar except that its eigenvalues are all greater than zero. We often shorten positive semi-definite to just \textit{positive}. If $X\in \Pos(A)$ is positive then there exists an $M\in \Lin(A, B)$ such that $X = M^\dagger M$ and a self-adjoint $Y\in \Lin(A)$ such that $X = Y^2$.

\item A \textit{density} matrix is a trace-1 positive matrix.

\item An \textit{isometry} $U$ is an operator such that $U^\dagger U = \id$. The length of a vector is invariant under isometric operations.

\item A \textit{unitary} operator $U$ is an isometry which is also square. Any evolution of a quantum state can be represented as a unitary operator since it takes unit vectors to unit vectors of the same dimension.
\end{itemize}
\subsubsection{Functions on operators}
Given a function from $\C \rightarrow \C$, we can define an extension of this function on any normal matrix, depending on its domain.

For a subset $\Omega\subseteq \C$ and $f:\Omega\rightarrow \C$ we define $f:\Omega^{m\times m} \rightarrow \C^{m\times m}$ where, if the eigenvalues of $M$ are $\{\lambda_i\}_i\subseteq \Omega$, the eigenvalues of $f(M)$ are $\{f(\lambda_i)\}_i$.

Most functions are well defined for normal matrices since we are guaranteed a unitarily similar diagonalisation~(see Eq.~\eqref{SDN}). However, sometimes we must restrict our focus for functions with limited natural domains. For example, the logarithm can only be applied to positive matrices since their eigenvalues are all real and non-negative.

We often use the function $f(M) = M^{-1}$ to denote the pseudoinverse \cite[Section~5.5.4]{golub2013matrix}. Usually this notation refers to the usual matrix inverse but we require a more general definition. In our context we take advantage of two main properties of the pseudoinverse. Firstly, for matrices with full support, the pseudoinverse coincides with the usual matrix inverse. However, if $M\in \Lin(A)$ does not have full support (and is therefore singular), $M^{-1}M$ is instead a projector onto $\supp(M)\subset\Hil_A$. This is problematic as we cannot guarantee a given density operator has full support. This brings us to the second useful property: there are situations where we maintain the usual behaviour of the inverse in the trace, i.e. if $\supp(M)\subset \Hil_A$ then for $X\in \Lin(A)$ with $\supp(X) \subseteq \supp(M)$ we have
$
	\tr M^{-1}MX  = \tr X
$.
In other words, equivalence holds if no information is lost by the action of the projector $M^{-1}M$.
\subsection{Operator decompositions}
As an aid to computation we often make use of operator decompositions. These observations essentially allow us to decompose certain classes of operators into combinations of simpler operators which are more mathematically convenient.
\subsubsection{The spectral decomposition}\label{sec:SD}
Also known as the `eigendecomposition', the spectral decomposition allows us to reformulate a given matrix in terms of its eigenvalues and eigenvectors:
\begin{proposition}
	Any diagonalisable matrix $A$ can be factorised as
	\begin{equation}
		A = P\Lambda P^{-1},
	\end{equation}
where $\Lambda$ is a diagonal matrix whose $i$th entry is the $i$th eigenvalue, $\lambda_i$, of $A$ and the $i$th column of $P$ is the eigenvector corresponding to $\lambda_i$.
For normal matrices this reduces to 
\begin{equation}
	X = U\Lambda U^\dagger,\label{SDN}
\end{equation}
where $U$ is unitary.
\end{proposition}
\subsubsection{The polar decomposition}
The polar decomposition allows us to express any square matrix in a much more convenient form, especially in the context of Schatten norms~(see the next section).
\begin{proposition}
	Any operator $M\in \Lin(A)$ can be decomposed into the product of a positive operator and a unitary operator i.e. for $P,Q \in \Pos(A)$ and $U\in \Uni(A)$
	\begin{equation}
		M = UP = QU.
	\end{equation}
The above are known respectively as the left and right polar decompositions.
\end{proposition}
\subsubsection{The singular value decomposition}
We may also decompose non-square matrices in a similar way to the spectral decomposition.
\begin{proposition}\label{SVD}
	For any $M\in\Lin(A,B)$ we may write
	\begin{equation}
		M = UDV
	\end{equation}
	where $U\in \Uni(C,B), V\in \Uni(A,C)$ are isometries and $D$ is a positive semi-definite diagonal matrix in $\Lin(C)$ with dimension equal to the number of non-zero eigenvalues of $M^\dagger M$.
\end{proposition}
The diagonal entries of $D$ are known as the \textit{singular values} of $M$ which we denote $\{s_i(M)\}_i$.

Most expositions of this proposition instead present the decomposition in terms of unitary $U$ and $V$ in $\Uni(B)$ and $\Uni(A)$, respectively, and $D$ as a rectangular diagonal matrix in $\Lin(A,B)$. However, for our purposes, the above form is more convenient (See, for example, Theorem~\ref{weighted3lines}). Hence we include a brief proof for clarity.

\begin{proof}[Proof of Proposition \ref{SVD}]
	Let $M\in\Lin(A,B)$ be an $n\times m$ matrix. We have that $M^\dagger M$ is positive semi-definite so by Eq.~\eqref{SDN} there exists a $\tilde V^\dagger\in \Uni(A)$ such that
	\begin{equation}
		\tilde V^\dagger M^\dagger M \tilde V = \tilde D = \begin{bmatrix}
		D^2&0\\0&0	
		\end{bmatrix}
	\end{equation}
where $D^2\in \Lin(C)$ is diagonal and positive semi-definite with $\dim(\Hil_C)$ equal to the number of non-zero eigenvalues of $M^\dagger M$. Again by Eq.~\eqref{SDN}, we may write 
\begin{equation}
	\tilde V^\dagger=\begin{bmatrix}V^\dagger&V_0^\dagger\end{bmatrix}
\end{equation}
where $V^\dagger$ and $V_0^\dagger$ are matrices whose columns are the eigenvectors of the, respectively, non-zero and zero eigenvalues.
In this sense we can write
\begin{equation}
	\begin{bmatrix}
		D^2&0\\0&0
	\end{bmatrix} = 
	\begin{bmatrix}
		V M^\dagger M V^\dagger &V M^\dagger M V_0^\dagger\\
		V_0 M^\dagger M V^\dagger & V_0 M^\dagger M V_0^\dagger
	\end{bmatrix},
\end{equation}
hence $D^2 = V M^\dagger M V^\dagger$ with $V^\dagger\in \Uni(C, A)$.

Define $U = MV^\dagger D^{-1}\in \Lin(C, B)$, thus
\begin{equation}
	UDV = MV^\dagger D^{-1}DV = MV^\dagger V = M.
\end{equation}
Moreover,
\begin{equation}
	U^\dagger U = (MV^\dagger D^{-1})^\dagger MV^\dagger D^{-1} = D^{-1}V M^\dagger MV^\dagger D^{-1} = D^{-1}D^2D^{-1} = \id_B,
\end{equation}
Hence $U\in \Uni(C,B)$.
\end{proof}

\subsubsection{The Schmidt decomposition}
Finally, we have the Schmidt decomposition, considered an extension of the singular value decomposition to vectors, which allows us to represent any pure state as a set of orthonormal vectors localised on distinct subsystems.
\begin{theorem}\label{SchmidtD}
	We may write any vector $v \in \Hil_{AB}$ in the form
	\begin{equation}
		v = \sum_ i r_i \left(e^i_A\otimes f^i_B\right)
	\end{equation}
	where $\{e^i_A \}_i$ and $\{f^i_B \}_i$ form orthonormal bases on their respective subsystems and the Schmidt coefficients $\{r_i\}_i$ are real and non-negative. \end{theorem}For proofs and related material see~\cite[Section~2.5]{NC}.

\subsection{The Schatten operator norm}\label{sec:the-schatten-operator-norm}
If we introduce a norm on a set of a finite complex-valued linear operators we form a Banach space. The Schatten norms form a family which generalise some of the more common norms associated with spaces of operators. Essentially, they can be considered as a matrix generalisation of the $p$-norms on Lebesgue spaces.

For $M\in \Lin(A,B)$, we define
\begin{equation}
	\left\|M\right\|_p = \left[\tr\left(M^\dagger M\right)^\frac{p}{2}\right]^\frac{1}{p},\quad p\geq 1.
\end{equation}
Note these norms are well defined for rectangular matrices since $M^\dagger M$ is by definition self-adjoint. In fact, the eigenvalues of $|M| = \left(M^\dagger M\right)^\frac{1}{2}$ coincide with the singular values of $M$ (and $M^\dagger$)~(see \cite[Section~1.1.3]{Wat}). As such, we may write the above in the form
\begin{equation}\label{SVnorm}
	\left\|M\right\|_p = \left[\sum_{i=1}^k s_i\left(M\right)^p\right]^\frac{1}{p},
\end{equation}
where $k$ is the rank of $M$.

\subsubsection{Properties}
The above definition induces some useful properties:
\begin{proposition}
	The Schatten norm is unitarily invariant, i.e. for $M\in \Lin(A, B)$, $U\in \Uni(B)$ and $V\in \Uni(A)$,
	\begin{equation}\label{UIV}
		\|UMV\|_p = \|UM\|_p = \|MV\|_p = \|M\|_p.
	\end{equation}
\end{proposition}
\begin{proof}
	By the singular value decomposition we may write
$
		UMV = UU'DV'V
$,
where $D\in \Lin(C)$ is a diagonal matrix, $UU'\in \Uni(C, B)$ and $V'V\in \Uni(A,C)$. This demonstrates that the singular values of $M$ and $UVM$ coincide, hence by Eq.~\eqref{SVnorm} we have the equivalences in Eq.~\eqref{UIV}.
\end{proof}
\begin{proposition}
	For $M\in \Lin(A, B)$,
	\begin{equation}
		\left\|M\right\|_p = \left\|M^\dagger\right\|_p =\left\|\overline{M}\right\|_p = \left\|M^\top\right\|_p = \left\||M|\right\|_p.
	\end{equation}
\end{proposition}
\begin{proof}
	We simply observe that the singular values of $M$ are invariant under these transformations.
\end{proof}
\begin{proposition}\label{fold}
	Given $M\in \Lin(A, B)$ we may write
$
		\left\|M^\dagger M\right\|_{p} = \left\|M\right\|^2_{2p}
$.
\end{proposition}
\begin{proof}
	\begin{align}
		\left\|M^\dagger M\right\|_{p}  &= \left[\tr\left(M^\dagger M M^\dagger M\right)^\frac{p}{2}\right]^\frac{1}{p}\\
		&=\left[\tr\left(M^\dagger M \right)^\frac{2p}{2}\right]^\frac{2}{2p}\\
		&=\left\|M\right\|^2_{2p}.
	\end{align}
\end{proof}
We conclude with an observation about the form of positive operators in the unit ball of a given norm:
\begin{proposition}
	For all $\sigma\in \Den(A)$
$
		\left\|\sigma^\frac{1}{p}\right\|_p = 1
$.
\end{proposition}
\begin{proof}
	This is evident by considering the definition of the Schatten norm and the fact that ${\tr(\sigma) = 1}$.
\end{proof}
\subsubsection{Examples}
The Schatten norms generalise some well known norms on spaces of operators:
\begin{itemize}
	\item The \textit{trace} norm is equivalent to $\|\cdot\|_1$ and recovers the trace for self-adjoint operators.
	
	\item The \textit{Hilbert-Schmidt (or Frobenius)} norm is equivalent to $\|\cdot\|_2$. This norm is induced by the inner product $\langle M, N\rangle = \tr(M^\dagger N)$, hence $\Lin(A, B)$ equipped with this norm is a Hilbert space.
	
	\item The \textit{operator} norm is equivalent to $\|\cdot\|_\infty$, defined as
	\begin{align}
		\|M\|_\infty &= \sup_{v\in \Hil_A} \frac{\|Mv\|_2}{\|v\|_2}&&\left[\|v\|_2 = \sqrt{\langle v, v\rangle}\right].
	\end{align}
\end{itemize}
\section{Quantum mechanics and information}\label{sec:quantum-mechanics-and-information}
The following follows the material presented in~\cite{NC}.

\subsection{The state of a quantum system}
A quantum system can be any physical system, but generally we only consider systems where quantum effects are significant, such as an electron, a molecule or the crystalline structure of a material like doped silicon. A quantum system $A$ is associated with a separable\footnote{Having a countable orthonormal basis.} Hilbert space $\Hil_A$ over $\C$. Here, we consider only finite Hilbert spaces equipped with an inner product \begin{alignat}{4}
	\langle \cdot,\cdot\rangle&:&\;\Hil_A\times \Hil_A&&\; &\longrightarrow &\;&\C\\
	&:&\;\langle u,v\rangle &&\;&\longmapsto&\;&\sum_{i} \overline{u}_i v_i,\label{IP}
\end{alignat}
where $u_i$ and $v_i$ are respectively the $i$th elements of $u$ and $v$. The possible states of the system $A$ are described by the elements of $\Hil_A$.

We use bra-ket notation to represent quantum states: for $\psi\in\Hil_A$, its `ket' is written $\ket{\psi}$ and it can be considered as a column vector such that for an operator $M\in \Lin(A,B)$,
\begin{equation}
	M\ket{\psi} = M\psi.
\end{equation}
From this point we use the shorthand $\ket{\psi} = \psi \in \Hil_A$.
Considering $\ket{\psi}$ as a matrix, we define the `bra' of $\psi$ as
$
	\bra{\psi} = \ket{\psi}^\dagger
$.

We find this definition notationally convenient as, considering Eq.~\eqref{IP}, we may write
$
{\braket{\varphi}{\psi} = \langle\varphi,\psi\rangle}
$
and
\begin{equation}
	\langle \varphi, A\psi\rangle = \braket{\varphi}{A\psi} = \braket{A^\dagger\varphi }{\psi} = \bra{\varphi}A\ket{\psi}.
\end{equation}
\subsection{State vectors and probability density matrices}

A finite dimensional state vector can be represented by a complex-valued ket $\ket{\psi}$, for which
$
	\braket{\psi}{\psi} = 1.
$
This is to say that in an orthonormal basis $\{\ket{e_i}\}_i$ where
$
	\ket{\psi} = \sum_i a_i\ket{e_i}
$,
the square sum of the moduli of the complex coefficients,
$
	\sum_i |a_i|^2  = 1
$.

The probability density matrix (or probability operator) $\rho$ of a quantum system is determined by an ensemble of pure states, $\{p_i, \ket{\psi_i}\}_i$. We define
\begin{equation}
	\rho = \sum_i p_i\ket{\psi_i}\bra{\psi_i},
\end{equation}
where $\rho$ is a density operator and as such we may express its spectral decomposition~(see section \ref{sec:SD}) in terms of bra-ket notation:
\begin{equation}\label{SpD}
	\rho = \sum_j \lambda_j\ket{j}\bra{j},
\end{equation}
where the eigenvectors form an orthonormal eigenbasis and $\sum_j\lambda_j = 1$.

A system is considered to be in a pure state when its probability density matrix can be represented by a single state vector, i.e. $\rho = \ket{\psi}\bra{\psi}$. This reflects that we have complete knowledge of the state of the system. Note we must still consider superposition and the underlying uncertainty when measuring a pure state in a particular basis. A system that is not in a pure state is said to be in a mixed state.

If $\rho$ has the spectral decomposition as in Eq.~\eqref{SpD} then
\begin{align}
	\rho^2               & = \sum_{i,j} \lambda_i\lambda_j \ket{i}\braket{i}{j}\bra{j} \\
	                     & =\sum_j \lambda_j^2\ket{j}\bra{j}                           \\
	\implies \tr(\rho^2) & = \sum \lambda_j^2.
\end{align}
For all $j$, $0\leq\lambda_j\leq1$, so $\lambda_j^2\leq \lambda_j$, with equality if and only if $\lambda_j$ is the only eigenvalue, i.e. equality if and only if $\rho$ is a pure state.

This gives us the criterion:
\begin{equation}\label{purity}
	\tr(\rho^2)\leq 1,
\end{equation}
with equality if and only if $\rho$ is a pure state.
\subsection{Composite systems}
A composite quantum system is one that is made up of distinct subsystems. We write that $\rho_{AB} \in \Den(AB)$ is a probability density matrix of the composite system of $A$ and $B$. To recover the state on one of the subsystems we `trace out' the other subsystem with the \textit{partial trace}: given a basis $\{\ket{e_i}\}_i \subset \Hil_B$ and $M\in \Lin(AB)$, we define the marginal of $M$ in $A$:
\begin{equation}
	\tr_B(M) = \sum_i\bra{e_i} M\ket{e_i} = M_A \in \Lin(A).
\end{equation}

Given $\rho_{AB}\in \Den(AB)$, the \textit{reduced state} $\rho_A\in \Den(A)$ is the marginal of $\rho_{AB}$ in $A$, i.e.
$
	\tr_B(\rho_{AB}) =\rho_A
$.
We can consider the reduced state as a description of the state on a subsystem, at least when it comes to making measurements only on this subsystem~\citep{NC}. The state of a composite system may or may not be entangled. A pure state on a composite system can be decomposed as $\rho_{AB} = \rho_A\otimes\rho_B$ if and only if no entanglement exists between the subsystems.

\subsubsection{Classical-quantum states}\label{sec:classical-quantum-states}
It is possible to define joint classical-quantum systems, required in the case of measurement on a single subsystem~(see next section). In order to consider classical subsystems in the context of density operators we introduce an auxiliary Hilbert space $\Hil_X$ for the random variable $X$, with orthonormal basis $\{\ket{x}\}_x$ that acts as a `classical register'~\citep{coles17}. Accordingly, we have
\begin{equation}
	\rho_{XB} = \sum_x p(x)\ketbra{x}{x}\otimes \rho^x_B,
\end{equation}
where $\rho^x_B$ is the quantum state of the system $B$ conditioned on $X=x$, i.e. we have
\begin{equation}
	\tr_X\rho_{XB} = \sum_x p(x)\rho^x_B = \rho_B.
\end{equation}
\subsubsection{Observations about pure states}
Given a pure state is represented by a vector in its respective Hilbert space, we can use the Schmidt decomposition to derive some useful equivalences:

\begin{proposition}
	Given the pure state $\rho_{AB} = \ketbra{\varphi}{\varphi}\in\Den(AB)$ and the Schmidt decomposition ${\ket{\varphi} = \sum_i r_i \ket{i}_A\otimes \ket{i}_B}$, if we let $X\in \Lin(A,B)$ such that $X = \sum_i r_i \ket{i}_B \bra{i}_A$ then
	\begin{equation}
		X^\dagger X = \rho_A\quad\text{and}\quad X X^\dagger = \rho_B.
	\end{equation}
\end{proposition}
\begin{proof}
	We prove the first statement explicitly, then the second follows from a symmetric argument.
	
	First note that since $\{\ket{i}_A \}_i$ is an orthonormal basis for $\Hil_{A}$ we may write the partial trace over $A$ as
	\begin{align}
		\rho_B &= \tr_A(\rho_{AB})\\ &= \sum_k \bra{k}_A\sum_{i,j}r_ir_j \left(\ket{i}_A\otimes \ket{i}_B\right) \left(\bra{j}_A\otimes \bra{j}_B\right)\ket{k}_A\\
		&=\sum_{i,j,k}r_ir_j \left(\braket{k}{i}_A\otimes \ket{i}_B\right) \left(\braket{j}{k}_A\otimes \bra{j}_B\right)\\
		&=\sum_i r_i^2 \ket{i}_B\bra{i}_B.
	\end{align}
	Moreover,
	\begin{align}
		XX^\dagger &= \sum_{i,j}r_ir_j \ket{i}_B \braket{i}{j}_A \bra{j}_B\\
		&= \sum_{i}r_i^2 \ket{i}_B\bra{i}_B.
	\end{align}
\end{proof}

\begin{proposition}\label{traceFree}
	Given a pure state $\ket{\varphi} \in \Hil_{ABC}$ with Schmidt decompositions
	\begin{equation}
		\ket{\varphi} = \sum_i r_i \ket{i}_{A}\otimes \ket{i}_{BC} = \sum_i s_i \ket{i}_{AB}\otimes \ket{i}_{C},
	\end{equation}
	we have, for $K_A \in \Lin(A)$ and $L_C \in \Lin(C)$,
	\begin{equation}
		\left\|L_C \sum_ir_i \ket{i}_{BC}\bra{i}_A K_A\right\|_2 = \left\|L_C \sum_is_i \ket{i}_{C}\bra{i}_{AB} K_A\right\|_2.
	\end{equation}
\end{proposition}
This implies that we are free to choose any valid Schmidt decomposition in the product of an argument of the Schatten $2$-norm~--~effectively `shifting' a subsystem to the other side~--~given that subsystem is not otherwise represented in the product.
\begin{proof}
	We first rewrite the left-hand side as a trace:
	\begin{align}
		\left\|L_C \sum_ir_i \ket{i}_{BC}\bra{i}_A K_A\right\|_2^2&=\tr\left(K_A^\dagger \sum_ir_i \ket{i}_{A}\bra{i}_{BC}L_C^\dagger L_C \sum_jr_j \ket{j}_{BC}\bra{j}_A K_A \right)\\
		&=\tr\left(\sum_{i, j}r_ir_jK_A^\dagger  \ket{i}_{A}\bra{i}_{BC}L_C^\dagger L_C  \ket{j}_{BC}\bra{j}_A K_A \right)\\
		&=\sum_{i, j,k}\bra{k}_Ar_ir_jK_A^\dagger  \ket{i}_{A}\bra{i}_{BC}L_C^\dagger L_C  \ket{j}_{BC}\bra{j}_A K_A\ket{k}_A\\
		&=\sum_{i, j,k}r_ir_j\bra{i}_AK_A  \ket{k}_{A}\bra{k}_A K_A^\dagger\ket{j}_A\bra{i}_{BC}L_C^\dagger L_C  \ket{j}_{BC}\\
		&=\sum_{i, j}r_ir_j\bra{i}_A\otimes\bra{i}_{BC}\left(K_AK_A^\dagger\otimes L_C^\dagger L_C\right) \ket{j}_A \otimes\ket{j}_{BC}\\
		&=\bra{\varphi}\left(K_AK_A^\dagger\otimes L_C^\dagger L_C\right)\ket{\varphi}\\
		&= \sum_{i, j}s_is_j\bra{i}_{AB}\otimes\bra{i}_{C}\left(K_AK_A^\dagger\otimes L_C^\dagger L_C\right) \ket{j}_{AB} \otimes\ket{j}_{C}.
	\end{align}
	From this point, reversing the process obtains the right-hand side.
\end{proof}

\subsection{Measurement}\label{sec:measurement}
Uncertainty relations are a characterisation of how much information we can extract from a system by measurement, hence we require a formal understanding of what measurement means in the context of quantum information theory.

Although an interesting avenue for further research, this thesis does not treat general projector operator-valued measurements~--~or POVMs, instead we cover the simpler situation of measurements in orthonormal bases, or ONBs for short.

An ONB $\mathbb{X}$ can be represented by a set of rank 1 projectors:
$
\mathbb{X} = \{\ket{x}\bra{x}\}
$,
whose spans are mutually orthogonal.

In a real physical situation, if a state is measured in this basis the system will collapse to one of the pure states described by the basis vectors. The probability of observing the system in that basis state can be determined by acting on the state with its associated measurement map:
\begin{alignat}{4}
	\mathcal{M}_X( \cdot)&:&\;\Lin(A)&&\; &\longrightarrow &\;&\Lin(X)\\
	&:&\;\rho_A &&\;&\longmapsto&\;&\sum_{x}\ketbra{x}{x}\rho_A\ketbra{x}{x} = \sum_{x}\bra{x}\rho_A\ket{x}\ketbra{x}{x},
\end{alignat}
$\mathcal{M}_X\in \CPTP(A,X)$ and we often denote $\mathcal{M}_X(\rho_A) = \rho_X$.

By the Born Rule, when a pure state $\ket{\psi}$ is measured in an ONB with basis vectors $\left\{\ket{x_i}\right\}_i$, the square modulus of the coefficients $\lambda_x$ of each eigenvector can be interpreted as the probability $p(k)$ that the system will be observed in the state $\ket{x_k}$, i.e.
\begin{align}
	\braket{\psi}{x_k}\braket{x_k}{\psi} & = \sum_i|\lambda_{x_i}|^2 \braket{x_i}{x_k}\braket{x_k}{x_i} \\
	& =\sum_i|\lambda_{x_i}|^2 \delta_{ik}                         \\
	& =|\lambda_k|^2=p(k).
\end{align}
For a mixed state the probability is instead the weighted sum of the coefficients for each pure state in the ensemble.

We may also formalise measurement maps on only part of a multipartite system. In this case we consider $\mathcal{M}^*_X \in \CPTP(AB,XB)$ noting that even though it is defined on the larger space it only acts on the relevant subsystem. With this in mind, the expansion to any additional subsystem is assumed in a given context and we denote $\mathcal{M}_X(\rho_{AB}) = \rho_{XB}$, where it does not cause ambiguity.

The probability that measuring $\rho$ in $\mathbb{X}$ will have the outcome $x$, i.e. the sum of the joint probabilities that $\rho$ is in the state $\ket{\psi_i}$ and that it collapses to $\ket{x}$, induces the random variable $X$ with probability mass function:
$
	\bra{x}\rho\ket{x} = p(x)
$,
 given by the eigenvalues of $\mathcal{M}_X(\rho)$.

\section{Entropy}\label{sec:entropy}
Entropy itself is, in broad terms, a measurement of chaos. Systems that are well structured and predictable have low entropy~--~consider a closed container of liquid water and air~--~whereas high entropy systems have little structure and are hard to predict, such as the same container holding a homogenous mixture of water vapour and air. Basically, a system that is more ``mixed up" has more entropy.

From an information theoretic perspective, entropy is instead viewed in terms of random variables. A random variable with a uniform distribution~--~where each outcome is equally likely~--~has high entropy and one with a degenerate distribution~--~where there is only one possible outcome~--~has low entropy. In terms of the above example, we can consider the distribution of whether or not a water molecule will be observed at a particular position. In the low entropy, liquid state we would be guaranteed one way or the other depending on which side of the surface we measure, but in the high entropy, gaseous state we would be much less sure of the outcome. This statistical interpretation is well formalised by the Shannon entropy.

\subsection{Classical Entropy}
The Shannon entropy $H(X)$ of a random variable $X$ is the most common, and arguably the most useful, definition of entropy. It satisfies the following postulates, desirable for a measurement of uncertainty~\citep{Ren}: If $n=|X|$, then
\begin{enumerate}[(a)]
	\item $H(X) = H(p(x_1), p(x_2), ..., p(x_n))$ \textit{is symmetric},
	\item $H(p(x), 1-p(x))$ \textit{is continuous for} $0<p(x)\leq 1$,
	\item $H(\underbrace{1/n, 1/n,...,1/n}_\text{$n$ times}) = \log_2n$,
	\item $H(tp(x_1),(1-t)p(x_1), p(x_2), ..., p(x_n)) = H(p(x_1), p(x_2), ... ,p(x_n))+p(x_1)H(t, 1-t)$,\\\textit{for} $0\leq t\leq 1$.
\end{enumerate}
\subsubsection{Shannon entropy and surprisal}
To determine a rigorous definition of entropy we first consider surprisal \cite[Section III.A.1]{coles17}.

Let's restrict our focus to a single outcome. The amount of information gained from observing this outcome is related to its probability~--~we gain more information when less likely outcomes occur, i.e. we can consider surprisal to quantify how surprised one would be if the outcome occurred. Additionally, we want the information gained from two independent outcomes to be the sum of the information gained from each outcome, i.e. we want a function $f(P(X=x)) = f(x)$ such that
\begin{align}
	p(x_1)<p(x_2)\implies f(x_1)>f(x_2)
	\text{ and }
	f(x,y) =f(p(x)p(y))= f(x)+f(y).
\end{align}
Any function of the form $-\lambda\log_2 p(x)$, $\lambda\in \R^+$ satisfies these conditions, and we take $\lambda =1$ for simplicity~\citep{shannon48}, giving us the definition $f(x) := -\log_2 p(x)$. [cite]

Since we want to know the expected amount of information gained from observing the random variable we take the average of the surprisal to obtain the Shannon entropy:
\begin{equation}\label{shannon}
	H(X) = -\sum_x p(x)\log p(x).
\end{equation}
This satisfies the postulates (a)-(d) and is an adequate if not preferable measurement of uncertainty in many situations.

For any random variable $X$ its Shannon entropy is bounded:
$
	0\leq H(X)\leq \log |X|
$,
with equality on the left if and only if the distribution of $X$ is degenerate and equality on the right if and only if the distribution of $X$ is uniform on its support.
\subsubsection{Joint entropy}
Joint entropy is the entropy of the random variable described by the joint probability of two or more random variables. Given two random variables $X$ and $Y$, their joint entropy is defined
\begin{equation}
	H(XY) = -\sum_{x,y}p(xy) \log p(xy).
\end{equation}
By Bayes' Theorem~(see Eq.~\eqref{Bayes}) and Jensen's inequality, we have
\begin{equation}\label{jointineq}
	H(XY) \leq H(X)+H(Y),
\end{equation}
with equality if and only if $X$ and $Y$ are independent.
\subsubsection{Conditional entropy}
The conditional entropy of $X$ given $Y$ is average the amount of information left to gain about $X$ once $Y$ has been observed. This is the remaining uncertainty of $XY$ when given access to the side information provided by $Y$ about $X$. As such
\begin{equation}\label{condentdef}
	H(X|Y) = H(XY) - H(Y).
\end{equation}
From Eq.~\eqref{jointineq} we can see that
\begin{equation}\label{condineq}
	H(X|Y)\leq H(X).
\end{equation}
The conditional entropy is equivalently defined \citep{CT}
\begin{align}
	H(X|Y) & = \sum_y p(y)H(X|Y=y)              \\
	       & =-\sum_{x,y} p(y)p(x|y)\log p(x|y)\\
	       &= -\sum_{x,y}p(xy) \log \frac{p(xy)}{p(y)}.
\end{align}

Note that Eq.~\eqref{condentdef} is the bivariate form of the more general multivariate chain rule:
\begin{equation}
	H(X|YZ) = H(XY|Z) + H(Y|Z).
\end{equation}

\subsubsection{Mutual information}
Mutual information, on the other hand, is the amount of information shared between random variables, quantifying how well they are correlated. This notion is conveyed through its definition in terms of the decomposition rules:
\begin{align}
	I(X:Y) &= H(X)+ H(Y) - H(XY)\\ &=H(XY) - H(X|Y) - H(Y|X)\\ &= H(X)-H(X|Y) = H(Y)-H(Y|X).\label{CDecomp}
\end{align}
From Eqs.~\eqref{condineq} and~\eqref{CDecomp} we can see
$
	I(X:Y) \geq 0
$.

\subsubsection{Relative entropy}\label{sec:CRE}
Relative entropy, also know as Kullback-Leibler divergence, is a quantity that can be used to measure the closeness of two distributions over the same index set in terms of their entropies.

For $X_1$ and $X_2$, $P(X_1=x) = p_1(x), \, P(X_2=x) = p_2(x)$ we define
\begin{align}
	D(X_1\|X_2) = \sum_x p_1(x)\log\frac{p_1(x)}{p_2(x)}.
\end{align}

In the case that $p_2(x)$ does not dominate $p_1(x)$, i.e. there exists a $x$ in the index set such that $p_1(x)>0$ and $p_2(x) = 0$, then $D(X_1\|X_2)=\infty$.

We can express classical entropies in terms of this quantity:
\begin{align}
	H(X) &= \log|X|-D\left(p(x)\left\|\frac{\{1\}_x}{|X|}\right.\right)\\
	&=-\sum_{x}p(x)\log {p(x)},\\\nonumber\\
	H(X|Y) &= -D(p(xy)\|p(y))\\
	&= -\sum_{x,y}p(xy)\log \frac{p(xy)}{p(y)}\\
	&= \sum_{y}p(y)\log {p(y)} - \sum_{x,y}p(xy)\log {p(xy)}\\
	&= H(XY) - H(Y),\\\nonumber\\
	I(X:Y) & = D(p(xy)\|p(x)p(y))                                                     \\
	& =\sum_{x,y}p(xy)\log \frac{p(xy)}{p(x)p(y)}                              \\
	& =\sum_{x,y}p(xy)\log p(xy) - \sum_{x}p(x)\log p(x)-\sum_{y}p(y)\log p(y) \\
	& =H(X)+H(Y)-H(XY)\\
	&= H(X) - H(X|Y).
\end{align}

\subsection{Quantum entropies}
We want to find applications of entropy in quantum information, so here we discuss the quantum analogues of the classical entropies covered in the previous sections. But first, a run-down of the concepts and notation.

\subsubsection{Von Neumann entropy}\label{sec:VNE}
Related to the Shannon entropy is its quantum analogue, von Neumann entropy. Instead of taking a random variable as input, the von Neumann entropy takes the probability density matrix $\rho$ describing the state of a quantum system. It is defined similarly to Shannon entropy:
\begin{equation}\label{VNE}
	H(\rho) = -\tr(\rho\log \rho).
\end{equation}
Fortunately, when $\rho$ is expressed in terms of its spectral decomposition~(see Eq.~\eqref{SpD}), the above simplifies to the Shannon entropy of the random variable with probability mass function given by the eigenvalues of $\rho$. i.e.
$
	H(\rho) = -\sum_j\lambda_j\log \lambda_j
$.
We can use the von Neumann entropy as a stepping off point to derive the quantum generalisations of the classical entropic quantities via the intuitive chain and decomposition rules. Accordingly, we have the quantum joint entropy:
\begin{equation}
	H(\rho_{AB}) = -\tr(\rho_{AB}\log\rho_{AB}),
\end{equation}
the quantum conditional entropy:
\begin{equation}
	H(A|B) = H(\rho_{AB}) - H(\rho_B)
\end{equation}
and the quantum mutual information:
\begin{equation}
	I(A:B) = H(\rho_A) + H(\rho_B) - H(\rho_{AB}).
\end{equation}
Similar to the classical entropies, these relationships coincide with the definitions established via the quantum relative entropy~\citep{umegaki62}
\begin{equation}
	D(\rho\|\sigma) = \tr\left[\rho(\log\rho - \log\sigma)\right].
\end{equation}
Analogous to the classical case, if $\sigma \not\gg \rho$ then $D(\rho\|\sigma)=\infty$.

We have for $\rho_{AB}\in \Den(AB)$:
\begin{align}
	H(\rho_{AB})  & = \log|\supp(\rho_{AB})|-D\left(\rho_{AB}\left\|\frac{\id_{AB}}{|\supp(\rho_{AB})|}\right.\right),                                                       \\\nonumber\\
	H(A|B)_\rho & = -D(\rho_{AB}\|\id_A\otimes \rho_B)                                             \\
	            & =-\tr(\rho_{AB}\log\rho_{AB})+\tr(\rho_B\log\rho_B)                              \\
	            & =H(\rho_{AB}) - H(\rho_{B}),                                                         \\\nonumber\\
	I(A:B)_\rho & = D(\rho_{AB}\|\rho_A\otimes\rho_B)                                              \\
	            & =\tr(\rho_{AB}\log\rho_{AB})-\tr(\rho_{A}\log\rho_{A})-\tr(\rho_{B}\log\rho_{B}) \\
	            & =H(\rho_A) + H(\rho_B) - H(\rho_{AB})                                              \\
	            & =H(\rho_A) - H(A|B)_\rho.\label{VNDecomp}
\end{align}

In fact we obtain equivalent definitions for the conditional entropy and mutual information by minimising over one or both subsystems. As this is shown for conditional entropy in~\cite{lennert13}, we provide a similar demonstration for mutual information: if we consider the positive-definiteness of the relative entropy due to Klein's inequality~\citep{Klein}, and observe that when $\sigma_B = \rho_B$ then $D(\rho_B\|\sigma_B) = 0$, we can write
\begin{align}
	I(A:B)_\rho & = D(\rho_{AB}\|\rho_A\otimes\rho_B) + \inf_{\sigma_B \in \Den(B)}D(\rho_B\|\sigma_B)                                                        \\
	            & =\tr(\rho_{AB}\log\rho_{AB} -\rho_{AB}\log(\rho_A\otimes\rho_B)) +\tr(\rho_B\log\rho_B) -\inf_{\sigma_B \in \Den(B)}\tr(\rho_B\log\sigma_B) \\
	            & = \inf_{\sigma_B \in \Den(B)}\tr(\rho_{AB}\log\rho_{AB} -\rho_{AB}\log(\rho_A\otimes\sigma_B))                                              \\
	            & =\inf_{\sigma_B \in \Den(B)}D(\rho_{AB}\|\rho_A\otimes\sigma_B).
\end{align}
A similar calculation can be used to show that the equivalence also holds when minimised over both subsystems, i.e.
\begin{equation}
	I(A:B)_\rho = \inf_{\substack{\sigma_A \in \Den(A)\\\sigma_B \in \Den(B)}}D(\rho_{AB}\|\sigma_A\otimes\sigma_B).
\end{equation}

So far we have been discussing these relative entropy or divergence quantities only in terms of Shannon or von Neumann entropies. It is important to point out that this framework allows for a relatively clear and simple definition of all these quantities but in the case of R\'enyi entropies things are not so straight-forward. Nonetheless, armed with these definitions it is much easier to discuss how to extend this structure to R\'enyi entropies.

\subsection{R\'enyi entropy}
\label{sec:relent}
R\'enyi entropies were first proposed in~\cite{Ren} as an alternative definition of entropy, having most of the desired properties as well as allowing the consideration of information of different order.

The R\'enyi entropy of order $\alpha\in (0,1) \cup (1, \infty)$ of a classical random variable $X$, with probability mass function $p(x)$, is defined as
\begin{equation}
	H_\alpha(X) = \frac{1}{1-\alpha}\log\left(\sum_x p(x)^\alpha\right).
\end{equation}
It generalises the Shannon entropy and serves to weigh outcomes with more or less likelihood differently depending on the order $\alpha$. The Shannon entropy is recovered in the limit $\alpha\rightarrow 1$ (for the proof see Appendix~\ref{sec:alpha1}). R\'enyi's original derivation is beyond the scope of this thesis, but it is constructive to note that in his formalism the Shannon entropy corresponds to a linear case of a more general exponential function where the order $\alpha$ acts as the exponent.

\subsubsection{R\'enyi entropic quantities via relative entropy}
There have been various quantities with which we could derive definitions of R\'enyi conditional entropy and R\'enyi mutual information, the foremost being classical R\'enyi divergence (R\'enyi relative entropy)~\citep{Ren}: given two random variables $X_1$ and $X_2$ indexed on the same set $\{x\}$ with probability mass functions $p_1(x)$ and $p_2(x)$ respectively, the R\'enyi divergence is defined 
\begin{equation}\label{Rdiv}
	D_\alpha(X_1\|X_2) = \frac{1}{\alpha-1} \log \sum_x \frac{p_1(x)^\alpha}{p_2(x)^{\alpha-1}}.
\end{equation}

Some specific examples of classical divergences belonging to this family of divergences are
\begin{align}
	D_0(X_1\|X_2)      & = -\log \sum_x p_1(x)^0 p_2(x) = -\log \sum_{\{x:p_1(x)>0\}} p_2(x),                      \\
	D_{1/2}(X_1\|X_2)  & = -2\sum_{x}\log \sqrt{p_1(x)p_2(x)} \\&= 2D_B(X_1,X_2)\quad\textit{(twice the Bhattacharyya distance~\cite{bhattacharyya})},    \\
	D_1(X_1\|X_2)      & := D(X_1\|X_2),                                                                       \\
	D_2(X_1\|X_2)      & = \log\sum_x p_1(x)\frac{p_1(x)}{p_2(x)} = \log\left\langle\frac{p_1(x)}{p_2(x)}\right\rangle_{X_1} ,   \\
	D_\infty(X_1\|X_2) & := \lim_{\alpha\rightarrow \infty}D_\alpha(X_1\|X_2) = \log \sup_x \frac{p_1(x)}{p_2(x)}.
\end{align}

We would expect the R\'enyi divergence to induce definitions of the the R\'enyi entropic quantities in an analogous way to the definition of classical entropies via relative entropy in Section~\ref{sec:CRE}. That is to say, the R\'enyi conditional entropy and mutual information would be defined
\begin{align}
	H_\alpha(X|Y) = -D_\alpha\left(p_{XY}\|p_Y\right)\quad\text{and}\quad
	I_\alpha(X:Y) = D_\alpha\left(p_{XY}\|p_Xp_Y\right).
\end{align}
Since, in the limit $\alpha \rightarrow 1$, Eq.~\eqref{Rdiv} produces the relative entropy, we expect the above to produce the Shannon and von Neumann quantities in the same limit.

There are some alternate propositions for conditional R\'enyi entropy from which we could also derive a definition of mutual information.
\subsubsection{Alternate definitions}
The main alternative to defining mutual information directly with divergence is to define it via the conditional entropy. There are some definitions which have found utility, derived from analogous definitions based on Shannon entropy or generalisations of other related quantities~\citep{IS}. For example:
\begin{align}
	H_\alpha^\textsc{C}(X|Y)  & = \sum_y p(y)H_\alpha(X|Y=y)   &&     \text{\citep{cachin1997entropy},}               \\
	H_\alpha^\textsc{JA}(X|Y) & = H_\alpha(XY) - H_\alpha(Y)    &&  \text{\citep{jizba04a,jizba04b},}            \\
	H_\alpha^\textsc{RW}(X|Y) & = \frac{1}{1-\alpha} \max_y \log\sum p(x|y)^\alpha&&\text{\citep{renner+wolf}.}
\end{align}
Our classical R\'enyi mutual information could then be defined
\begin{align}
	I^*_\alpha(X:Y) & = H_\alpha(X)+H_\alpha(Y) - H_\alpha(XY) \\
	&=H_\alpha(X)-H^*_\alpha(X|Y)=H_\alpha(Y)-H^*_\alpha(Y|X).
\end{align}
On closer examination the above does not hold in general. Although similar or related techniques to derive a R\'enyi measure of information may find use in specific applications, this method is not fruitful for our purposes.
\subsection{Quantum R\'enyi entropy}\label{sec:QRE}
The quantum R\'enyi entropy is a quantum generalisation of the R\'enyi entropy and is derived in an analogous way to von Neumann entropy: for a probability density matrix $\rho$ and $\alpha \in (0, 1) \cup (1, \infty)$, we define
\begin{equation}
	H_\alpha(\rho) = \frac{1}{1-\alpha}\log\tr(\rho^\alpha).
\end{equation}
There are some particular choices of $\alpha$ which are either mathematically convenient or reflect specific physical situations.

When $\alpha\rightarrow1$ we recover the von Neumann entropy:
$
	H_1(\rho) := H(\rho)
$. When $\alpha = 2$ we have the `collision' entropy: $
H_2(\rho) = -\log\tr(\rho^2)
$, which characterises the purity of a quantum system~(see Eq.~\eqref{purity}). The max-entropy could be naturally defined for $\alpha \rightarrow 0$ but, due to some mathematical restrictions, we instead use $\alpha = \frac{1}{2}$~(see, for example, Eq.~\eqref{RMU}). We have
\begin{align}
	H_0(\rho)                                  & = \log|\supp(\rho)|,     \\
	H_\frac{1}{2}(\rho) = H_\text{max}(\rho) & = 2\log\tr(\sqrt{\rho}).
\end{align}
The last quantity, and perhaps the most useful except for $\alpha \rightarrow 1$, is the min-entropy which we obtain when $\alpha\rightarrow \infty$:
\begin{align}
	H_\infty(\rho) = H_\text{min}(\rho) = -\log\lambda_{\max} [\rho].
\end{align}

\section{R\'enyi divergence and related quantities}
\label{sec:sand}

We now introduce a generalisation of the quantum relative entropy and the classical R\'enyi divergence~--~the `sandwiched' R\'enyi divergence~\citep{lennert13, wilde13}: for $\rho\in \Den(A), \sigma \in\Pos(A)$ and ${\alpha \in (0,1) \cup (1,\infty)}$,
\begin{equation}
	D_\alpha(\rho\|\sigma):=\begin{cases}
		\frac{1}{\alpha-1}\log\tr\left[\left({\sigma}^{\frac{1-\alpha}{2\alpha}}\rho{\sigma}^{\frac{1-\alpha}{2\alpha}}\right)^\alpha\right]\quad & \text{if } \rho \not\perp \sigma \wedge (\sigma\gg\rho \vee \alpha<1) \\
		\infty \quad                                                                                                                              & \text{else}
	\end{cases} .
\end{equation}
From this point we will refer to this quantity as simply `R\'enyi divergence'.

It is prudent at this point to formally introduce the following shorthands which are used extensively in the remainder of the thesis. We often treat the relationships
$
\frac{1}{\alpha}+\frac{1}{\alpha'} = 1$ and $\frac{1}{\alpha}+\frac{1}{\hat\alpha} = 2$.
Accordingly, we equate
\begin{align}
	\alpha' = \frac{\alpha}{\alpha-1}\quad\text{and} \quad
	\hat\alpha =\frac{\alpha}{2\alpha-1}.
\end{align}
This produces some equivalences that will be useful for later calculations. We have
\begin{align}
	\frac{\alpha'}{\alpha}&=\frac{\alpha'(\alpha'-1)}{\alpha'}={\alpha' -1},\\
	\frac{\alpha'}{\hat\alpha} & = \frac{\alpha'(2\alpha-1)}{\alpha}=\frac{\alpha'}{\alpha}\left(\frac{2\alpha'}{\alpha'-1}-1\right) = \left({\alpha'-1}\right)\left(\frac{\alpha'+1}{\alpha'-1}\right) = {\alpha'+1}\quad\text{and} \\
	-\alpha'          & = \frac{\alpha}{1-\alpha} = \left(1- \frac{\hat\alpha}{2\hat\alpha-1}\right)^{-1}\left(\frac{\hat\alpha}{2\hat\alpha-1}\right) = \left( \frac{2\hat\alpha-1}{\hat\alpha-1}\right)\left(\frac{\hat\alpha}{2\hat\alpha-1}\right)=\frac{\hat\alpha}{\hat\alpha-1} = \hat\alpha'.
\end{align}

Applying a similar technique as in Section~\ref{sec:VNE}, we find that equivalence does not extend to R\'enyi divergence but we can still produce rigorous definitions of the relevant quantum R\'enyi entropic quantities. The following notation is adapted from the notation introduced in~\cite{tomamichel13}. We define the quantum R\'enyi entropy as
\begin{align}
	H_\alpha(\rho)& = \log|\supp(\rho)|^\frac{\alpha}{1-\alpha}-D_\alpha\left(\rho\left\|\frac{\id}{|\supp(x)|}\right.\right).
\end{align}
We may also derive quantities that generalise the quantum conditional entropy:
\begin{align}
	\label{maxcond} H^\downarrow_\alpha(A|B)_\rho & = -D_\alpha(\rho_{AB}\|\id_A\otimes \rho_B)\quad\text{and}                 \\
	H^\uparrow_\alpha(A|B)_\rho                   & = -\inf_{\sigma_B \in \Den(B)}D_\alpha(\rho_{AB}\|\id_A\otimes \sigma_B).
\end{align}
The ordering $H^\downarrow_\alpha(A|B)_\rho\leq H^\uparrow_\alpha(A|B)_\rho$, obvious from the definition, is indicated by the direction of the superscript arrow. We can safely assume that $\id_A\otimes \sigma_B\gg \rho_{AB}$, since any choice of $\sigma_B$ where this is not the case would certainly not achieve the infimum. Similarly, we derive generalisations of the quantum mutual information, originally proposed in~\cite{Hay}:
\begin{align}
	I^\uparrow_\alpha(A\;;\>\!B)_\rho            & = \inf_{\sigma_B \in \Den(B)} D_\alpha(\rho_{AB}\|\rho_A\otimes \sigma_B)\label{maxmut}\quad\text{and} \\
	\label{minmut} I^\downarrow_\alpha(A:B)_\rho & = \inf_{\substack{\sigma_A \in \Den(A)                                                                \\\sigma_B \in \Den(B)}} D_\alpha(\rho_{AB}\|\sigma_A\otimes \sigma_B).
\end{align}
Again, the superscript arrows indicate the ordering of each version and we satisfy the support condition as result of the minimisations. The use of `;' in Eq.~\eqref{maxmut} indicates that this quantity is not symmetric between the systems $A$ and $B$, whereas Eq.~\eqref{minmut} is.
\subsection{Generalised R\'enyi entropic quantities}
In order to express our results in a more general setting, we make use of the following notation for the generalised R\'enyi mutual information~\citep{Hay} and conditional entropy:
\begin{align}
	H_\alpha(\rho_{AB}\|\tau_B) & = -D_\alpha(\rho_{AB}\|\id_A\otimes \tau_B),\label{GCE}                              \\
	I_\alpha(\rho_{AB}\|\tau_A) & = \inf_{\sigma_B \in \Den(B)}D_\alpha(\rho_{AB}\|\tau_A\otimes \sigma_B).\label{GMI}
\end{align}
We can readily verify that the above quantities generalise Eqs.~\eqref{maxcond}-\eqref{minmut}, i.e.
\begin{align}
	&  & H^\downarrow_\alpha(A|B)_\rho     & = H_\alpha(\rho_{AB}\|\rho_B), & H^\uparrow_\alpha(A|B)_\rho   & = \sup_{\sigma_B \in \Den(B)}H_\alpha(\rho_{AB}\|\sigma_B), &  & \\
	&  & I^\uparrow_\alpha(A\;;\>\!B)_\rho & = I_\alpha(\rho_{AB}\|\rho_A), & I^\downarrow_\alpha(A:B)_\rho & = \inf_{\sigma_A \in \Den(A)}I_\alpha(\rho_{AB}\|\sigma_A). &  &
\end{align}
\subsection{Properties of the R\'enyi divergence}
This definition of a quantum R\'enyi divergence is not the only one proposed, however, it maintains many useful and desirable properties of the quantities it generalises for a wider range R\'enyi orders, whereas other propositions\footnote{The main contender in this context is the Petz quantum R\'enyi divergence. It is a simpler generalisation but only satisfies many desired properties of a divergence for $\alpha\in [0,2]$. See \cite{mybook} and \cite{petz86} for more details.} do not. There are a few properties of the R\'enyi divergence which we find particularly useful, which we now detail (for a more comprehensive treatment see~\citep{lennert13, wilde13, beigi13}).

Firstly, it generalises the quantum relative entropy. Indeed, it is recovered by taking the limit $\alpha\rightarrow 1$ as in Proposition~\ref{ato1}. It is also monotonically increasing in $\alpha$~\citep{beigi13}, i.e for $\alpha>\beta$,
\begin{equation}\label{monot}
	D_\alpha(\rho\|\sigma) \geq D_\beta(\rho\|\sigma).
\end{equation}

It is useful to re-express the R\'enyi divergence as a Schatten norm of the arguments dependent on $\alpha$:
\begin{equation}
	D_\alpha\left(\rho\|\sigma\right) = \log\left\|\sigma^\frac{-1}{2\alpha'}\rho\sigma^\frac{-1}{2\alpha'}\right\|_{\alpha}^{\alpha'}\label{divNorm}.
\end{equation}
Its close relation to Schatten norms is mathematically convenient and as such derived quantities exhibit duality relations. For example we have the duality of the R\'enyi conditional entropy~\citep{beigi13, lennert13}:
\begin{proposition}\label{condDual}
	For a pure state $\rho_{ABC}\in \Den(ABC)$ and $\alpha>\frac{1}{2}$ such that $\frac{1}{\alpha} + \frac{1}{\hat\alpha} = 2$ we have
	\begin{align}
		H^\uparrow_\alpha(A|B)_\rho = -H^\uparrow_{\hat\alpha}(A|C)_\rho.
	\end{align}
\end{proposition}
And the duality of the generalised mutual information~\citep{Hay}:
\begin{proposition}\label{Dual} For a pure state $\rho_{ABC}\in \Den(ABC), \tau_B\in \Pos(B)$ and $\alpha>\frac{1}{2}$ such that $\frac{1}{\alpha} + \frac{1}{\hat\alpha} = 2$ we have
\begin{equation}
	I_\alpha(\rho_{AB}\|\tau_A) = - I_{\hat\alpha}(\rho_{AC}\|\tau_A^{-1}).
\end{equation}
\end{proposition}
Additionally, it satisfies the data-processing inequality~\citep{beigi13, FL, lennert13}:
\begin{proposition}\label{DPI}
	If $\alpha\geq \frac{1}{2}$ and $\Phi \in \CPTP(A,B)$, then for $\rho\in \Den(A), \sigma\in\Pos(A)$
	\begin{equation}
		D_\alpha(\rho\|\sigma) \geq D_\alpha(\Phi(\rho)\|\Phi(\sigma)).
	\end{equation}
\end{proposition}
In order to extend certain statements about strictly positive matrices to positive semi-definite matrices we can observe the continuity of the R\'enyi divergence. The following proposition is adapted from \cite[Lemma 13]{lennert13} and ensures the R\'enyi divergence is continuous even when the rank of $Y$ decreases.
\begin{proposition}\label{continuity}
	Let $X\in\Den(A),Y\in\Pos(A)$ with $X\neq 0$ and $\alpha\in (0, 1)\cup(1,\infty)$. We have
	\begin{equation}
		D_\alpha(X\|Y)=\lim_{\varepsilon\rightarrow 0^+}\frac{1}{\alpha-1}\log\tr\left[\left({(Y+\varepsilon\id_A)}^{\frac{1-\alpha}{2\alpha}}X{(Y+\varepsilon\id_A)}^{\frac{1-\alpha}{2\alpha}}\right)^\alpha\right]
	\end{equation}
and the limit exists in the weaker sense in which a real valued sequence which is bounded from below and not bounded from above and which does not have an accumulation point is considered as being convergent to $+\infty$.
\end{proposition}
Essentially, this indicates that if we can make a statements for strictly positive matrices, e.g. $Y+\varepsilon\id_A\in \PosS(A)$, then we can extend that statement to positive semi-definite matrices by taking ${Y+\varepsilon\id_A \rightarrow Y}$ with $\supp(Y)\subset \Hil_A$. Note however, that in some cases the quantity will still diverge if $Y\not\gg X$.

These properties naturally extend to any quantity defined using R\'enyi divergence~--~properties which coincide with the mathematical and physical interpretation of quantum entropies. That is to say, monotonicity in $\alpha$ reflects the expected behaviour of R\'enyi entropy when weighing more or less likely outcomes differently and the data-processing inequality reflects that entropy can only ever increase (or correlation decrease) when information is processed (on each system independently).

\subsection{Dupuis' chain rules}\label{sec:dupuis-chain-rule}
One of the major contributions of this work is to formalise and provide R\'enyi approximations of the chain and decomposition rules that exist for Shannon and von Neumann entropies and which often are an important tool in the derivation of related results. To this end, Dupuis' chain rule~\citep{Dup} provides a foundation for the methodology and a proof-of-concept for the more general inequalities covered in Chapter~\ref{sec:renyi-entropy-divergence-inequalities}.
\begin{theorem}\label{chainDup}
	Let $\frac{\alpha}{\alpha-1} = \frac{\beta}{\beta -1} + \frac{\gamma}{\gamma-1}$ with $\alpha,\beta,\gamma \in \left(1/2, 1\right) \cup \left(1,\infty\right)$.	For $\rho_{ABC}\in\Den(ABC)$ and $\tau_C\in \Pos(C)$, if $(\alpha-1)(\beta-1)(\gamma-1)>0$,
	\begin{equation}
		H_\alpha(\rho_{ABC}\|\tau_C)\geq H^\uparrow_\beta(A|BC)_\rho + H_\gamma(\rho_{BC}\|\tau_C).
	\end{equation}
	Otherwise, if  $(\alpha-1)(\beta-1)(\gamma-1)<0$,
	\begin{equation}
		H_\alpha(\rho_{ABC}\|\tau_C)\leq H^\uparrow_\beta(A|BC)_\rho + H_\gamma(\rho_{BC}\|\tau_C).
	\end{equation}
\end{theorem}
This result shows that we may establish chain rules for conditional R\'enyi entropies for particular choices of R\'enyi order. Evidently, these chain rules are weaker than the Shannon and von Neumann versions, but still provide a rigorous generalisation~--~note that by pinching the inequalities we recover an equality when all orders tend to $1$. We provide an alternative proof of this theorem and a variation on the bipartite case in Chapter~\ref{sec:renyi-entropy-divergence-inequalities}.
\section{Entropic uncertainty relations}\label{sec:entropic-uncertainty-relations}
Uncertainty relations manifest in many areas of science, and can be applied to situations ranging across the accuracy of radar readings (trajectory vs position), the analysis of sound waves (frequency vs instant) and the measurement of subatomic particles (position vs momentum).
\subsection{Commutator formulation}\label{sec:ComDef}
In quantum science, uncertainty relations are a way to characterise what are referred to as `incompatible' observables. Observables in this context are operators whose action yields something which can be observed about the state of a quantum system. Incompatibility can be measured by the commutator $[C,D]$, a function of operators. Only when $[C,D] = 0$ are the observables said to commute, otherwise they are incompatible. Note that observables commute only if they are simultaneously diagonalisable.

As explained in Section~\ref{sec:measurement}, when an observable $M$ acts on a quantum state $\ket{\psi}$, we can recover the probability $P(m)$ that a certain value $m$ will be measured for that observable. Using this notation we write the expected value of the observable
\begin{equation}
	\langle\psi|M|\psi\rangle = \langle M\rangle = \sum_m m P(m)
\end{equation}
and its standard deviation
$
\sigma (M) = \sqrt{\langle M^2\rangle - \langle M\rangle^2},
$
which we find by examining the expected difference of the values $m$ from the expected value:
\begin{align}
	(\sigma(M))^2 & = \langle (M-\langle M\rangle)^2\rangle                       \\
	& =\langle M^2 - 2M\langle M\rangle + \langle M\rangle^2\rangle \\
	& =\langle M^2\rangle-\langle M\rangle^2.
\end{align}
Taking the square root will bring things back to the same dimensionality as the observable. Standard deviation is a measure of how spread out a distribution is, and as such is connected to the uncertainty in the measurement of that observable.

The Robertson relation~\cite{Rob} is described by the inequality
\begin{equation}
	\sigma(C)\sigma(D) \geq \frac{|\langle[C,D]\rangle|}{2}.
\end{equation}
This implies the combined spread of the observables must be greater than a function of their incompatibility, the right-hand side is always positive and only equal to zero for observables that commute on the support of $\ket{\psi}$. This allows the left-hand side to be a trivial bound for some situations where the eigenstates of the observables are orthogonal to $\ket{\psi}$.

In general, for non-commuting observables, if the outcomes of one observable can be predicted with a high level of accuracy (small standard deviation) then the uncertainty in predicting the outcome of the other observable must be proportionately high (large standard deviation).

Arguably the most famous application of this relation is the \textit{Heisenberg Uncertainty Principle}~\citep{heisenberg27} relating position and momentum:
\begin{equation}
	\sigma(X)\sigma(P)\geq \frac{\hbar}{2},
\end{equation}
where $X$ is the position observable, $P$ is the momentum observable and $\hbar$ is the reduced Planck's constant, equal to the expected value of the their commutator~\citep{kennard27,weyl28}.
\subsection{Entropic formulation}\label{sec:ED}
To discuss entropic uncertainty relations we consider finite quantum states represented by probability density matrices measured in a choice of ONB~(see Section~\ref{sec:measurement}).

\subsubsection{Uncertainty in measuring in two orthonormal bases}

Since for each $\rho\in\Den(A)$ we can recover the random variables $X$ and $Z$ for two ONBs $\mathbb{X}$ and $\mathbb{Z}$ respectively, we can examine the Shannon entropies $H(X)_\rho$ and $H(Z)_\rho$. Note that this is not the entropy associated with the state $\rho$, or even the post-measurement state, but rather the entropy associated with the measurement of $\rho$ in either basis. We then have the canonical entropic uncertainty relation~\cite{maassen88}:
\begin{theorem}
	For a mixed state $\rho_A\in \Den(A)$ and measurement bases $\mathbb{X}$ and $\mathbb{Z}$ on $A$ we have
	\begin{equation}
		H(X)_\rho + H(Z)_\rho \geq \log\frac{1}{c}=:\qmu, \quad c=\max_{x,z}|\braket{x}{z}|^2.
	\end{equation}
\end{theorem}
This indicates that the total uncertainty of measuring two different observables on the same system must be at least as large as a positive constant depending on the measurement bases but not on the measured state.

We may extend this to a result involving conditional entropy via the following rationale:
consider a generic Shannon entropic relation
\begin{equation}\label{genUR}
	\sum_n H(X_n) \geq q,
\end{equation}
with $q$ state independent. For some classical memory $Y$ containing information about the preparation of the state being measured, we can show~\citep[Section~IV.C]{coles17} that Eq.~\eqref{genUR} implies
$
	\sum_n H(X_n|Y) \geq q
$. We may therefore conclude that
\begin{equation}\label{SMU}
	H(X|Y)_\rho+H(Z|Y)_\rho           \geq \qmu.
\end{equation}

\subsubsection{Information exclusion relations}
Also known as mutual information uncertainty relations, information exclusion relations are entropic uncertainty relations expressed in terms of mutual information. 

The canonical version of this type of relation is the Hall relation~\cite{H95} which we may derive from Eq.~\eqref{SMU} with an application of the decomposition rule, Eq.~\eqref{CDecomp}:
\begin{theorem}\label{Hall}
	For a mixed state $\rho_{AB}\in \Den(A)$ and measurement bases $\mathbb{X}$ and $\mathbb{Z}$ on $A$ and $\mathbb{Y}$ on $B$ we have
	\begin{align}                                    
		I(X:Y)_\rho+I(Z:Y)_\rho  & \leq \log (d^2c)=:r_{\textsc{H}}\label{HREL},
	\end{align}
	where $d$ is the dimension of $\Hil_A$.
\end{theorem}
Note $d$ appears due to both $H(X)$ and $H(Y)$ being bounded above by $\log d$.
\subsection{Improvements and extensions}\label{sec:improvements-and-extensions}
Here we detail some of the pre-existing results relevant to the generalised entropic uncertainty relations in Chapter~\ref{sec:generalised-renyi-divergence-uncertainty-relations}. These take the form of improved bounds, tighter inequalities, extensions to R\'enyi entropies or a combination thereof.
\subsubsection{Improved bounds} 
\citet{Ber} have improved on Eq.~\eqref{SMU} by considering the conditional entropy of the state $\rho$, which quantifies its inherent mixedness:\begin{theorem}\label{CPIER}
	For a mixed state $\rho_{AB}\in \Den(AB)$ and measurement bases $\mathbb{X}$ and $\mathbb{Z}$ on $A$ we have
	\begin{equation}
		H(X|B)_\rho + H(Z|B)_\rho \geq \qmu + H(A|B)_\rho.\label{BMU}
	\end{equation}
\end{theorem}
With this relation as a starting point,~\citet{coles14} have derived a tighter, state-dependent version:
\begin{theorem}\label{CPBUR}
	For a mixed state $\rho_{AB}\in \Den(AB)$ and measurement bases $\mathbb{X}$ and $\mathbb{Z}$ on $A$ we have
	\begin{equation}
		H(X|B)_\rho + H(Z|B)_\rho \geq q(\rho) + H(A|B)_\rho,
	\end{equation}
	where $q(\rho) := \max\{q(\rho, \mathbb{X}, \mathbb{Z}), q(\rho, \mathbb{Z}, \mathbb{X}) \}$ and $q(\rho_A, \mathbb{A}, \mathbb{B}) := -\sum_a p(a)\log\left(\max_b\left|\braket{a}{b}\right|^2\right)$.
\end{theorem}
This can be weakened slightly to a state-independent version:
\begin{theorem}
	For a mixed state $\rho_{AB}\in \Den(AB)$ and measurement bases $\mathbb{X}$ and $\mathbb{Z}$ on $A$ we have
	\begin{equation}
		H(X|B)_\rho + H(Z|B)_\rho \geq q_\textsc{CP} + H(A|B)_\rho,
	\end{equation}
	where $q_\textsc{CP} := \displaystyle{\min_{\sigma\in\Den(AB)} q(\sigma)}$.
\end{theorem}

These relations are then used to produce improved versions of information exclusion relation in Eq.~\eqref{HREL}. Firstly, they proved the conjecture of~\citet{Gru}, who proposed the following:
\begin{proposition}
	For a mixed state $\rho_{AB}\in \Den(AB)$ and measurement bases $\mathbb{X}$ and $\mathbb{Z}$ on $A$ we have
	\begin{align}
		I(X:B)_\rho+I(Z:B)_\rho     \leq r_{\textsc{G}}\quad \text{where}\quad r_{\textsc{G}} : = \log \left(d\sum_{\text{d largest}}\max_{x,z}|\braket{x}{z}|^2 \right),
	\end{align}
with the sum over the largest $d$ terms of the matrix $\left[\left|\braket{x}{z}\right|^2\right]$.
\end{proposition}
Indeed, this is actually a weaker version of the main result of Coles and Piani:
\begin{theorem}
	For a mixed state $\rho_{AB}\in \Den(AB)$ and measurement bases $\mathbb{X}$ and $\mathbb{Z}$ on $A$ we have
	\begin{align}
		I(X:B)_\rho+I(Z:B)_\rho & \leq \rcp,
	\end{align}
	where
	$
	\rcp :          = \min\left\{r(\mathbb{X},\mathbb{Z}),r(\mathbb{Z},\mathbb{X})\right\}$ and $
	r(\mathbb{A},\mathbb{B}):  = \log\left(d\sum_a \max_b \left|\braket{a}{b}\right|^2\right)
	$.
\end{theorem}
\subsubsection{Extensions to R\'enyi entropies}
We now consider the known extensions of these relations to R\'enyi entropies of specific order. In fact, to begin we need look no further than the Maassen-Uffink relation as it was originally proposed~\citep{maassen88} in the general form:
\begin{theorem}
	For a mixed state $\rho_A\in \Den(A)$ and measurement bases $\mathbb{X}$ and $\mathbb{Z}$ on $A$ with $\alpha, \hat\alpha \geq \frac{1}{2}$ such that $ \frac{1}{\alpha} + \frac{1}{\hat\alpha} = 2$,
	\begin{align}
		H_\alpha(X)+H_{\hat\alpha}(Z) \geq \qmu\label{RMU}.
	\end{align}
\end{theorem}
Note for $\alpha,\beta = 1$ we recover Eq.~\eqref{SMU}.

This not only establishes a strong precedent for R\'enyi uncertainty relations but reveals that they are a fundamental part of the mathematical structure of this field.

\citet{mybook} derived the main improvement of this relation, described in terms of quantum R\'enyi conditional entropy: 
\begin{theorem}\label{Marcos}
	For a mixed state $\rho_{AB}\in \Den(AB)$ and measurement bases $\mathbb{X}$ and $\mathbb{Z}$ on $A$ with ${\frac{\alpha}{\alpha-1} = \frac{\gamma}{\gamma-1}+\frac{\beta}{\beta-1}},$ and $(\alpha-1)(\beta-1)(\gamma-1) < 0$.
	\begin{equation}
		H^\uparrow_\gamma(X|B) + H^\uparrow_\beta(Z|B) \geq \qmu + H^\uparrow_\alpha(A|B).
	\end{equation}
\end{theorem}
This result generalises both Eq.~\eqref{BMU} and Eq.~\eqref{RMU}, respectively, by letting all parameters go to $1$ or choosing $B$ trivial and allowing $\alpha \rightarrow 0$.

Once adequate R\'enyi entropy chain and decomposition rules are established, we may adapt the methodology used for the derivations and proofs of these improved relations to demonstrate the results in Chapter~\ref{sec:generalised-renyi-divergence-uncertainty-relations}.

\chapter{Interpolation framework}\label{sec:interpolation-framework}
The principle mechanism which we use to compare R\'enyi divergence of different order is complex interpolation. There is a well-established and sophisticated theory associated with this field of mathematics but in our context we do not need to go into such depth. Our result is essentially an extension of the Hadamard three-line theorem and, although Hadamard's result can be derived via the broader theory, we purposefully restrict our focus. However, in order to understand the motivation behind the choices of interpolation function it is prudent to provide an overview of the interpolation structure and associated norms of Pisier. The statement and proof of the three-line theorem for unbalanced weighted norms on linear operators are given in the last section of this chapter and the specialised results that follow are given in Chapter~\ref{sec:renyi-entropy-divergence-inequalities}, in the context in which they are used.
\section{The super-operator $\Gamma_{\sigma,\tau}$}
For conciseness in the subsequent chapters, we introduce a more general form of the super-operator notation found in~\cite{beigi13}: with $\sigma\in \PosS(B), \tau\in \PosS(A)$, define
\begin{alignat}{4}
	\Gamma_{\sigma,\tau}: &  & \,\Lin(A,B) & \, & \longrightarrow & \, & \Lin(A,B)\label{GammaDef}            \\
	:                     &  & M           &    & \longmapsto     &    & \sigma^\frac{1}{2}M\tau^\frac{1}{2}.
\end{alignat}
If $\sigma$ or $\tau$ are instead operators on only part of a tensor product we still use the above notation when there is no ambiguity, e.g. for $\sigma_C\in\Pos(C),\tau_A \in \Pos(A)$ and $M\in\Lin(AB, CD)$
\begin{equation}
	\Gamma_{\sigma_C,\tau_A}^\frac{1}{\lambda}(M) = \left(\sigma_C^\frac{1}{2\lambda}\otimes \id_D\right)M\left(\tau_A^\frac{1}{2\lambda}\otimes \id_B\right).
\end{equation}
In Eq.~\eqref{GammaDef} we may choose $\Hil_B = \Hil_A$. Moreover, we write $\Gamma_{\sigma,\sigma}^\frac{1}{\lambda} = \Gamma_{\sigma}^\frac{1}{\lambda}$.
\subsection{Unbalanced weighted norms}\label{sec:WN}
We define a norm on $\Lin(A,B)$ in the following way: given $\sigma\in \PosS(B)$ and $\tau\in \PosS(A)$, let
\begin{equation}
	\left\| Y \right\|_{p, (\sigma,  \tau)} = \left\|\sigma^\frac{1}{2p}Y\tau^\frac{1}{2p}\right\|_p = \left\|\Gamma_{\sigma,\tau}^\frac{1}{p}\left( Y \right)\right\|_p.
\end{equation}

It is relatively straight-forward to show this is indeed a norm for $1\leq p\leq \infty$.

By defining the inner product $\left\langle Y, X\right\rangle_{\sigma,\tau} = \tr\left(Y^\dagger \sigma^\frac{1}{2} X \tau^\frac{1}{2}\right)$ we may show the following duality: for $1\leq p\leq \infty$,
\begin{equation}
	\label{weightedDual}\|X\|_{p, ( \sigma, \tau)} = \sup_{\|Y\|_{p', (\sigma, \tau)} = 1}\left|\left\langle Y, X\right\rangle_{\sigma,\tau}\right|,
\end{equation}
where $p'$ is the H\"older conjugate such that $\frac{1}{p} + \frac{1}{p'} = 1$.
\begin{proof}[Proof of Eq.~\eqref{weightedDual}]
	By the definition of the dual norm~\citep[Section~2.10]{kreyszig1989introductory} we may write
	\begin{align}
		\|X\|_{p, ( \sigma,  \tau)}^* & = \sup_{Y\in\Lin(A,B)}\frac{\left|\left\langle Y,X \right\rangle_{\sigma, \tau} \right|}{\left\| Y\right\|_{p, ( \sigma,  \tau)}}                                    \\
		                              & =\sup_{Y\in\Lin(A,B)}\frac{\left|\tr\left(Y^\dagger \sigma^\frac{1}{2} X \tau^\frac{1}{2}\right)\right|}{\left\|\sigma^\frac{1}{2p}Y\tau^\frac{1}{2p}\right\|_p}     \\
		                              & =\sup_{Y\in\Lin(A,B)}\frac{\left|\tr\left((\sigma^\frac{-1}{2p}Y\tau^\frac{-1}{2p})^\dagger \sigma^\frac{1}{2} X \tau^\frac{1}{2}\right)\right|}{\left\|Y\right\|_p} \\
		                              & =\sup_{Y\in\Lin(A,B)}\frac{\left|\tr\left(Y^\dagger\sigma^\frac{-1}{2p} \sigma^\frac{1}{2} X \tau^\frac{1}{2}\tau^\frac{-1}{2p}\right)\right|}{\left\|Y\right\|_p}   \\
		                              & =\sup_{Y\in\Lin(A,B)}\frac{\left|\tr\left(Y^\dagger\sigma^\frac{1}{2p'} X \tau^\frac{1}{2p'}\right)\right|}{\left\|Y\right\|_p}                                      \\
		                              & =\left\|\sigma^\frac{1}{2p'} X \tau^\frac{1}{2p'}\right\|_{p'}                                                                                                       \\
		                              & =\left\| X\right\|_{p', ( \sigma,  \tau)}.
	\end{align}
\end{proof}
The above implies that a H\"older type inequality holds for these norms, i.e.
\begin{equation}\label{HoldLike}
	\left|\left\langle Y, X\right\rangle_{\sigma,\tau}\right| \leq \|X\|_{p, ( \sigma, \tau)}\|Y\|_{p', ( \sigma, \tau)},\quad 1\leq p,p'\leq \infty.
\end{equation}
We must take note that in general these weighted norms are not unitarily invariant, however we may state a specialised unitary invariance: for $U$ and $V$ unitary such that $[U,\sigma] = [V,\tau] = 0$,
\begin{equation}
	\|UXV\|_{p, ( \sigma,  \tau)} = \|UX\|_{p, ( \sigma,  \tau)} = \|XV\|_{p, ( \sigma,  \tau)} = \|X\|_{p, ( \sigma,  \tau)}.
\end{equation}

\subsection{R\'enyi Divergence in terms of weighted norms}
We may establish the following identities for R\'enyi entropic quantities in terms of super-operator notation:
\begin{lemma}\label{QuantEq}
	For $M\in \Lin(AB, C)$ such that $M^\dagger M = \rho_{AB}$, $\tau_B\in \PosS(B)$ and $\alpha \in (0,1)\cup(1, \infty)$:
	\begin{align}
		H_\alpha(\rho_{A})          & =-\log\sup_{\sigma_A \in \Den^*(A)}\left\|\Gamma_{\id_C, \sigma_A}^\frac{1}{\alpha'}(M)\right\|_2^{2\alpha'}\label{EtoD},                           \\
		H_\alpha(\rho_{AB}\|\tau_B) & = -\log\left\| \Gamma_{\id_C, \tau_B}^\frac{-1}{\alpha'}(M) \right\|_{ 2\alpha}^{2\alpha'},\label{CEtoD}                                          \\
		I_\alpha(\rho_{AB}\|\tau_B) & = \log\inf_{\sigma_A\in \Den(A)}\left\| \Gamma_{\id_C,\sigma_A\otimes\tau_B}^{\frac{-1}{\alpha'}}(M) \right\|_{ 2\alpha}^{2\alpha'}.\label{MItoD}
	\end{align}
\end{lemma}
Before we present the proof of Lemma $\ref{QuantEq}$ we find it useful to examine a more general comparison between R\'enyi divergence and weighted norms.
\begin{lemma}\label{DivComp}
	Assume $\alpha\in (0,1) \cup (1,\infty)$. With $M\in \Lin(AB,C)$ such that $M^\dagger M = \rho_{AB}\in \Den(AB)$, $\sigma_A \in \Pos(A)$ and $\tau_B \in \PosS(B)$, we have
	\begin{equation}
		\log\left\| \Gamma_{\id_C,\sigma_A}^{\frac{1}{\alpha}-\frac{1}{\lambda}}\left(\Gamma_{\id_C, \tau_B}^{-1}(M)\right) \right\|_{ 2\alpha,(\id_C,\tau_B^2)}^{2\alpha'} = D_\alpha\left(\rho_{AB}\left\|\sigma_A^{1-\frac{\alpha'}{\lambda'}}\otimes \tau_B \right.\right).\label{DivEq}
	\end{equation}
\end{lemma}
\begin{proof}
	\begin{align}
		\left\| \Gamma_{\id_C,\sigma_A}^{\frac{1}{\alpha}-\frac{1}{\lambda}}(M) \right\|_{ 2\alpha,(\id_C,\tau_B^2)}^2                                                               & =\left\| \Gamma_{\sigma_A}^{\frac{1}{\alpha}-\frac{1}{\lambda}}(\rho_{AB}) \right\|_{ \alpha,\tau_B}                                                                                                                                                      \\
		                                                                                                                                                                             & = \left\| \left(\sigma_A^{\frac{1}{2\alpha} - \frac{1}{2\lambda}}\otimes \tau_B^\frac{1}{2\alpha}\right)\rho_{AB}\left(\sigma_A^{\frac{1}{2\alpha} - \frac{1}{2\lambda}}\otimes \tau_B^\frac{1}{2\alpha}\right) \right\|_{\alpha}                         \\
		                                                                                                                                                                             & =\left\| \left(\sigma_A^{\frac{1}{2\lambda'} - \frac{1}{2\alpha'}}\otimes \tau_B^\frac{-1}{2\alpha'}\right)\Gamma_{\tau_B}(\rho_{AB})\left(\sigma_A^{\frac{1}{2\lambda'} - \frac{1}{2\alpha'}}\otimes \tau_B^\frac{-1}{2\alpha'}\right) \right\|_{\alpha} \\
		\implies\left\| \Gamma_{\id_C,\sigma_A}^{\frac{1}{\alpha}-\frac{1}{\lambda}}\left(\Gamma_{\id_C, \tau_B}^{-1}(M)\right) \right\|_{ 2\alpha,(\id_C,\tau_B^2)}^2               & =\left\| \left(\sigma_A^{1-\frac{\alpha'}{\lambda'}}\otimes \tau_B\right)^\frac{-1}{2\alpha'}\rho_{AB} \left(\sigma_A^{1-\frac{\alpha'}{\lambda'}}\otimes \tau_B\right)^\frac{-1}{2\alpha'}\right\|_{\alpha}                                              \\
		\implies \log\left\| \Gamma_{\id_C,\sigma_A}^{\frac{1}{\alpha}-\frac{1}{\lambda}}\left(\Gamma_{\id_C, \tau_B}^{-1}(M)\right) \right\|_{ 2\alpha,(\id_C,\tau_B^2)}^{2\alpha'} & = D_\alpha\left(\rho_{AB}\left\|\sigma_A^{1-\frac{\alpha'}{\lambda'}}\otimes \tau_B \right.\right).
	\end{align}
\end{proof}

We now make use of the above comparison to prove Lemma~\ref{QuantEq}.
\begin{proof}[Proof of Lemma~\ref{QuantEq}.]
	By definition we have
	\begin{align}
		H_\alpha(\rho_{A}) & = -D_\alpha(\rho_A\|\id_A)                                                                                              \\
		                   & = -\log\|\rho_A\|_{\alpha}^{\alpha'}                                                                                    \\
		                   & =-\log\sup_{\sigma_A \in \Den^*(A)} \left[\tr\left(\rho_A\sigma_A^\frac{1}{\alpha'}\right)^{\alpha'}\right]               \\
		                   & =-\log\sup_{\sigma_A \in \Den^*(A)} \left[\tr\left(\Gamma_{\sigma_A}^\frac{1}{\alpha'}(\rho_{AB})\right)^{\alpha'}\right] \\
		                   & =-\log\sup_{\sigma_A \in \Den^*(A)}\left\|\Gamma_{\id_C, \sigma_A}^\frac{1}{\alpha'}(M)\right\|_2^{2\alpha'},
	\end{align}
	where in the third line we use Lemma~\ref{newLem12}. Note also that in the fourth line, the partial trace on $B$ does not affect operators localised on $A$.

	To show Eq.~\eqref{CEtoD} we begin with the comparison in Lemma~\ref{DivComp} and choose $\lambda = \alpha$ in Eq.~\eqref{DivEq}:
	\begin{align}
		H_\alpha(\rho_{AB}\|\tau_B) & = -D_\alpha(\rho_{AB}\|\id_A\otimes \tau_B)                                                 \\
		                            & =-\log\left\| \Gamma_{\id_C, \tau_B}^{-1}(M) \right\|_{ 2\alpha,(\id_C,\tau_B^2)}^{2\alpha'} \\
		                            & =-\log\left\| \Gamma_{\id_C, \tau_B}^{\frac{1}{\alpha}-1}(M) \right\|_{ 2\alpha}^{2\alpha'} \\
		                            & =-\log\left\| \Gamma_{\id_C, \tau_B}^\frac{-1}{\alpha'}(M) \right\|_{ 2\alpha}^{2\alpha'}.
	\end{align}
	Moreover, choosing instead $\lambda = 1$ and $\tau_B = \id_B$,
	\begin{align}
		H_\alpha(\rho_{AB}\|\sigma_A) & = -D_\alpha(\rho_{AB}\|\sigma_A\otimes \id_B)                                                 \\
		                              & =-\log\left\| \Gamma_{\id_C,\sigma_A}^{\frac{-1}{\alpha'}}(M) \right\|_{ 2\alpha}^{2\alpha'}.
	\end{align}
	After relabelling in the second case we obtain Eq.~\eqref{CEtoD}.
	Similarly, to show Eq.~\eqref{MItoD} we start with Eq.~\eqref{DivEq}, choosing $\lambda = 1$:
	\begin{align}
		I_\alpha(\rho_{AB}\|\tau_B) & = \log\inf_{\sigma_A\in \Den(A)}D(\rho_{AB}\|\sigma_A\otimes \tau_B)                                                                                                           \\
		                            & =\log\inf_{\sigma_A\in \Den(A)}\left\| \Gamma_{\id_C,\sigma_A}^{\frac{-1}{\alpha'}}\left(\Gamma_{\id_C, \tau_B}^{-1}(M)\right) \right\|_{ 2\alpha,(\id_C,\tau_B^2)}^{2\alpha'} \\
		                            & =\log\inf_{\sigma_A\in \Den(A)}\left\| \Gamma_{\id_C,\sigma_A\otimes\tau_B}^{\frac{-1}{\alpha'}}(M) \right\|_{ 2\alpha}^{2\alpha'}.
	\end{align}
	Moreover,
	\begin{align}
		I_\alpha(\rho_{AB}\|\sigma_A) & = \log\inf_{\tau_B\in \Den(B)}\left\| \Gamma_{\id_C,\sigma_A\otimes\tau_B}^{\frac{-1}{\alpha'}}(M) \right\|_{ 2\alpha}^{2\alpha'}.
	\end{align}
	Again, after relabelling, we arrive at Eq.~\eqref{MItoD}.
\end{proof}

\section{Pisier's norms}\label{sec:PN}
Although not used directly in the main results of this thesis, the non-commutative norms of Pisier still inform the interpolation structure and were crucial in understanding the choices which produce the desired comparisons. Their relation to R\'enyi divergence and related quantities has already been put to use in other areas of quantum information~\citep{devetak06, delgosha13}.

In this section we cover the basic concepts, some insights which further simplify things in our context of positive operators and, ultimately, how they factor into the interpolation results of Chapter~\ref{sec:renyi-entropy-divergence-inequalities}.
\subsection{Two-part norms}
Based on the work of~\citet{pisier1998} and the subsequent refinements in~\cite{devetak06} we can use complex interpolation to derive the following definition of a non-commutative norm on the product of complex Hilbert spaces of finite linear operators:
\begin{theorem}
For $1\leq q \leq p \leq \infty$ there is a unique $r\in [1,\infty]$ such that $\frac{1}{q} = \frac{1}{p} + \frac{1}{r}$. With $Y\in \Lin(AB)$ the identities
\begin{align*}
	&&\numberthis \|Y\|_{(p,q)} \equiv& \sup_{M ,N\in \Lin(A)}\frac{\|(M\otimes \id_B)Y(N\otimes \id_B)\|_q}{\|M\|_{2r}\|N\|_{2r}}\label{pq}\\
	\text{and}&&&&\\
	&& \numberthis\|Y\|_{(q,p)}\equiv&\inf_{\substack{M,N\in \Lin(A)\\Y = (M\otimes \id_B)Z(N\otimes \id_B)}} \|M\|_{2r}\|N\|_{2r}\|Z\|_{p}\label{qp}
\end{align*}
define a norm.
\end{theorem}
Although the above expressions are required for the more general case, we may take advantage of our specific context to re-express them in a more workable form. 
\begin{proposition}\label{reduceToDen}
	The expressions in Eq.~\eqref{pq} and Eq.~\eqref{qp} are equivalent to
	\begin{equation}
		\label{pqG}\|Y\|_{(p,q)} = \begin{cases}
			\displaystyle{\sup_{\sigma,\tau\in \Den^*(A)}}\|\Gamma_{\sigma,\tau}^\frac{-1}{p}(Y)\|_{q, (\sigma, \tau)}&\text{if } p\geq q\\
			\displaystyle{\inf_{\sigma,\tau\in \Den^*(A)}}\|\Gamma_{\sigma,\tau}^\frac{-1}{p}(Y)\|_{q, (\sigma, \tau)}&\text{if } p\leq q
		\end{cases}.
	\end{equation}
\end{proposition}
\begin{proof}
	First we note that for all $M,N\in \Lin(A)$ we have a left and right polar decomposition, i.e. there exists unitary operators $U, V\in \Uni(A)$ and positive semi-definite operators $P,Q\in\Pos(A)$ such that
$
		M = UP$ and $N = QV
$.
	Since Schatten norms are unitarily invariant we have that Eq.~\eqref{pq} becomes
	\begin{equation}
		\sup_{P ,Q\in \Pos(A)}\frac{\|(P\otimes \id_B)Y(Q\otimes \id_B)\|_q}{\|P\|_{2r}\|Q\|_{2r}} = \sup_{P ,Q\in \Pos(A)}\frac{\left\|\Gamma_{P,Q}^{2}(Y)\right\|_q}{\|P\|_{2r}\|Q\|_{2r}}.
	\end{equation}
	In order to achieve the supremum we may consider only $P$ and $Q$ such that $\|P\|_{2r},\|Q\|_{2r}<\infty$ and, without loss of generality, we can further use linearity and the scalability of Schatten norms to reduce $P$ and $Q$ to operators such that $\|P\|_{2r} = \|Q\|_{2r} = 1$, in which case we can identify them with density operators with full support, i.e.
	\begin{align}
		\sup_{\sigma ,\tau\in \Den^*(A)}\frac{\left\|\Gamma_{\sigma,\tau}^\frac{1}{r}(Y)\right\|_q}{\left\|\sigma^\frac{1}{2r}\right\|_{2r}\left\|\tau^\frac{1}{2r}\right\|_{2r}} = \sup_{\sigma ,\tau\in \Den^*(A)}\left\|\Gamma_{\sigma,\tau}^\frac{1}{r}(Y)\right\|_q.
	\end{align}
	Note $\frac{1}{2r} = \frac{1}{2q} -\frac{1}{2p}$, hence
	\begin{equation}
		\left\|\Gamma_{\sigma,\tau}^\frac{1}{r}(Y)\right\|_q = \left\|\Gamma_{\sigma,\tau}^\frac{1}{q}(\Gamma_{\sigma,\tau}^\frac{-1}{p}(Y))\right\|_q = \left\|\Gamma_{\sigma,\tau}^\frac{-1}{p}(Y)\right\|_{q, ( \sigma,  \tau)}.
	\end{equation}
	Similarly, if we start with Eq.~\eqref{qp} we arrive at the second case of Eq.~\eqref{pqG}.
\end{proof}
This formulation includes some cases of particular interest.
\begin{corollary}
	\begin{align}
		\|Y\|_{(\infty,q)} &= \sup_{\sigma,\tau\in \Den^*(A)}\|Y\|_{q, ( \sigma,  \tau)},\\
		\|Y\|_{(1, q)} &= \inf_{\sigma, \tau\in \Den^*(A)} \left\|\Gamma_{\sigma,\tau}^{-1}(Y)\right\|_{q, ( \sigma,  \tau)},\\
		\|Y\|_{(p, p)} & = \|Y\|_p.
	\end{align}
\end{corollary}
\begin{proof}
	In Eq.~\eqref{pqG}, replace $p$ with the relevant values.
\end{proof}

An important feature of these norms is that they extend the log-convexity of Schatten norms:
\begin{theorem}	
	For $\theta \in (0,1)$, if $0\leq q_0 < q_1\leq \infty$, $0\leq  p_0 <  p_1\leq \infty$ and 
	\begin{equation}
		\frac{1}{p_\theta} = \frac{1-\theta}{p_0} + \frac{\theta}{p_1},\quad\frac{1}{q_\theta} = \frac{1-\theta}{q_0} + \frac{\theta}{q_1}\label{thetaRel}
	\end{equation} then
	\begin{equation}
		\|Y\|_{(p_\theta, q_\theta)} \leq \|Y\|_{(p_0, q_0)}^{1-\theta}\|Y\|_{(p_1, q_1)}^\theta.\label{pqRT}
	\end{equation}
\end{theorem}
This estimate follows from the original formulation of~\citet{pisier1998} for more general spaces, but showing this in a self-contained way for our purposes has proven non-trivial. Although it is related to our three-line theorem~--~Theorem~\ref{weighted3lines}, and the derived log-convexity results, it is not actually required to demonstrate them. We include it here as an indication of the potential of the general framework, of which we are only utilising a specific part.

\subsection{The case of positive operators}
If we restrict the focus to positive operators we can show that the optimisation is achieved for the same operator, i.e. $\sigma=\tau\in\PosS(A)$. This result and some of the subsequent observations were originally explored in~\cite{devetak06}. It follows that we have a special case of Proposition~\ref{reduceToDen}:
\begin{proposition}
	For $X\in\Pos(AB)$,
	\begin{equation}
		\|X\|_{(p,q)} = \begin{cases}
			\displaystyle{\sup_{\sigma\in \Den^*(A)}}\left\|\Gamma_{\sigma}^\frac{-1}{p}(X)\right\|_{q,\sigma}&\text{if } p\geq q\\
			\displaystyle{\inf_{\sigma\in \Den^*(A)}}\left\|\Gamma_{\sigma}^\frac{-1}{p}(X)\right\|_{q,\sigma}&\text{if } p\leq q
		\end{cases}.
	\end{equation}
\end{proposition}
\begin{proof}
	
	For $X\in\Pos(A), M\in\Lin(A)$, H\"older's inequality implies
	\begin{align}
		\|MXN^\dagger\|_p^2 &= \|(MX^\frac{1}{2}X^\frac{1}{2} N^\dagger)\|_{p}^2\\
		&\leq \|MX^\frac{1}{2}\|_{2p}^2\|X^\frac{1}{2}N^\dagger\|_{2p}^2\\
		&= \|MXM^\dagger\|_p\|NXN^\dagger\|_p\\
		\|MXN^\dagger\|_p&\leq \sqrt{\|MXM^\dagger\|_p\|NXN^\dagger\|_p}\\
		&\leq \max\{\|MXM^\dagger\|_p, \|NXN^\dagger\|_p\}.
	\end{align}
	Moreover, unitary invariance implies $\|MXM^\dagger\|_p = \||M|X|M|\|_p$.
	
\end{proof}

This special case gives rise to the following identity:
\begin{proposition}	
	Given $X_{AB}\in\Pos(AB)$, 	
$
		\|X_{AB}\|_{(p,1)} = \|X_A\|_p\label{p1eqp}
$.
\end{proposition}
\begin{proof}
	Using the duality of Schatten norms we can write
	\begin{align}
		\|X_{AB}\|_{(p,1)} &= \sup_{\sigma\in \Den^*(A)}\left\|\Gamma_{\sigma}^\frac{-1}{p}(X_{AB})\right\|_{1,\sigma}\\
		&= \sup_{\sigma\in \Den^*(A)} {\|\sigma^\frac{1}{2p'}\tr_B(X_{AB})\sigma^\frac{1}{2p'}\|_1}\\
		&= \sup_{\sigma\in \Den^*(A)}{\tr X_{A}\sigma^\frac{1}{p'}}\\
		&= \|X_A\|_p.
	\end{align}	
\end{proof}

\subsection{Weighted two-part norms}
Based on the definition of weighted norms in section \ref{sec:WN}, we may define a weighted version of the two-part norm on $\Lin(A,B)$, with $\sigma\in \PosS(B)$ and $\tau\in \PosS(A)$:
\begin{equation}
	\|X\|_{p,q, (\sigma, \tau)} = \|\Gamma_{\sigma,\tau}^\frac{1}{q}(X)\|_{(p,q)}.
\end{equation}
This definition may not be particularly useful in general but does allow for the following identity:
\begin{proposition} 
	For $Y\in\Lin(AB)$ and $\sigma_B,\tau_B\in \PosS(B)$ we have
	\begin{equation} \label{pqW}
		\|Y\|_{p,q, (\sigma_B, \tau_B)} = \begin{cases}
			\displaystyle{\sup_{\sigma_A,\tau_A\in \Den^*(A)}}\|\Gamma_{\sigma_A,\tau_A}^\frac{-1}{p}(Y)\|_{q, (\sigma_A\otimes\sigma_B,  \tau_A\otimes\tau_B)}&\text{if } p\geq q\\
			\displaystyle{\inf_{\sigma_A,\tau_A\in \Den^*(A)}}\|\Gamma_{\sigma_A,\tau_A}^\frac{-1}{p}(Y)\|_{q, (\sigma_A\otimes\sigma_B,  \tau_A\otimes\tau_B)}&\text{if } p\leq q
		\end{cases}.
	\end{equation}
\end{proposition}
\begin{proof}
	\begin{align}
		\left\|\Gamma_{\sigma_A,\tau_A}^\frac{-1}{p}\left(\Gamma_{\sigma_B,\tau_B}^\frac{1}{q}(Y)\right)\right\|_{q, (\sigma_A, \tau_A)}&=\left\|\left( \sigma_A\otimes \sigma_B \right)^\frac{1}{2q}\left(\sigma_A\otimes \id_B\right)^\frac{-1}{2p}Y\left(\sigma_A\otimes \id_B\right)^\frac{-1}{2p}\left(\tau_A\otimes \tau_B \right)^\frac{1}{2q}\right\|_{q}\\
		&=\left\|\Gamma_{\sigma_A,\tau_A}^\frac{-1}{p}(Y)\right\|_{q, ( \sigma_A\otimes\sigma_B,  \tau_A\otimes\tau_B)}.
	\end{align}

Moreover
$
	\|\Gamma_{\sigma_B,\tau_B}^\frac{1}{q}(Y)\|_{(p,q)}=\|Y\|_{p,q, (\sigma_B, \tau_B)}$,
	therefore, taking the infimum or supremum for the relevant comparisons between $p$ and $q$, we arrive at Eq.~\eqref{pqW}.
\end{proof}

This identity also has a specific form for positive operators:
\begin{corollary}
	With $1\leq q,p\leq \infty$, $X\in\Pos(AB)$ and $\sigma_B\in \PosS(B)$,
	\begin{equation} \label{PospqW}
		\|X\|_{(p,q),\sigma_B} = \begin{cases}
			\displaystyle{\sup_{\sigma_A\in \Den^*(A)}}\|\Gamma_{\sigma_A}^\frac{-1}{p}(X)\|_{q,\sigma_A\otimes\sigma_B}&\text{if } p\geq q\\
			\displaystyle{\inf_{\sigma_A\in \Den^*(A)}}\|\Gamma_{\sigma_A}^\frac{-1}{p}(X)\|_{q,\sigma_A\otimes\sigma_B}&\text{if } p\leq q
		\end{cases}.
	\end{equation}
\end{corollary}
\begin{proof}
	This is clear by replacing $Y$ in Eq.~\eqref{pqW} with $X\in\Pos(AB)$.
\end{proof}

\subsection{Three-part norms}
The framework of~\citet{pisier1998} also allows for a natural extension of the two-part norm. We define a three-part norm on tripartite systems in an analogous way: for $Y\in\Lin(ABC)$
\begin{equation}
	\|Y\|_{(p,q,s)} = \begin{cases}
		\displaystyle{\sup_{\sigma_A,\tau_A\in \Den^*(A)}}\|\Gamma_{\sigma_A,\tau_A}^{\frac{1}{q}-\frac{1}{p}}(Y)\|_{(q,s;AB)}&\text{if } p\geq q\\
		\displaystyle{\inf_{\sigma_A,\tau_A\in \Den^*(A)}}\|\Gamma_{\sigma_A,\tau_A}^{\frac{1}{q}-\frac{1}{p}}(Y)\|_{(q,s;AB)}&\text{if } p\leq q
	\end{cases}.
\end{equation}
In the above we use the notation $\|Y\|_{(p,q;A)}$ to indicate over which parts of the space the optimisation occurs. This is an important distinction as in general, with $Y\in \Lin(ABC)$ and, for instance, $p\leq q$:
\begin{align*}
	\numberthis\|Y\|_{(p,q;A)} &= \inf_{\sigma_A,\tau_A\in \Den^*(A)}\|\Gamma_{\sigma_A,\tau_A}^\frac{-1}{p}(Y)\|_{q, (\sigma_A, \tau_A)}\\ &\neq \\
	\numberthis\|Y\|_{(p,q;AB)} &=\inf_{\sigma_{AB},\tau_{AB}\in \Den^*(AB)}\|\Gamma_{\sigma_{AB},\tau_{AB}}^\frac{-1}{p}(Y)\|_{q, (\sigma_{AB}, \tau_{AB})}.
\end{align*}
This definition allows for a reduction similar to Eq.~\eqref{p1eqp}, originally given in~\cite{devetak06}:
\begin{proposition}
	
	For $X_{ABC} \in\Pos(ABC)$,
	\begin{equation}
		\|X_{ABC}\|_{(p,q,1)} = \|X_{AB}\|_{(p,q)}.\label{pq1eqpq}
	\end{equation}
\end{proposition}
\begin{proof}
	By Eq.~\eqref{p1eqp}, with $\sigma_A\in \PosS(A)$
	\begin{align}
		\|\Gamma_{\sigma_A}^{\frac{1}{q}-\frac{1}{p}}(X_{ABC})\|_{(q,1;AB)} &= \|\Gamma_{\sigma_A}^{\frac{1}{q}-\frac{1}{p}}(X_{AB})\|_{q}\\
		&=\|\Gamma_{\sigma_A}^{\frac{-1}{p}}(X_{AB})\|_{q, \sigma_A}.
	\end{align}
	Choosing $p\geq q$ and optimising over $\sigma_A\in \Den(A)$ on both sides we obtain
	\begin{align}
		\sup_{\sigma_A\in \Den^*(A)}\|\Gamma_{\sigma_A}^{\frac{1}{q}-\frac{1}{p}}(X_{ABC})\|_{(q,1;AB)} &= \sup_{\sigma_A\in \Den^*(A)}\|\Gamma_{\sigma_A}^{\frac{-1}{p}}(X_{AB})\|_{q, \sigma_A},
	\end{align}
	which is equivalent to Eq.~\eqref{pq1eqpq}. Choosing $p\leq q$ we instead take the infimum and arrive again at Eq.~\eqref{pq1eqpq}. 
\end{proof}

\subsection{Adapting the general formulation}
In the remainder of this thesis, we essentially use the above quantities but ignore the optimisations. Overall, the optimisations prove problematic and unnecessary to reproduce the convenient interpolation structure, i.e. we work with expressions of the form:
\begin{equation}
	\left\|\Gamma_{\sigma_A}^\frac{-1}{p}(X)\right\|_{q,\sigma_A\otimes\sigma_B}\quad\text{and}\quad\left\|\Gamma_{\tau_{AB}}^{\frac{1}{s} - \frac{1}{q}}\left(\Gamma_{\sigma_A}^{\frac{1}{q}-\frac{1}{p}}(X)\right)\right\|_{s}.
\end{equation}
Evidently, these expressions do not define a unique value determined by $p$ and $q$ (and $s$) but rather encompass a family of norm-like expressions which are free in the first (and second) parameter.

When we remove the optimisation, the reductions in the previous sections relate closely to situations which produce particular R\'enyi entropic quantities. Moreover, the choice of parameters for the relevant comparisons almost uncannily reproduces the expected form of the chain and decomposition rules. We may even take advantage of the structure to derive general divergence inequalities which prove indispensable in establishing the improved versions of the uncertainty relations and information exclusion relations covered in Chapter~\ref{sec:generalised-renyi-divergence-uncertainty-relations}.

\section{Three-line theorem for unbalanced weighted norms}
This result and its specifications provide the main mechanism by which we compare divergences of different order. We present an extension of the three-line theorem from~\cite{beigi13} to spaces equipped with unbalanced weighted norms.
\begin{theorem}\label{weighted3lines}
	We denote $S := \{z\in \C : 0\leq \Re(z)\leq 1 \}$, the complex strip. Let $F: S \rightarrow \Lin(A,B)$ be a bounded map that is holomorphic on the interior of $S$ and continuous on the boundary. Define
	\begin{equation}
		\mathsf{M}_k = \sup_{t\in\R}\|F(k+\im t)\|_{p_k, ( \sigma, \tau)},
	\end{equation}
	then for $1\leq p_\theta, p_0, p_1\leq \infty$ such that $\frac{1}{p_\theta} = \frac{1-\theta}{p_0} + \frac{\theta}{p_1}$,
	\begin{equation}
		\|F(\theta)\|_{p_\theta, ( \sigma,  \tau)} \leq \mathsf{M}_0^{1-\theta}\mathsf{M}_1^\theta.
	\end{equation}
\end{theorem}
This is not only applicable to norms on spaces of square matrices but also on spaces of non-square matrices~--~a notion we take advantage of to consider R\'enyi orders less than one. Basically, since we are in essence dealing with positive operators we can instead consider their decomposition into the product of another operator and its adjoint.

The proof of Theorem~\ref{weighted3lines} relies on Hadamard's three-line theorem, which we now include for reference.
\begin{theorem}[Hadamard's three-line~\citep{reed1975ii}]\label{H3line}
	Let $f:S\rightarrow \C$ be a bounded function that is holomorphic on the interior of $S$ and continuous on its boundary. For $k\in\{0,1\}$, let
	\begin{equation}
		\mathsf{M}_k = \sup_{t\in \R}|f(k+\im t)|.
	\end{equation}
	Then for every $\theta\in [0,1]$ we have
	\begin{equation}
		\left|f(\theta)\right|\leq \mathsf{M}_0^{1-\theta}\mathsf{M}_1^{\theta}.
	\end{equation}
\end{theorem}

The proof of Theorem~\ref{weighted3lines} closely follows the structure of the proof of Theorem 2 in~\cite{beigi13}, with some adjustments to allow for general spaces of linear operators.
\begin{proof}[Proof of Theorem~\ref{weighted3lines}]
	Let $X\in\Lin(A,B)$ be such that $\|X\|_{p'_\theta, ( \sigma, \tau)} = 1$ and $\|F(\theta)\|_{p_\theta, ( \sigma,  \tau)} =\left|\left\langle X^\dagger, F(\theta)\right\rangle_{\sigma,\tau}\right|$. We have 
	\begin{align}
		\|X\|_{p'_\theta, ( \sigma, \tau)} = \left\|\Gamma_{\sigma, \tau}^\frac{1}{p'_\theta}(X) \right\|_{p'_\theta} = 1,
	\end{align}
	hence the singular value decomposition~(see Proposition~\ref{SVD}) of $\Gamma_{\sigma, \tau}^\frac{1}{p'_\theta}(X)$ has the form
$
		{\Gamma_{\sigma, \tau}^\frac{1}{p'_\theta}(X) = UD^\frac{1}{p'_\theta}V}
$,
	where $U\in \Uni(C,B), V\in \Uni(A, C)$ are isometries and $D$ is a square diagonal matrix in $\Lin(C)$ whose singular values are real, positive and sum to $1$.
	Define
	\begin{equation}
		X(z) = \Gamma_{\sigma, \tau}^{-\left(\frac{1-z}{p'_0} + \frac{z}{p'_1} \right)}\left( UD^{\left(\frac{1-z}{p'_0} + \frac{z}{p'_1} \right)}V\right) = \sigma^{-\left(\frac{1-z}{2p'_0} + \frac{z}{2p'_1} \right)}\left( UD^{\left(\frac{1-z}{p'_0} + \frac{z}{p'_1} \right)}V\right)\tau^{-\left(\frac{1-z}{2p'_0} + \frac{z}{2p'_1} \right)}.
	\end{equation}
	Observe that the map $z\mapsto X(z)$ is holomorphic and $\frac{1}{p_\theta'} = \frac{1-\theta}{p_0'}  + \frac{\theta}{p_1'}$, therefore
	\begin{align}
		X(\theta) = \Gamma_{\sigma, \tau}^{-\frac{1}{p'_\theta}}\left( UD^\frac{1}{p'_\theta}V\right) = \Gamma_{\sigma, \tau}^{-\frac{1}{p'_\theta}}\left(\Gamma_{\sigma, \tau}^\frac{1}{p'_\theta}(X)\right) = X.
	\end{align}
	
	Now define
$
		g(z) = \left\langle X^\dagger(z), F(z)\right\rangle_{\sigma,\tau} = \tr\left( X(z)\sigma^\frac{1}{2}F(z)\tau^\frac{1}{2}\right)
$.
	Hence $g$ satisfies the requirements of Hadamard's three-line theorem (Theorem~\ref{H3line}), and we may write
	\begin{align}
		\|F(\theta)\|_{p_\theta, ( \sigma,  \tau)} & = \left|\left\langle X^\dagger, F(\theta)\right\rangle_{\sigma,\tau}\right|                                                                                                                                                                               \\
		& =\left|\tr\left( X(\theta)\sigma^\frac{1}{2}F(\theta)\tau^\frac{1}{2}\right)\right|                                                                                                                                                                                      \\
		& =|g(\theta)|                                                                                                                                                                                                                                                \\
		& \leq\sup_{t\in \R}|g(\im t)|^{1-\theta}\sup_{t\in \R}|g(1+\im t)|^\theta                                                                                                                                                                                    \\
		& = \sup_{t\in \R}\left|\left\langle X^\dagger(\im t), F(\im t)\right\rangle_{\sigma,\tau}\right|^{1-\theta}\sup_{t\in \R}\left|\left\langle X^\dagger(1+\im t), F(1 + \im t)\right\rangle_{\sigma,\tau}\right|^\theta                                        \\
		& \leq \sup_{t\in \R}\left( \|X(\im t)\|_{p'_0, ( \sigma, \tau)}\|F(\im t)\|_{p_0, ( \sigma, \tau)} \right)^{1-\theta}\sup_{t\in \R}\left( \|X(1+\im t)\|_{p'_1, ( \sigma, \tau)}\|F(1+\im t)\|_{p_1, ( \sigma, \tau)} \right)^\theta,\label{preNormalisation}
	\end{align}
	where in the last line we used the H\"older type inequality, Eq.~\eqref{HoldLike}.
	
	By definition, $\|X(\im t)\|_{p'_0, ( \sigma, \tau)} = \left\|\Gamma_{\sigma, \tau}^\frac{1}{p'_0}(X(\im t)) \right\|_{p'_0}$ and
	\begin{align}
		\Gamma_{\sigma, \tau}^\frac{1}{p'_0}(X(\im t)) & = \Gamma_{\sigma, \tau}^\frac{1}{p'_0}\left(\Gamma_{\sigma, \tau}^{-\left(\frac{1-\im t}{p'_0} + \frac{\im t}{p'_1} \right)}\left( UD^{\left(\frac{1-\im t}{p'_0} + \frac{\im t}{p'_1} \right)}V\right)\right)    \\
		& =\Gamma_{\sigma, \tau}^{-\left(\frac{-\im t}{p'_0} + \frac{\im t}{p'_1} \right)}\left( UD^{\left(\frac{1-\im t}{p'_0} + \frac{\im t}{p'_1} \right)}V\right)                                                       \\
		& = \sigma^{-\left(\frac{-\im t}{2p'_0} + \frac{\im t}{2p'_1} \right)}UD^{\left(\frac{-\im t}{p'_0} + \frac{\im t}{p'_1} \right)}D^\frac{1}{p'_0}V\tau^{-\left(\frac{-\im t}{2p'_0} + \frac{\im t}{2p'_1} \right)}.
	\end{align}
	Observe $\sigma^{-\left(\frac{-\im t}{2p'_0} + \frac{\im t}{2p'_1} \right)}UD^{\left(\frac{-\im t}{p'_0} + \frac{\im t}{p'_1} \right)}$ and $V\tau^{-\left(\frac{-\im t}{2p'_0} + \frac{\im t}{2p'_1} \right)}$ are isometries, thus $\|X(\im t)\|_{p'_0, ( \sigma, \tau)} = 1$. In a similar way we can show $\|X(1+ \im t)\|_{p'_1, ( \sigma, \tau)} = 1$.
	
	Hence Eq.~\eqref{preNormalisation} becomes
	\begin{align}
		\|F(\theta)\|_{p_\theta, ( \sigma,  \tau)} \leq \sup_{t\in \R}\left( \|F(\im t)\|_{p_0, ( \sigma, \tau)} \right)^{1-\theta}\sup_{t\in \R}\left( \|F(1+\im t)\|_{p_1, ( \sigma, \tau)} \right)^\theta.
	\end{align}
\end{proof}

We defer the introduction of the refinements of this theorem to Chapter~\ref{sec:renyi-entropy-divergence-inequalities}, where they can be viewed with the results in which they are used.

\chapter{R\'enyi entropy divergence inequalities}\label{sec:renyi-entropy-divergence-inequalities}
This chapter covers R\'enyi extensions of the chain rule, decomposition rule and more general divergence equivalences in the form of divergence inequalities. The first section details the inequalities and the relationships between them while the technique of interpolation and more detailed proofs are found in the next section.

As it is closely related to these results, we also include a previously published approach to the proof of the decomposition rule in Appendix~\ref{sec:an-alternative-proof-of-the-decomposition-rule}. This alternative approach is based on the method employed by Dupuis in his proof of the R\'enyi entropy chain rules~(see Section~\ref{sec:dupuis-chain-rule}).

The effective difference between these two version of the decomposition rule are the valid ranges of the parameters. For the previous result in Appendix~\ref{sec:an-alternative-proof-of-the-decomposition-rule}, $\alpha$ may go to zero while the other two parameters must be no less than $\frac{1}{2}$. However, for the new result in this chapter, instead $\gamma$ can go to zero and $\alpha$ is greater than or equal to $\frac{1}{2}$.

\section{Main results}\label{sec:main-results}
First, we have a family of general bipartite divergence inequalities. These more versatile comparisons are one of the main tools used to improve the bounds on the uncertainty relations in Chapter~\ref{sec:generalised-renyi-divergence-uncertainty-relations}.
\begin{theorem}\label{DivFormIntLem}
	Let $\alpha\beta\gamma  -2\beta\gamma - \alpha + \beta +\gamma = 0$\footnote{This condition is equivalent to $\alpha' = \beta' + \gamma'$ which may provide a more convenient form for certain calculations. For further details see Remark~\ref{diffForms}.} with $\alpha, \beta\geq \frac{1}{2}$ and $\gamma \in \R$. For $\rho_{AB}\in \Den(AB), \sigma_A\in\PosS(A)$ and $\tau_B\in\Pos(B)$, if $\frac{1}{\beta} + \frac{1}{\gamma}\leq 2$ then
	\begin{equation}
		-H_\alpha\left(\rho_{AB}\left\| \tau_B \right.\right)\leq D_\beta\left(\rho_{AB}\left\|\sigma_A\otimes \tau_B \right.\right)+\log\left(\tr\rho_{A}\sigma_A^\frac{1}{\gamma'} \right)^{\gamma'}.\label{genFor}
	\end{equation}
	Otherwise, if $\frac{1}{\beta} + \frac{1}{\gamma}\geq 2$ then
	\begin{equation}
		-H_\alpha\left(\rho_{AB}\left\| \tau_B \right.\right)\geq D_\beta\left(\rho_{AB}\left\|\sigma_A\otimes \tau_B \right.\right)+\log\left(\tr\rho_{A}\sigma_A^\frac{1}{\gamma'} \right)^{\gamma'}.\label{genRev}
	\end{equation}
Where in both cases $\alpha, \beta<1\vee \id_A\otimes \tau_B \gg \rho_{AB}$.
\end{theorem}
The derivation of these inequalities essentially involves applying the machinery of Theorem~\ref{weighted3lines} to a choice of function motivated by the relationship between R\'enyi divergence and Pisier norms, and performing some analysis on the resulting valid ranges for the R\'enyi parameters.

There are some things to note about this result, especially when it comes to the nature of $\sigma_A$ and $\tau_B$. In general the only requirement for the second argument of the R\'enyi divergence is that it be positive semi-definite, but here we restrict $\sigma_A\in \PosS(A)$. This follows from the application of the interpolation which is only valid for strictly positive matrices. However we can extend the argument to $\tau_B\in \Pos(B)$ by observing the continuity of the R\'enyi divergence~(see Proposition~\ref{continuity}). This does not directly apply to the term on the far right of the above inequalities which involve $\sigma_A$, but this issue is resolved when further optimisation is performed. Indeed, these terms are the logarithm applied to the Schatten inner product so, by duality, optimising over $\sigma_A\in \Den^*(A)$ produces the Schatten $\gamma$-norm.

Even though we may consider any $\tau_B\in\Pos(B)$, we still include the support conditions. Without these conditions the quantities might diverge to $\pm\infty$~--~but we can easily verify that the inequalities would be satisfied in any case, however trivially.

It is informative to consider the behaviour of the condition on the R\'enyi parameters for particular values. We summarise some important cases in the following remark.
\renewcommand{\arraystretch}{1.5}
\begin{remark}\label{diffForms}
	The condition $\alpha\beta\gamma  -2\beta\gamma - \alpha + \beta +\gamma = 0$ with $\alpha, \beta, \gamma \in \R$ can also be written in two other convenient forms:
	\begin{align}
		\frac{\alpha}{\alpha-1} &= \frac{\beta}{\beta-1} + \frac{\gamma}{\gamma-1} \quad\text{and}\\
		\frac{1}{1-\alpha} &= \frac{1}{1-\beta} + \frac{1}{1-\gamma} - 1
	\end{align}
Given this condition the following hold:
	\begin{center}
	\begin{tabular}{lll}
		$\alpha = 0 \implies \beta = \dfrac{\gamma}{1-2\gamma}$,&\qquad&$\gamma = 0\implies \beta = \alpha$,\\
		$\alpha = \dfrac{1}{2} \implies \beta = \dfrac{1-2\gamma}{2-3\gamma}$,&\quad&$\gamma = \dfrac{1}{2} \implies \beta = 2-\dfrac{1}{\alpha} = \dfrac{1}{\hat\alpha}$,\\
		$\alpha = 1 \implies 1\in\{\gamma, \beta\}$,&\qquad&$\gamma = 1 \implies 1\in\{\alpha, \beta\} $,\\
		$\alpha = 2 \implies \beta = 2-\gamma$,&\qquad&$\gamma = 2\implies \beta = \dfrac{\alpha -2}{2\alpha -3}$,\\
		$\alpha\rightarrow \infty \implies \beta = \dfrac{1}{\gamma}$,&\qquad&$\gamma \rightarrow \infty \implies\beta = \dfrac{1}{2-\alpha}$.\\
	\end{tabular}
\end{center}
Combining these statements we may find the situations summarises in Table~\ref{sits}
\end{remark}
\begin{table}\caption{Some particular choices satisfying $\alpha\beta\gamma  -2\beta\gamma - \alpha + \beta +\gamma = 0$}\label{sits}
\begin{center}
	\begin{tabular}{|c||c|c|c|c|c|}
		\hline
		\diagbox{$\gamma =$}{$\alpha =$}&$0$&$\frac{1}{2}$&$1$&$2$&$\infty$\\
		\hhline{|=|=|=|=|=|=|}
		$0$&$\beta = 0$&$\beta = \frac{1}{2}$&$\beta = 1$&$\beta= 2$&$\beta = \infty$\\
		\hline
		$\frac{1}{2}$&$\beta = \infty$&$\beta = 0$&$\beta=1$&$\beta = \frac{3}{2}$&$\beta = 2$\\
		\hline
		$1$&$\beta = 1$&$\beta = 1$&$\beta \in\R$&$\beta = 1$&$\beta = 1$\\
		\hline
		$2$&$\beta = \frac{2}{3}$&$\beta = \frac{3}{4} $&$\beta = 1$&$\beta = 0$&$\beta =\frac{1}{2}$\\
		\hline
		$\infty$&$\beta = \frac{1}{2}$&$\beta = \frac{2}{3}$&$\beta = 1$&$\beta = \infty$&$\beta = 0$\\
		\hline
	\end{tabular}
\end{center}
\end{table}

We cannot always achieve every one of the above combinations due to the conditions on the parameters imposed in the results of this thesis. However, they reveal some of the patterns and structure that the condition implies. We note in particular that although one can choose all parameters equal to $1$ (hence recovering the Shannon and von Neumann equivalences) one is only obliged to choose at least two parameters equal to $1$. Of course, this results in a weaker inequality, whose direction is determined by the other parameter. Interestingly, this also is the only case when $\alpha = \beta$ and $\gamma \neq 0$. 

We can derive the more familiar bipartite comparisons directly from the Theorem~\ref{DivFormIntLem} by relabelling, optimising and making certain specifications.
\begin{theorem}\label{decomp}
	Let $\alpha\beta\gamma  -2\beta\gamma - \alpha + \beta +\gamma = 0$ with $\alpha,\beta\geq \frac{1}{2}  $ and $\gamma\geq 0$.
	
	For $\rho_{AB}\in \Den(AB)$ and $\tau_A\in\Pos(A)$, if $\frac{1}{\beta} + \frac{1}{\gamma} \leq 2$ then
	\begin{equation}\label{decompEq1}
		I_\beta(\rho_{AB}\|\tau_A) \geq H_\gamma(\rho_{B})-H_\alpha(\rho_{AB}\|\tau_A).
	\end{equation}
	Otherwise, if $\frac{1}{\beta} + \frac{1}{\gamma} \geq 2$ then
	\begin{equation}\label{decompEq2}
		I_\beta(\rho_{AB}\|\tau_A) \leq H_\gamma(\rho_{B})-H_\alpha(\rho_{AB}\|\tau_A).
	\end{equation}
Where in both cases $\alpha, \beta<1\vee \tau_A\otimes \id_B \gg \rho_{AB}$.
\end{theorem}
\begin{theorem}\label{Bchain}
	Let $\alpha\beta\gamma  -2\beta\gamma - \alpha + \beta +\gamma = 0 $ with $\alpha,\beta \geq \frac{1}{2}$ and $\gamma\geq0$.
	
	For $\rho_{AB}\in \Den(AB)$, if $\frac{1}{\beta} + \frac{1}{\gamma} \leq 2$ then
	\begin{equation}
		H_\alpha(\rho_{AB})\geq H^\uparrow_\beta(A|B)_\rho + H_\gamma(\rho_{B}).\label{ChainGBP}
	\end{equation}
	Otherwise, if  $\frac{1}{\beta} + \frac{1}{\gamma} \geq 2$ then
	\begin{equation}\label{BchainEq2}
		H_\alpha(\rho_{AB})\leq H^\uparrow_\beta(A|B)_\rho + H_\gamma(\rho_{B}).
	\end{equation}
\end{theorem}
Note the main difference between the two results above is choosing $\tau_A = \id_A$ in the chain rule.

We may combine these results to produce an extension of the other familiar form of the mutual information decomposition rule. This inequality is somewhat weaker in exchange for not involving the conditional entropy. Naturally, we recover equivalence when all parameters tend to $1$.
\begin{corollary}\label{noncond}
	For $\rho_{AB}\in \Den(AB)$ with $\alpha, \beta\geq0$, $\gamma, \delta\geq\frac{1}{2}$ such that $\frac{\delta - \gamma}{\delta\gamma - 2\gamma + 1} = \frac{2\beta\alpha - \beta -\alpha}{\beta\alpha}$ we have, when $\delta<\alpha, \beta, \gamma$,
	\begin{align}
		I^\downarrow_\beta(A:B)_\rho &\geq H_\alpha(\rho_A) + H_\beta(\rho_B) - H_\delta(\rho_{AB})
	\end{align}
	and, when $\delta>\alpha, \beta, \gamma$,
	\begin{align}
		I^\downarrow_\beta(A:B)_\rho &\leq H_\alpha(\rho_A) + H_\beta(\rho_B) - H_\delta(\rho_{AB}).\label{reverse}
	\end{align}
\end{corollary}
This can be determined by applying a further optimisation and substituting the conditional entropy in Theorem~\ref{decomp} with the appropriate comparison in Theorem~\ref{Bchain}. Although not necessary in the main decomposition rules, here we must minimize over $\tau_A$ to produce a statement which is compatible with the bipartite chain rules, hence the specification to $I^\downarrow_\beta(A:B)_\rho$.

We now introduce the tripartite chain rules. Although we provide a new approach to their derivation using our interpolation technique, this is not much more than a reproduction of Dupuis' chain rules~(see Theorem~\ref{chainDup}), with a slight reworking of the conditions on the parameters.
\begin{theorem}\label{chain}
	Let $\alpha\beta\gamma  -2\beta\gamma - \alpha + \beta +\gamma = 0 $ with $\alpha,\gamma \geq \frac{1}{2}$ and $\beta>\frac{1}{2}$.
	
	For $\rho_{ABC}\in \Den(ABC)$ and $\tau_C\in\Pos(C)$, if $\frac{1}{\beta} + \frac{1}{\gamma} \leq 2$ then
	\begin{equation}
		H_\alpha(\rho_{ABC}\|\tau_C)\geq H^\uparrow_\beta(A|BC)_\rho + H_\gamma(\rho_{BC}\|\tau_C).\label{ChainG}
	\end{equation}
	Otherwise, if  $\frac{1}{\beta} + \frac{1}{\gamma} \geq 2$ then
	\begin{equation}
		H_\alpha(\rho_{ABC}\|\tau_C)\leq H^\uparrow_\beta(A|BC)_\rho + H_\gamma(\rho_{BC}\|\tau_C).
	\end{equation}
Where in both cases $\alpha, \gamma<1\vee \id_{AB}\otimes \tau_C \gg \rho_{ABC}$.
\end{theorem}
A significant difference of this tripartite result from the bipartite ones above is the occurrence of non-commuting operators. The resolution of this complication involves introducing purifications and using the resulting dualities to make comparisons on the subsystems with which we are concerned. As a result, this method is closely aligned to Dupuis', but it nonetheless conveys that this type of derivation is contained within a more general structure which warrants closer attention.

From Theorem~\ref{chain} we may also derive a bipartite chain rule which is distinct from Theorem~\ref{Bchain}.
\begin{corollary}\label{BchainCor}
	Let $\alpha\beta\gamma  -2\beta\gamma - \alpha + \beta +\gamma = 0 $ with $\alpha\geq 0,\beta>\frac{1}{2}$ and $\gamma\geq\frac{1}{2}$.
	
	For $\rho_{AB}\in \Den(AB)$, if $\frac{1}{\beta} + \frac{1}{\gamma} \leq 2$ then
	\begin{equation}
		H_\alpha(\rho_{AB})\geq H^\uparrow_\beta(A|B)_\rho + H_\gamma(\rho_{B}).\label{ChainGBPCor}
	\end{equation}
	Otherwise, if  $\frac{1}{\beta} + \frac{1}{\gamma} \geq 2$ then
	\begin{equation}
		H_\alpha(\rho_{AB})\leq H^\uparrow_\beta(A|B)_\rho + H_\gamma(\rho_{B}).
	\end{equation}
\end{corollary}
This is a direct result of choosing $C$ as a trivial subsystem. Note this has slightly different constraints on the parameters, which may be determined by examining the conditions imposed in the proof of Theorem~\ref{chain}.

\section{Proofs}
This section covers the technical proofs of the above results. We begin by with an exploration of the relationship between the R\'enyi parameters then treat the results in more or less logical order. We include a restatement of each theorem inline for ease of reference.

\subsection{Applying the three-line theorem}
Before moving forward with the proofs of the divergence inequalities, we will first look at some motivation for the choice of parameters.

We want to use the interpolation result to find inequalities of a particular form, for example the decomposition rule:
\begin{equation}
- H_\alpha(\rho_{AB}\|\tau_A)_\rho \leq I_\beta(\rho_{AB}\|\tau_A)_\rho - H_\gamma(\rho_B).
\end{equation}

Exponentiating on both sides and keeping in mind that we can express the resulting quantities as Schatten norms to the power of a function of the relevant index we obtain an inequality of the form
\begin{equation}\label{genform}
	\|X_{B|A}\|_{p_\alpha}^{\alpha'}\leq \|X_{B}\|_{p_\gamma}^{\gamma'}\|X_{A:B}\|_{p_\beta}^{\beta'},
\end{equation}
where $X_{B|A}$, etc. are simply place-holders for the actual operators, used for brevity. We can then put Eq.~\eqref{genform} in the form required for Theorem~\ref{weighted3lines} by taking both sides to the power of $\frac{1}{\alpha'}$, resulting in
\begin{equation}
\|X_{B|A}\|_{p_\alpha}\leq \|X_{B}\|_{p_\gamma}^\frac{\gamma'}{\alpha'}\|X_{A:B}\|_{p_\beta}^\frac{\beta'}{\alpha'},
\end{equation}
where $1-\theta = \frac{\gamma'}{\alpha'}$ and $\theta = \frac{\beta'}{\alpha'}$. This implies
\begin{align}
	\label{prerel}1-\frac{\beta'}{\alpha'} = \frac{\gamma'}{\alpha'}
	\implies {\alpha'} = {\beta'} + {\gamma'} .\numberthis
\end{align}

We can find the reverse of the inequality in Eq.~\eqref{genform} by negating all the exponents but this does not affect Eq.~\eqref{prerel}. Additionally, the order of the quantities in Eq.~\eqref{genform} has no effect, since we can choose a $\theta$ in each case that reproduces Eq.~\eqref{prerel}. For example we could rewrite Eq.~\eqref{genform} as
\begin{equation}
	\|X_{A:B}\|_{p_\beta}^{-\beta'}\leq \|X_{B|A}\|_{p_\alpha}^{-\alpha'}\|X_{B}\|_{p_\gamma}^{\gamma'}.
\end{equation}
To apply Theorem~\ref{weighted3lines} in this case we would choose $1-\theta = \frac{\alpha'}{\beta'}$ and $\theta = \frac{-\gamma'}{\beta'}$, resulting in $1+\frac{\gamma'}{\beta'} = \frac{\alpha'}{\beta'}$,
which is again Eq.~\eqref{prerel}.

A more in-depth discussion of the implications and restrictions of this condition, which inform the choices in the following results, is deferred to Appendix~\ref{res}.
\subsection{Bipartite divergence inequalities}
We first have the general bipartite inequality:
\begingroup
\def\thetheorem{\ref*{DivFormIntLem}}
\begin{theorem}
	Let $\alpha\beta\gamma  -2\beta\gamma - \alpha + \beta +\gamma = 0$ with $\alpha, \beta\geq \frac{1}{2}$ and $\gamma \in \R$. For $\rho_{AB}\in \Den(AB), \sigma_A\in\PosS(A)$ and $\tau_B\in\Pos(B)$, if $\frac{1}{\beta} + \frac{1}{\gamma}\leq 2$ then
	\begin{equation}
		-H_\alpha\left(\rho_{AB}\left\| \tau_B \right.\right)\leq D_\beta\left(\rho_{AB}\left\|\sigma_A\otimes \tau_B \right.\right)+\log\left(\tr\rho_{A}\sigma_A^\frac{1}{\gamma'} \right)^{\gamma'}.\label{genForIL}
	\end{equation}
	Otherwise, if $\frac{1}{\beta} + \frac{1}{\gamma}\geq 2$ then
	\begin{equation}
		-H_\alpha\left(\rho_{AB}\left\| \tau_B \right.\right)\geq D_\beta\left(\rho_{AB}\left\|\sigma_A\otimes \tau_B \right.\right)+\log\left(\tr\rho_{A}\sigma_A^\frac{1}{\gamma'} \right)^{\gamma'}.\label{genRevIL}
	\end{equation}
	Where in both cases $\alpha, \beta<1\vee \id_A\otimes \tau_B \gg \rho_{AB}$.
\end{theorem}
\addtocounter{theorem}{-1}
\endgroup

To facilitate the demonstration of the above theorem, we introduce a refinement of Theorem~\ref{weighted3lines}, which establishes a general log-convexity result more closely aligned to our particular context.
\begin{lemma}\label{genLogConv}
	For $Y\in \Lin(A,B)$ and $\sigma_1, \tau_1\in \PosS(B)$, $\sigma_2, \tau_2\in \PosS(A)$
	such that $[\sigma_1, \tau_1] = [\sigma_2, \tau_2] = 0$,
	\begin{align}
		\left\|\Gamma_{\sigma_1, \sigma_2}^{f(\theta)}(Y)\right\|_{q_\theta, ( \tau_1,  \tau_2)}\leq\left\|\Gamma_{\sigma_1, \sigma_2}^{f(0)}(Y)\right\|_{q_0, ( \tau_1,  \tau_2)}^{1-\theta}\left\|\Gamma_{\sigma_1, \sigma_2}^{f(1)}(Y)\right\|_{q_1, ( \tau_1,  \tau_2)}^\theta,\label{genLogConvEq}
	\end{align}
	where $f:(0,1) \rightarrow \R$ is a affine function and $1\leq q_\theta, q_0, q_1 \leq \infty$ are related by
	\begin{equation}
		\frac{1}{q_\theta} = \frac{1-\theta}{q_0} + \frac{\theta}{q_1}.
	\end{equation}
\end{lemma}

The proof of Lemma~\ref{genLogConv} follows from making some specifications in Theorem~\ref{weighted3lines}.
\begin{proof}
	Let $\hat f$ be the complex continuation of $f$ on $S$, i.e if $f(x) = ax + b$ for $x,a,b\in \R$, then $\hat f(z) = az + b$. Define $F:S\rightarrow \Lin(A,B)$ such that
$
		F:z \longmapsto \Gamma_{\sigma_1,\sigma_2}^{\hat f(z)}(Y)
$.
	Accordingly, $f = \Re(\hat f)$, hence $F(\theta) = \Gamma_{\sigma_1,\sigma_2}^{f(\theta)}(Y)$.
	
	Additionally, we have
	\begin{align}
		\sup_{t\in\R}\left\|F(\im t)\right\|_{q_0} & =\sup_{t\in\R}\left\|\Gamma_{\sigma_1, \sigma_2}^{\hat f(\im t)}(Y)\right\|_{q_0, ( \tau_1,  \tau_2)} \\ &= \left\|\Gamma_{\sigma_1, \sigma_2}^{f(0)}(Y)\right\|_{q_0, ( \tau_1,  \tau_2)},
	\end{align}
	where the second equality is a result of the assumption that $\sigma_1^{\im t}$ and $\sigma_2^{\im t}$ commute with $\tau_1$ and $\tau_2$ respectively, and that they are unitary for all $t\in \R$ .
	Similarly, we have
	\begin{equation}
		\sup_{t\in\R}\left\|F(1+\im t)\right\|_{q_1, ( \tau_1,  \tau_2)}=\left\|\Gamma_{\sigma_1, \sigma_2}^{f(1)}(Y)\right\|_{q_1, ( \tau_1,  \tau_2)}.
	\end{equation}
	
	With the above conditions we have by Theorem~\ref{weighted3lines}
	\begin{align}
		\left\|\Gamma_{\sigma_1, \sigma_2}^{f(\theta)}(Y)\right\|_{q_\theta, ( \tau_1,  \tau_2)}\leq\left\|\Gamma_{\sigma_1, \sigma_2}^{f(0)}(Y)\right\|_{q_0, ( \tau_1,  \tau_2)}^{1-\theta}\left\|\Gamma_{\sigma_1, \sigma_2}^{f(1)}(Y)\right\|_{q_1, ( \tau_1,  \tau_2)}^\theta.
	\end{align}
\end{proof}
With the scaffold provided by this log-convexity result, we need only make some informed choices to produce the general divergence inequalities and the comparisons that follow.
\begin{proof}[Proof of Theorem~\ref{DivFormIntLem}]
	Noting $\frac{\alpha'}{\alpha} = \alpha' - 1$, we may determine
	\begin{equation}
		\frac{\alpha'}{\alpha} = \frac{\beta'}{\beta} + \gamma'= \beta' + \frac{\gamma'}{\gamma}. \label{intrel}
	\end{equation}
	Accordingly, in Lemma~\ref{genLogConv} we choose
	\begin{gather}
		\theta = \frac{\gamma'}{\alpha'}, \quad q_\theta = 2\alpha,\quad q_0 = 2\beta, \quad q_1 = 2,\\
		\sigma_1 = \id_C,\quad \sigma_2 = \sigma_A,\quad \tau_1 = \id_C,\quad \tau_2 = \tau_B^2,\\
		Y = \Gamma^{-1}_{\id_C, \tau_B}(M),\quad f(x) = \frac{x-1}{\beta'} + \frac{x}{\gamma'}.
	\end{gather}
	Hence for $\alpha'>0$ we can write Eq.~\eqref{genLogConvEq} as
	\begin{align}
		\left\|\Gamma_{\id_C, \sigma_A}^{\frac{1}{\alpha'} - \frac{1}{\alpha'}}\left(\Gamma^{-1}_{\id_C, \tau_B}(M)\right)\right\|_{2\alpha, ( \id_C,\tau_B^2)}                       & \leq \left\|\Gamma_{\id_C, \sigma_A}^{\frac{-1}{\beta'}}\left(\Gamma^{-1}_{\id_C, \tau_B}(M)\right)\right\|_{2\beta, ( \id_C,\tau_B^2)}^{\frac{\beta'}{\alpha'}}\left\|\Gamma_{\id_C, \sigma_A}^{\frac{1}{\gamma'}}\left(\Gamma^{-1}_{\id_C, \tau_B}(M)\right)\right\|_{2, ( \id_C,\tau_B^2)}^\frac{\gamma'}{\alpha'} \\
		\implies\log\left\|\Gamma^{-1}_{\id_C, \tau_B}(M)\right\|_{2\alpha, ( \id_C,\tau_B^2)}^{\alpha'} & \leq \log\left\|\Gamma_{\id_C, \sigma_A}^{\frac{-1}{\beta'}}\left(\Gamma^{-1}_{\id_C, \tau_B}(M)\right)\right\|_{2\beta, ( \id_C,\tau_B^2)}^{{\beta'}}+\log\left\|\Gamma_{\id_C, \sigma_A}^{\frac{1}{\gamma'}}\left(M\right)\right\|_{2}^{\gamma'}.\label{genIt}                                                            
	\end{align}
Considering Lemmata~\ref{QuantEq} and~\ref{DivComp}, we can write the above as
\begin{align}
		-H_\alpha\left(\rho_{AB}\left\| \tau_B \right.\right)                                                                                                                       & \leq D_\beta\left(\rho_{AB}\left\|\sigma_A\otimes \tau_B \right.\right)+\log\left(\tr\rho_{A}\sigma_A^\frac{1}{\gamma'} \right)^{\gamma'}.
	\end{align}
	Otherwise, if $\alpha'<0$, the inequality in Eq.~\eqref{genIt} is reversed.
	
	If instead $\frac{\alpha'}{\gamma'} \in (0, 1)$ we consider the alternate form of Eq.~\eqref{intrel}:
	\begin{equation}
		 \gamma' = \frac{\alpha'}{\alpha}-\frac{\beta'}{\beta},\quad\frac{\gamma'}{\gamma} = \frac{\alpha'}{\alpha} - \beta.
	\end{equation}
	Taking similar choices in Lemma~\ref{genLogConv}, with $\theta = \frac{\alpha'}{\gamma'}$ and $f(x) = \frac{x-1}{\beta'}$, and employing the same process we may write for $\gamma'>0$:
	\begin{align}
		\left\|\Gamma_{\id_C, \sigma_A}^{\frac{1}{\gamma'}}\left(\Gamma^{-1}_{\id_C, \tau_B}(M)\right)\right\|_{2, ( \id_C,\tau_B^2)} & \leq \left\|\Gamma_{\id_C, \sigma_A}^{\frac{-1}{\beta'}}\left(\Gamma^{-1}_{\id_C, \tau_B}(M)\right)\right\|_{2\beta, ( \id_C,\tau_B^2)}^{\frac{-\beta'}{\gamma'}}\left\|\Gamma^{-1}_{\id_C, \tau_B}(M)\right\|_{2\alpha, ( \id_C,\tau_B^2)}^\frac{\alpha'}{\gamma'},
		\end{align}
	which implies 
		\begin{align}
		-H_\alpha\left(\rho_{AB}\left\| \tau_B \right.\right)                                                                   & \geq D_\beta\left(\rho_{AB}\left\|\sigma_A\otimes \tau_B \right.\right)+\log\left(\tr\rho_{A}\sigma_A^\frac{1}{\gamma'} \right)^{\gamma'}.
	\end{align}
	We again obtain the reverse of this inequality when $\gamma'< 0$.
	
	Finally, we consider $\theta = \frac{\alpha'}{\beta'} \in (0, 1)$ and the form of Eq.~\eqref{intrel}:
	\begin{equation}
		\frac{\beta'}{\beta} = \frac{\alpha'}{\alpha}-\gamma',\quad \beta = \frac{\alpha'}{\alpha} - \frac{\gamma'}{\gamma},
	\end{equation}
	so, with $f(x) = \frac{1-x}{\gamma'}$, for $\beta'>0$ or $\beta' <0$ we respectively obtain
	\begin{equation}
		-H_\alpha\left(\rho_{AB}\left\| \tau_B \right.\right)\geq D_\beta\left(\rho_{AB}\left\|\sigma_A\otimes \tau_B \right.\right)+\log\left(\tr\rho_{A}\sigma_A^\frac{1}{\gamma'} \right)^{\gamma'}
	\end{equation}
	and its reverse.
	
	We now invoke Proposition~\ref{continuity}. Consider instead $\tau_B\in \Pos(B)$~--~then we may choose in the above $\tau_B^*\in \PosS(B)$ such that $\tau_B^* = \tau_B + \varepsilon\id_B.$ Taking $\varepsilon \rightarrow 0^+$, we obtain the same statement for positive semi-definite matrices. Note, without the support conditions $\alpha, \beta<1\vee \id_A\otimes \tau_B \gg \rho_{AB}$, the quantities on each side may diverge to $+\infty$.
	
	To obtain the conditions on the R\'enyi parameters we multiply the equation $\alpha' = \beta' + \gamma'$ by\\$(\alpha-1)(\beta-1)(\gamma-1)$ and apply the observations in Lemma~\ref{ana} and Corollary~\ref{dualcond}, i.e.
	\begin{align}
		\frac{\alpha}{\alpha-1} &= \frac{\beta}{\beta-1} + \frac{\gamma}{\gamma-1}\\
		\implies \alpha(\beta\gamma -\beta - \gamma +1) &= \beta(\alpha\gamma -\alpha - \gamma + 1) + \gamma(\beta\alpha -\alpha -\beta +1)\\
		= \alpha\beta\gamma - \alpha\beta-\alpha\gamma +\alpha &= \alpha\beta\gamma -\alpha\beta - \beta\gamma + \beta + \alpha\beta\gamma -\alpha\gamma - \beta\gamma + \gamma\\
		\implies \alpha\beta\gamma -2\beta\gamma -\alpha + \beta + \gamma &= 0.
	\end{align} Note that this does not cause any trivial satisfactions since $\alpha' \rightarrow \pm\infty$ if and only if at least one of $\beta$ or $\gamma$ approaches $1$.
	
	Taking into account the valid ranges for the choices of the parameters given in Lemma~\ref{ana} we have the statement of the lemma.
\end{proof}
The decomposition rule and chain rule follow directly:
\begingroup
\def\thetheorem{\ref*{decomp}}
\begin{theorem}
	Let $\alpha\beta\gamma  -2\beta\gamma - \alpha + \beta +\gamma = 0$ with $\alpha,\beta\geq \frac{1}{2}  $ and $\gamma\geq 0$.
	
	For $\rho_{AB}\in \Den(AB)$ and $\tau_A\in\Pos(A)$, if $\frac{1}{\beta} + \frac{1}{\gamma} \leq 2$ then
	\begin{equation}\label{decompEq1IL}
		I_\beta(\rho_{AB}\|\tau_A) \geq H_\gamma(\rho_{B})-H_\alpha(\rho_{AB}\|\tau_A).
	\end{equation}
	Otherwise, if $\frac{1}{\beta} + \frac{1}{\gamma} \geq 2$ then
	\begin{equation}\label{decompEq2IL}
		I_\beta(\rho_{AB}\|\tau_A) \leq H_\gamma(\rho_{B})-H_\alpha(\rho_{AB}\|\tau_A).
	\end{equation}
Where in both cases $\alpha, \beta<1\vee \tau_A\otimes \id_B \gg \rho_{AB}$.
\end{theorem}
\addtocounter{theorem}{-1}
\endgroup

\begin{proof}
	We begin with a relabelled version of Eq.~\eqref{genForIL}, i.e. for $\frac{1}{\beta} + \frac{1}{\gamma}\leq {2}$,
	\begin{align}
		-H_\alpha\left(\rho_{AB}\left\| \tau_A \right.\right)          & \leq D_\beta\left(\rho_{AB}\left\| \tau_A\otimes\sigma_B \right.\right)+\log\left(\tr\rho_{B}\sigma_B^\frac{1}{\gamma'} \right)^{\gamma'}.
	\end{align}
	Since the left-hand side is independent of $\sigma_{B}$ we may choose a $\sigma_{B}$ such that the first term is $\varepsilon$-close to its infimum, i.e. we have
	\begin{align}		
		\label{optstart}-H_\alpha\left(\rho_{AB}\left\| \tau_A \right.\right) & \leq \inf_{\sigma_B\in \Den(B)}\left[D_\alpha\left(\rho_{AB}\left\| \tau_A\otimes\sigma_B \right.\right)\right]+\log\left(\tr\rho_{B}\sigma_B^\frac{1}{\gamma'} \right)^{\gamma'} + \varepsilon \\
		\label{optend}-H_\alpha\left(\rho_{AB}\left\| \tau_A \right.\right) & \leq \inf_{\sigma_B\in \Den(B)}\left[D_\alpha\left(\rho_{AB}\left\| \tau_A\otimes\sigma_B \right.\right)\right]+\log\sup_{\sigma_B\in \Den(B)}\left[\left(\tr\rho_{B}\sigma_B^\frac{1}{\gamma'} \right)^{\gamma'}\right] +\varepsilon\\
		\implies -H_\alpha\left(\rho_{AB}\left\| \tau_A \right.\right) & \leq I_\beta\left(\rho_{AB}\left\|\tau_A \right.\right)-H_\gamma(\rho_B)                                                                                                                       \\
		\implies I_\beta\left(\rho_{AB}\left\|\tau_A \right.\right)    & \geq H_\alpha\left(\rho_{AB}\left\| \tau_A \right.\right)-H_\gamma(\rho_B),
	\end{align}
	where in the third line we let $\varepsilon \rightarrow 0$. Observe that the same process applies to the reverse inequality and that in this context we may take $\gamma\rightarrow 0^+$, thus we obtain the theorem.	
\end{proof}
\begingroup
\def\thetheorem{\ref*{Bchain}}
\begin{theorem}
	Let $\alpha\beta\gamma  -2\beta\gamma - \alpha + \beta +\gamma = 0 $ with $\alpha,\beta \geq \frac{1}{2}$ and $\gamma\geq0$.
	
	For $\rho_{AB}\in \Den(AB)$, if $\frac{1}{\beta} + \frac{1}{\gamma} \leq 2$ then
	\begin{equation}
		H_\alpha(\rho_{AB})\geq H^\uparrow_\beta(A|B)_\rho + H_\gamma(\rho_{B}).\label{ChainGBPIL}
	\end{equation}
	Otherwise, if  $\frac{1}{\beta} + \frac{1}{\gamma} \geq 2$ then
	\begin{equation}\label{BchainEq2IL}
		H_\alpha(\rho_{AB})\leq H^\uparrow_\beta(A|B)_\rho + H_\gamma(\rho_{B}).
	\end{equation}
\end{theorem}
\addtocounter{theorem}{-1}
\endgroup
\begin{proof}
	In Theorem~\ref{DivFormIntLem} we relabel the system and choose $\tau_A= \id_{A}$ to obtain, if $\frac{1}{\beta} + \frac{1}{\gamma}\leq 2$,
	\begin{align}
		-H_\alpha\left(\rho_{AB}\right)&\leq D_\beta\left(\rho_{AB}\left\| \id_A\otimes\sigma_B \right.\right)+\log\left(\tr\rho_{B}\sigma_B^\frac{1}{\gamma'} \right)^{\gamma'}.
	\end{align}
	Optimising over $\sigma_B$ in the same manner as in Eqs.~\eqref{optstart}-\eqref{optend} the above becomes
	\begin{align}
		-H_\alpha\left(\rho_{AB}\right)&\leq -H^\uparrow_\beta\left(A|B\right)_\rho+H_\gamma(\rho_B)\\
		\implies H_\alpha\left(\rho_{AB}\right)&\geq H^\uparrow_\beta\left(A|B\right)_\rho-H_\gamma(\rho_B).
	\end{align}
	Similarly, for $\frac{1}{\beta} + \frac{1}{\gamma} \geq 2$ we have the reverse direction.
\end{proof}
We now include the proof of Corollary~\ref{noncond}, essentially combining the two above results.
\begingroup
\def\thetheorem{\ref*{noncond}}
\begin{corollary}
	For $\rho_{AB}\in \Den(AB)$ with $\alpha, \beta\geq0$, $\gamma, \delta\geq\frac{1}{2}$ such that $\frac{\delta - \gamma}{\delta\gamma - 2\gamma + 1} = \frac{2\beta\alpha - \beta -\alpha}{\beta\alpha}$ we have, when $\delta<\alpha, \beta, \gamma$,
	\begin{align}
		I^\downarrow_\beta(A:B)_\rho &\geq H_\alpha(\rho_A) + H_\beta(\rho_B) - H_\delta(\rho_{AB})
	\end{align}
	and, when $\delta>\alpha, \beta, \gamma$,
	\begin{align}
		I^\downarrow_\beta(A:B)_\rho &\leq H_\alpha(\rho_A) + H_\beta(\rho_B) - H_\delta(\rho_{AB}).\label{reverseIL}
	\end{align}
\end{corollary}
\addtocounter{theorem}{-1}
\endgroup
\begin{proof}
	We begin with a relabelled version of Eq.~\eqref{decompEq1IL}, optimising over $\tau_B$: For ${\tilde\alpha\beta\alpha  -2\beta\alpha -\tilde\alpha +\beta +\alpha =0}$ with $\tilde\alpha, \beta\geq \frac{1}{2}$, $\alpha\geq0$ and $\frac{1}{\beta} + \frac{1}{\alpha}\leq 2$ we have
	\begin{align}
		\inf_{\tau_B \in \Den(B)}I_\beta(\rho_{AB}\|\tau_B) &\geq H_\alpha(\rho_{A})-\sup_{\tau_B \in \Den(B)}H_{\tilde\alpha}(\rho_{AB}\|\tau_B)\\
		I^\downarrow_\beta(A:B) &\geq H_\alpha(\rho_{A})-H^\uparrow_{\tilde\alpha}(A|B).
	\end{align}
By Theorem~\ref{Bchain}, for for $\delta\tilde\alpha\gamma  -2\tilde\alpha\gamma - \delta + \tilde\alpha +\gamma = 0$ with $\delta, \tilde\alpha\geq \frac{1}{2}$, $\gamma \geq0$ and $\frac{1}{\tilde\alpha} + \frac{1}{\gamma}\leq 2$ we can substitute
\begin{align}
	I^\downarrow_\beta(A:B) &\geq H_\alpha(\rho_{A}) + H_\gamma(\rho_B) - H_\delta(\rho_{AB}).
\end{align}
	where $\frac{\delta - \gamma}{\delta\gamma - 2\gamma + 1} = \frac{2\beta\alpha - \beta -\alpha}{\beta\alpha}$.
	
	We know from Corollary~\ref{dualcond} that both
	\begin{align}
		\frac{1}{\beta} + \frac{1}{\alpha}\leq 2\quad\text{then}\quad&\tilde\alpha<\alpha, \beta \quad\text{and}\\
		\frac{1}{\tilde\alpha} + \frac{1}{\gamma}\leq 2\quad\text{then}\quad&\delta<\tilde\alpha, \gamma.
	\end{align}
	Similarly, if we begin with Eq.~\eqref{decompEq2IL} and substitute in Eq.~\eqref{BchainEq2IL}, we arrive at Eq.~\eqref{reverseIL} but with
	\begin{align}
		\frac{1}{\beta} + \frac{1}{\alpha}\geq 2\quad\text{then}\quad&\tilde\alpha>\alpha, \beta \quad\text{and}\\
		\frac{1}{\tilde\alpha} + \frac{1}{\gamma}\geq 2\quad\text{then}\quad&\delta>\tilde\alpha, \gamma.
	\end{align}
\end{proof}
\subsection{The tripartite chain rules}
We now include the novel approach to the proof of the tripartite chain rule and the bipartite version which follows.
\begingroup
\def\thetheorem{\ref*{chain}}
\begin{theorem}
	Let $\alpha\beta\gamma  -2\beta\gamma - \alpha + \beta +\gamma = 0 $ with $\alpha,\gamma \geq \frac{1}{2}$ and $\beta>\frac{1}{2}$.
	
	For $\rho_{ABC}\in \Den(ABC)$ and $\tau_C\in\Pos(C)$, if $\frac{1}{\beta} + \frac{1}{\gamma} \leq 2$ then
	\begin{equation}
		H_\alpha(\rho_{ABC}\|\tau_C)\geq H^\uparrow_\beta(A|BC)_\rho + H_\gamma(\rho_{BC}\|\tau_C).\label{ChainGIL}
	\end{equation}
	Otherwise, if  $\frac{1}{\beta} + \frac{1}{\gamma} \geq 2$ then
	\begin{equation}
		H_\alpha(\rho_{ABC}\|\tau_C)\leq H^\uparrow_\beta(A|BC)_\rho + H_\gamma(\rho_{BC}\|\tau_C).
	\end{equation}
Where in both cases $\alpha, \gamma<1\vee \id_{AB}\otimes \tau_C \gg \rho_{ABC}$.
\end{theorem}
\addtocounter{theorem}{-1}
\endgroup
To demonstrate the above result we first derive a specific form of Theorem~\ref{weighted3lines}:
\begin{lemma}\label{genChainDual}
	We define the complex strip $S:=\{ z\in \C \, :\, 0\leq \Re(z) \leq 1\}$. Let $F:S\rightarrow \Lin(A,B)$ be a function of the form
	\begin{equation}
		F(z) = \Gamma_{\sigma_{B},\id_{A}}^{(1-z)\left(\frac{1}{s_0} - \frac{1}{q_0} \right) + z\left(\frac{1}{s_1} - \frac{1}{q_1}\right)}\left(\Gamma_{\id_B,\tau_A}^{(1-z)\left(\frac{1}{q_0} - \frac{1}{p_0} \right) + z\left(\frac{1}{q_1} - \frac{1}{p_1}\right)}(M)\right),
	\end{equation}
	where $\sigma_{B}\in\PosS(B), \tau_A \in \PosS(A)$ and $M\in \Lin(A,B)$.
	
	Denote
	\begin{equation}
		\mathsf{M}_k = \sup_{t\in \R}\|F(k+\im t)\|_{2s_k}.
	\end{equation}
	
	Let $\alpha' = \beta' + \gamma'$ and $\frac{\gamma'}{\alpha'}\in (0,1)$. Then, given $p_\theta, p_1, q_0, q_1 \in \R^+$, $s_\theta, s_0, s_1 \in [1/2, \infty]$ such that
	\begin{equation}
		\frac{\alpha'}{p_\theta} = \frac{\beta'}{p_0} + \frac{\gamma'}{p_1},\quad\frac{\alpha'}{q_\theta} = \frac{\beta'}{q_0} + \frac{\gamma'}{q_1},\quad \frac{\alpha'}{s_\theta} = \frac{\beta'}{s_0} +\frac{\gamma'}{s_1},
	\end{equation}
	we have
	\begin{equation}\label{genChainEqDual}
		\left\|\Gamma_{\sigma_{B}, \id_A,}^{\frac{1}{s_\theta} - \frac{1}{q_\theta} }\left(\Gamma_{\id_B,\tau_A}^{\frac{1}{q_\theta} - \frac{1}{p_\theta}}(M)\right)\right\|_{2s_\theta}^{\alpha'} \begin{cases}
			\leq \mathsf{M}_0^{\beta'}\mathsf{M}_1^{\gamma'}\quad & \text{if } \alpha' > 0 \\
			\geq \mathsf{M}_0^{\beta'}\mathsf{M}_1^{\gamma'}\quad & \text{if } \alpha' < 0
		\end{cases}.
	\end{equation}	
\end{lemma}
\begin{proof}
	We let $\theta = \frac{\gamma'}{\alpha'} \implies 1-\theta = \frac{\beta'}{\alpha'}$ then by Theorem~\ref{weighted3lines} we have
$
		\left\|F(\theta)\right\|_{2s_\theta}
		\leq \mathsf{M}_0^{\frac{\beta'}{\alpha'}}\mathsf{M}_1^{\frac{\gamma'}{\alpha'}}
$.
	By evaluating $F(\theta)$ and exponentiating by $\alpha'$ on both sides we arrive at Eq.~\eqref{genChainEqDual}.
\end{proof}

\begin{proof}[Proof of Theorem~\ref{chain}]
	Noting $\frac{\alpha'}{\alpha} = \alpha' - 1$ and $\frac{\alpha'}{\hat\alpha} = \alpha' + 1$ We may determine
	\begin{align}
		\frac{\alpha'}{\alpha} =  \beta'+\frac{\gamma'}{\gamma}\quad\text{and}\quad\alpha'=  \frac{\beta'}{\hat\beta}+\frac{\gamma'}{\gamma}.\label{intrel2}
	\end{align}
	Let $M \in \Lin(BC, AD)$, $\sigma_{D} \in \PosS(D)$ and $\tau_C\in\PosS(C)$.	In Lemma~\ref{genChainDual} choose $p_\theta, p_0, p_1, q_0,s_\theta = 1$, $q_\theta = \alpha$, $q_1, s_1 = \gamma$ and $s_0 = \hat\beta$. Then for $\theta = \frac{\gamma'}{\alpha'}$ and $\alpha'>0$ we can write
	\begin{align}
		\left\|\Gamma_{\sigma_D, \id_{BC}}^\frac{1}{\alpha'}\left(\Gamma_{\id_{AD},\tau_C}^\frac{-1}{\alpha'}(M)\right)\right\|_2^{\alpha'} & \leq \sup_{t\in\R}\left\|\Gamma_{\sigma_D,\id_{BC}}^{\frac{1-\im t}{\beta'}}\left(\Gamma_{\id_{AD},\tau_C}^{\frac{-\im t}{\gamma'}}(M)\right)\right\|_{2\hat\beta}^{\beta'}\sup_{t\in \R}\left\|\Gamma_{\sigma_D,\id_{BC}}^{\frac{-\im t}{\beta'} }\left(\Gamma_{\id_{AD},\tau_C}^{ \frac{-(1+\im t)}{\gamma'}}(M)\right)\right\|_{2\gamma}^{\gamma'}.\label{firstItChain}
	\end{align}
	Let $\ket{\varphi} \in \Hil_{ABCD}$ be a pure state with Schmidt decompositions~(see Proposition~\ref{SchmidtD}):
	\begin{equation}
		\ket{\varphi} = \sum_i r_i  \ket{i}_{BC} \otimes\ket{i}_{AD} = \sum_i s_i \ket{i}_{ABC} \otimes\ket{i}_{D}
	\end{equation}
	and $\rho_{ABCD} = \ketbra{\varphi}{\varphi}.$
	
	We choose $M = \sum_i r_i  \ket{i}_{AD}\bra{i}_{BC}$, so by Proposition~\ref{traceFree} the above becomes
	\begin{align}
		\left\|\Gamma_{\sigma_D, \id_{ABC}}^\frac{1}{\alpha'}\left(\Gamma_{\id_{D},\tau_C}^\frac{-1}{\alpha'}\left(\sum_i s_i  \ket{i}_{D}\bra{i}_{ABC}\right)\right)\right\|_2^{\alpha'} & \leq \left\|\Gamma_{\sigma_D,\id_{BC}}^{\frac{1}{\beta'}}\left(M\right)\right\|_{2\hat\beta}^{\beta'}\left\|\Gamma_{\id_{AD},\tau_C}^{ \frac{-1}{\gamma'}}(M)\right\|_{2\gamma}^{\gamma'},
	\end{align}
	where we have also used the fact that $\sigma_D^{\im t}$ and $\tau_C^{\im t}$ are unitary for all $t\in \R$.
	
	Taking the supremum over $\sigma_D\in \Den^*(D)$ on both sides we obtain via Lemma~\ref{newLem12}:
	\begin{align}
		\implies \left\|\Gamma_{\id_{D},\tau_C}^\frac{-1}{\alpha'}\left(\sum_i s_i  \ket{i}_{D}\bra{i}_{ABC}\right)\right\|_{2\alpha}^{\alpha'} & \leq \sup_{\sigma_D \in \Den^*(D)}\left\|\Gamma_{\sigma_D,\id_{BC}}^{\frac{1}{\beta'}}\left(M\right)\right\|_{2\hat\beta}^{\beta'}\left\|\Gamma_{\id_{AD},\tau_C}^{ \frac{-1}{\gamma'}}(M)\right\|_{2\gamma}^{\gamma'}.
	\end{align}
	Further taking the logarithm of both sides and noting the duality of the conditional entropy~(see Proposition~\ref{condDual}) we have by Lemma~\ref{QuantEq}
	\begin{align}
		-H_\alpha(\rho_{ABC}\|\tau_C)         & \leq H_{\hat\beta}^\uparrow(A|D)_{\rho} - H_\gamma(\rho_{BC}\|\tau_C) \\
		\implies H_\alpha(\rho_{ABC}\|\tau_C) & \geq H_\beta^\uparrow(A|BC)_{\rho} + H_\gamma(\rho_{BC}\|\tau_C).
	\end{align}
	If instead $\alpha'<0$, the inequality in Eq.~\eqref{firstItChain} is reversed and we may use a similar optimisation to arrive at
	\begin{align}
		H_\alpha(\rho_{ABC}\|\tau_C) & \leq H_\beta^\uparrow(A|BC)_\rho + H(\rho_{BC}\|\tau_C).
	\end{align}

	Note that we may rewrite Eq.~\eqref{intrel2} in the forms
	\begin{align}
		\frac{\gamma'}{\gamma} = \frac{\alpha'}{\alpha} - \beta' \quad\text{and}\quad\frac{\gamma'}{\gamma} =\alpha'-  \frac{\beta'}{\hat\beta}.
	\end{align}
	Accordingly, if we consider instead $\frac{\alpha'}{\gamma'}\in (0,1)$ and $\gamma'>0$ we may make the following choices in Lemma~\ref{genChainDual}: $p_\theta, p_0, p_1, q_0, s_1 = 1$, $q_\theta, s_\theta = \gamma$, $q_1 = \alpha$ and $s_0 = \hat\beta$. This yields
	\begin{equation}
		\left\|\Gamma_{\id_{AD},\tau_C}^{ \frac{-1}{\gamma'}}(M)\right\|_{2\gamma}^{\gamma'} \leq \sup_{t\in\R}\left\|\Gamma_{\sigma_D,\id_{BC}}^{\frac{1-\im t}{\beta'} + \frac{\im t}{\alpha'}}\left(\Gamma_{\id_{AD},\tau_C}^{\frac{-\im t}{\alpha'}}(M)\right)\right\|_{2\hat\beta}^{-\beta'}\sup_{t\in \R}\left\|\Gamma_{\sigma_D, \id_{BC}}^{\frac{-\im t}{\beta'} + \frac{1+\im t}{\alpha'}}\left(\Gamma_{\id_{AD},\tau_C}^\frac{-(1+\im t)}{\alpha'}(M)\right)\right\|_2^{\alpha'},
	\end{equation}
	which, after rearranging, choosing the correct Schmidt decompositions, optimising and further taking the logarithm to then employ the identities in Lemma~\ref{QuantEq} we obtain
	\begin{equation}
		H_\alpha(\rho_{ABC}\|\tau_C)\leq H_\beta^\uparrow(A|BC)_{\rho} + H(\rho_{BC}\|\tau_C).
	\end{equation}
	Similarly, this inequality is reversed for $\gamma' < 0$.
	
	Finally, we consider $\frac{\alpha'}{\beta'}\in (0,1)$ with $\beta'> 0$ or $\beta'<0$. Choosing $p_\theta, p_0, p_1,q_\theta, s_1 = 1$, $q_0, s_0 = \gamma$, $q_1 = \alpha$ and $s_\theta = \hat\beta$ we again derive, respectively,
	\begin{equation}
		H_\alpha(\rho_{ABC}\|\tau_C)\leq H_\beta^\uparrow(A|BC)_{\rho} + H(\rho_{BC}\|\tau_C)
	\end{equation}
	and its reverse.
	
	The conditions on the R\'enyi parameters can be derived in the same way as for Theorem~\ref{DivFormIntLem} and the extension to positive semi-definite matrices follows from Proposition~\ref{continuity}.
\end{proof}
We can specialise this to result to a bipartite setting.
\begingroup
\def\thetheorem{\ref*{BchainCor}}
\begin{corollary}
	Let $\alpha\beta\gamma  -2\beta\gamma - \alpha + \beta +\gamma = 0 $ with $\alpha\geq 0,\beta>\frac{1}{2}$ and $\gamma\geq\frac{1}{2}$.
	
	For $\rho_{AB}\in \Den(AB)$, if $\frac{1}{\beta} + \frac{1}{\gamma} \leq 2$ then
	\begin{equation}
		H_\alpha(\rho_{AB})\geq H^\uparrow_\beta(A|B)_\rho + H_\gamma(\rho_{B}).\label{ChainGBPCorIL}
	\end{equation}
	Otherwise, if  $\frac{1}{\beta} + \frac{1}{\gamma} \geq 2$ then
	\begin{equation}
		H_\alpha(\rho_{AB})\leq H^\uparrow_\beta(A|B)_\rho + H_\gamma(\rho_{B}).
	\end{equation}
\end{corollary}
\addtocounter{theorem}{-1}
\endgroup
\begin{proof}
	We simply choose $C$ to be trivial in Theorem~\ref{chain} and examine the valid ranges in its proof.
\end{proof}

\chapter{Generalised R\'enyi divergence uncertainty relations}\label{sec:generalised-renyi-divergence-uncertainty-relations}
We now arrive at the applications which motivate the inequalities of the previous section. These take the form of some refinements, extensions and improvements of the relations given in Section~\ref{sec:entropic-uncertainty-relations}. We collect the Maassen-Uffink-like, bipartite uncertainty relations in the first section of this chapter and the information exclusion relations in the second. Again, the detailed proofs of these and their related results are deferred to the final section.

\section{Bipartite conditional entropy relations}
Here we detail generalised R\'enyi bipartite uncertainty relations, starting with a slightly more general version of Theorem~\ref{Marcos}.

For a summary of the formalism used for measurements and measured states, see Section~\ref{sec:measurement}.
\begin{theorem}\label{GBUR}
	Let $\Map_X\in\CPTP(A,X)$\footnote{Note we take advantage of the notation $\Map_X(\rho_{AB}) = \rho_{XB}$ for brevity and consistency in quantities such as $H^\uparrow_\beta\left(X|B\right)_{\Map_X(\rho_{AB})} = H^\uparrow_\beta\left(X|B\right)_\rho$} and $\Map_Z\in \CPTP(A,Z)$ be two incompatible measurement maps, defined by the ONBs $\mathbb{X}$ of $\Hil_X$ and $\mathbb{Z}$ of $\Hil_Z$.
	
	For $\alpha,\gamma \geq \frac{1}{2}$ and $\beta>\frac{1}{2}$ such that $\alpha\beta\gamma  -2\beta\gamma - \alpha + \beta +\gamma = 0 $ and $\frac{1}{\beta} + \frac{1}{\gamma} \geq 2$, then for all $\rho_{AB}\in\Den(AB)$ and $\tau_B\in \Pos(B)$
	\begin{equation}\label{MUlike}
		H^\uparrow_\beta\left(X|B\right)_\rho + H_\gamma(\Map_Z(\rho_{AB})\|\tau_B) \geq H_\alpha(\rho_{AB}\|\tau_B)+\qmu.
	\end{equation}
\end{theorem}

We now give R\'enyi extensions of the improved uncertainty relations in~\cite{coles14}. These results also constitute versions of the above theorem with an improved bound which is R\'enyi order dependent.  We first have the state-dependent version:
\begin{theorem}\label{SDGBUR}
	Let $\Map_X\in\CPTP(A,X)$ and $\Map_Z\in \CPTP(A,Z)$ be two incompatible measurement maps, defined by the ONBs $\mathbb{X}$ of $\Hil_X$ and $\mathbb{Z}$ of $\Hil_Z$.
	
	For $\alpha,\gamma\geq\frac{1}{2},\beta>\frac{1}{2},\delta\in \R$ if there exists a $\mu\geq\frac{1}{2}$ such that
	\begin{equation}
		\frac{\alpha-\gamma}{\alpha\gamma-2\gamma + 1} = \frac{\beta-\delta}{\beta\delta-2\delta +1}  = \mu \text{ and }\frac{1}{\delta}\leq2-\frac{1}{\mu}\leq \frac{1}{\gamma},\label{MUconstraints}
	\end{equation}
	then for all $\rho_{AB}\in\Den(AB)$ and $\tau_B\in \Pos(B)$
	\begin{equation}\label{IMUlike}
		H^\uparrow_\beta(X|B)_\rho +H_{\gamma}(\mathcal{M}_Z(\rho_{AB})\|\tau_B)	\geq H_{\alpha}(\rho_{AB}\|\tau_B)  + q_\delta(\rho, \mathbb{X}, \mathbb{Z})
	\end{equation}
	and
	\begin{equation}\label{IMUlike2}
		H_{\gamma}(\mathcal{M}_X(\rho_{AB})\|\tau_B) +H^\uparrow_{\beta}(Z|B)_\rho	\geq H_{\alpha}(\rho_{AB}\|\tau_B)  + q_\delta(\rho, \mathbb{Z}, \mathbb{X}),
	\end{equation}
	where $q_\delta(\rho, \mathbb{X}, \mathbb{Z}) = -\log\left(\tr\rho_{X}\sum_x \max_z(c_{x,z})^\frac{1}{\delta'}\ketbra{x}{x}\right)^{\delta'}$ and $c_{x,z} = \left|\braket{x}{z}\right|^2$.
	
	Moreover, if $\gamma>\frac{1}{2}$ and there exists a $\tilde \mu\geq \frac{1}{2}$ such that
	\begin{equation}
		\frac{\alpha-\beta}{\alpha\beta-2\beta + 1} = \frac{\gamma-\delta}{\gamma\delta-2\delta +1}  = \tilde\mu \text{ and }\frac{1}{\delta}\leq2-\frac{1}{\tilde\mu}\leq \frac{1}{\beta},\label{MUconstraints2}
	\end{equation}
then
	\begin{equation}\label{MUlike3}
		H^\uparrow_\beta(X|B)_\rho +H^\uparrow_{\gamma}(Z|B)_\rho	\geq H_{\alpha}(\rho_{AB}\|\tau_B)  + q_\delta(\rho),
	\end{equation}
	where $q_\delta(\rho) = \max\{q_\delta(\rho, \mathbb{X}, \mathbb{Z}) , q_\delta(\rho, \mathbb{Z}, \mathbb{X})\}$.
\end{theorem}
The derivation of this result could be considered an amalgamation of the structures of the proofs of both Theorems~\ref{GBUR} and~\ref{CPBUR}, possible due to the general comparisons of R\'enyi divergences available from Theorem~\ref{DivFormIntLem}.

We may also establish a weaker, state-independent version by considering the `worst case' state which would achieve the  minimum of the bound.
\begin{theorem}\label{SIGBUR}
	With the same conditions required for Eq.~\eqref{MUlike3}, we have
	\begin{equation}\label{MUlikeSID}
		H^\uparrow_\beta(X|B)_\rho +H^\uparrow_{\gamma}(Z|B)_\rho	\geq H_{\alpha}(\rho_{AB}\|\tau_B)  + q_\delta,
	\end{equation}
	where
	\begin{equation}\label{SIDbound}
		q_\delta = \min_{\sigma \in \Den(AB)}q_\delta(\sigma) = -\min_{0\leq p\leq 1}\log	\lambda_{\max}\left[\Delta_\delta(p)^{\delta'}\right]
	\end{equation}
	with $\Delta_\delta(p) = p\sum_{x} (\max_{z}c_{x,z})^\frac{1}{\delta'}\ketbra{x}{x}+(1-p)\sum_{z} (\max_{x}c_{x,z})^\frac{1}{\delta'}\ketbra{z}{z}$.
\end{theorem}

It is relatively straight-forward to show that these results generalise those given in Section~\ref{sec:improvements-and-extensions}. Indeed, when all R\'enyi parameters go to $1$ we recover the relations of Coles and Piani. Otherwise, when $\delta \rightarrow 0$ we recover Theorem~\ref{GBUR}.

The relationships between the $\delta$ dependent bounds are summarised in the following proposition:
\begin{proposition}\label{constToOne}
	With $q_\delta(\rho, \mathbb{X}, \mathbb{Z})$ defined as in Theorem~\ref{SDGBUR} we have
	\begin{align}
		\lim_{\delta \rightarrow 1}q_\delta(\rho, \mathbb{X}, \mathbb{Z}) &= q(\rho, \mathbb{X}, \mathbb{Z}),\label{dto1} \\
		\lim_{\delta \rightarrow 0}q_\delta(\rho, \mathbb{X}, \mathbb{Z}) &= \qmu.\label{dto0}
	\end{align}
It follows that
\begin{align}
	\lim_{\delta \rightarrow 1}q_\delta(\rho) &= q(\rho), \\
	\lim_{\delta \rightarrow 1}q_\delta &= q_{\textsc{CP}}.
\end{align}\end{proposition}
The proof is mainly an application of l'H\^opital's rule in the same vein as Proposition~\ref{ato1}.

\section{Information exclusion relations}
We may now adapt the above results to derive a R\'enyi extension of the Hall relation, Theorem~\ref{Hall}, and further generalise to the improved bounds of Coles and Piani.
\begin{theorem}\label{result2}
	Let $\Map_X\in\CPTP(A,X)$ and $\Map_Z\in \CPTP(A,Z)$ be two incompatible measurement maps, defined by the ONBs $\mathbb{X}$ of $\Hil_X$ and $\mathbb{Z}$ of $\Hil_Z$.
	
	For $\alpha,\gamma\geq \frac{1}{2}$ and $\beta>\frac{1}{2}$, satisfying both $\alpha\beta\gamma  -2\beta\gamma - \alpha + \beta +\gamma = 0 $ and $\frac{1}{\beta} + \frac{1}{\gamma} \geq 2$, then for all $\rho_{AB}\in\Den(AB)$ and $\tau_B\in \Pos(B)$
	\begin{equation}\label{MICR}
		I^\downarrow_\beta(X:B)_\rho + I_\gamma(\Map_Z(\rho_{AB})\|\tau_B) \leq \rh - H_\alpha(\rho_{AB}\|\tau_B).
	\end{equation}
\end{theorem}
The technique employed in the demonstration of this result is almost identical to the technique used for Theorem~\ref{Hall}. In this case however, the relevant comparisons are possible due to the decomposition and chain rules found in Chapter~\ref{sec:renyi-entropy-divergence-inequalities}.
 
Choosing $\alpha \rightarrow\infty$, we have the following corollary which summarises the possible choices of parameters which produce an optimal inequality.
\begin{corollary}\label{res2c}
	Given the same conditions as Theorem~\ref{result2}, for $\frac{1}{2}<\alpha<2$, we have
	\begin{equation}\label{optIER}
		I^\downarrow_\alpha(X:B)_\rho + I^\uparrow_\frac{1}{\alpha}(B\;;\>\!Z)_\rho \leq \rh-H_{\min}(A|B)_\rho.
	\end{equation}
\end{corollary}

This brings us to our final main result, the R\'enyi generalisation of Theorem~\ref{CPIER}. Note that, compared to the relations in the previous section, the bounds here are not order-dependent nor state-dependent, rather we find that the bounds coincide with those of the Shannon and von Neumann situations.
\begin{theorem}[Improved R\'enyi information exclusion relations]\label{IER}
	Let $\Map_X\in\CPTP(A,X)$ and $\Map_Z\in \CPTP(A,Z)$ be two incompatible measurement maps, defined by the ONBs $\mathbb{X}$ of $\Hil_X$ and $\mathbb{Z}$ of $\Hil_Z$.
	
	Given $\alpha,\gamma\geq\frac{1}{2}$, $\beta \in \left[\frac{1}{2}, 2\right]$ with
	\begin{equation}
		\label{IERconstraints}\frac{\alpha-\gamma}{\alpha\gamma -2\gamma +1} = \frac{1}{2-\beta}\text{ and }\beta\gamma\leq1,
	\end{equation}
	
	then for all $\rho_{AB}\in\Den(AB)$ and $\tau_B\in \Pos(B)$
	\begin{equation}
		I^\downarrow_\beta(X:B)_\rho + I_\gamma(\mathcal{M}_Z(\rho_{AB})\|\tau_B)  \leq  r(\mathbb{X},\mathbb{Z})-H_{\alpha}(\rho_{AB}\|\tau_B)
	\end{equation}
	and
	\begin{equation}
		I_\gamma(\mathcal{M}_X(\rho_{AB})\|\tau_B) + I^\downarrow_\beta(Z:B)_\rho  \leq  r(\mathbb{Z},\mathbb{X})-H_{\alpha}(\rho_{AB}\|\tau_B).
	\end{equation}
	Moreover, if $\gamma\leq 2-\beta$ and
	\begin{align}
		\label{IERconstraints2}\frac{2\alpha-\alpha\gamma-1}{\alpha-\gamma} = \frac{1}{2-\beta}, \quad
		\frac{2\alpha-\alpha\beta-1}{\alpha-\beta} = \frac{1}{2-\gamma}
	\end{align}
	then
	\begin{equation}\label{IEREqSym}
		I^\downarrow_\beta(X:B)_\rho + I^\downarrow_\gamma(Z:B)_\rho  \leq  \rcp-H^\uparrow_{\alpha}(A|B)_\rho.
	\end{equation}
\end{theorem}
The proof follows the structure of the derivation of the von Neumann result in Theorem~\ref{CPIER}, now achievable with a reduced version of the general comparison in Theorem~\ref{DivFormIntLem}.

Similarly, we find an optimal version of the above by choosing $\alpha \rightarrow \infty$.
\begin{corollary}\label{IIERopt}
	Given $\frac{1}{2}\leq\alpha\leq\frac{3}{2}$, then for all $\rho_{AB}\in\Den(AB)$ and $\tau_B\in \Pos(B)$
	\begin{equation}
		I^\downarrow_\alpha(X:B)_\rho + I_{2-\alpha}(\mathcal{M}_Z(\rho_{AB})\|\tau_B)  \leq  r(\mathbb{X},\mathbb{Z})-H_{\min}(A|B)_\rho.
	\end{equation}
	Specifically,
	\begin{equation}
		I^\downarrow_\frac{1}{2}(X:B)_\rho + I^\downarrow_\frac{3}{2}(Z:B)_\rho  \leq  \rcp-H_{\min}(A|B)_{\rho}.
	\end{equation}
\end{corollary}
\section{Proofs}
We now include the detailed proofs of the above results.
\subsection{Proofs of the bipartite conditional entropy relations}
We first demonstrate the generalised form of the bipartite R\'enyi uncertainty relation.
\begingroup
\def\thetheorem{\ref*{GBUR}}
\begin{theorem}
	Let $\Map_X\in\CPTP(A,X)$ and $\Map_Z\in \CPTP(A,Z)$ be two incompatible measurement maps, defined by the ONBs $\mathbb{X}$ of $\Hil_X$ and $\mathbb{Z}$ of $\Hil_Z$.
	
	For $\alpha,\gamma \geq \frac{1}{2}$ and $\beta>\frac{1}{2}$ such that $\alpha\beta\gamma  -2\beta\gamma - \alpha + \beta +\gamma = 0 $ and $\frac{1}{\beta} + \frac{1}{\gamma} \geq 2$, then for all $\rho_{AB}\in\Den(AB)$ and $\tau_B\in \Pos(B)$
	\begin{equation}\label{MUlikeIL}
		H^\uparrow_\beta\left(X|B\right)_\rho + H_\gamma(\Map_Z(\rho_{AB})\|\tau_B) \geq H_\alpha(\rho_{AB}\|\tau_B)+\qmu.
	\end{equation}
\end{theorem}
\addtocounter{theorem}{-1}
\endgroup

Before we treat the proof of Theorem~\ref{GBUR} we first introduce a specific form of the Stinespring dilation~\cite{Stine}.
\begin{definition}[Stinespring dilation]
	A map $\Map \in \CPTP(A, B)$ if and only if there exists an isometry $U\in \Lin(A, BC)$ such that $\Map(\rho) = \tr_C(U\rho U^\dagger) \text{ for all } \rho\in\Den(A)$.
\end{definition}
\begin{proof}[Proof of Theorem~\ref{GBUR}]
	Let $\mathcal{S}_Z\in \CPTP(A, ZZ')$ be a Stinespring dilation of $\Map_Z$ such that
	\begin{equation}
		\mathcal{S}_Z(\rho_A) = \sum_{z, z'}\bra{z}\rho_A\ket{z'}\ketbra{z}{z'}\otimes\ketbra{z}{z'}.
	\end{equation}
	We use the same argument as the proof of~\cite[Theorem~7.6]{mybook}, to arrive at
	\begin{equation}
		H^\uparrow_\beta\left(X|B\right)_\rho\geq H^\uparrow_\beta\left(Z|Z'B\right)_{\mathcal{S}_Z(\rho)} +\qmu.\label{initineq}
	\end{equation}
	The two main components of this argument are the comparisons: for $\beta\geq \frac{1}{2}$,
	\begin{align}
		\label{comp1}H^\uparrow_\beta\left(Z|Z'B\right)_{\mathcal{S}_Z(\rho)}&\leq -\inf_{\sigma_{Z'B}\in \Den(Z'B)}D_\beta\left(\Map_X(\rho_{AB})\|\Map_X\left(\mathcal{S}_Z^{-1}(\id_{Z}\otimes\sigma_{Z'B})\right)\right)\quad\text{and}\\
		\label{comp2}\Map_X\left(\mathcal{S}_Z^{-1}(\id_{Z}\otimes\sigma_{Z'B})\right) &= \sum_{x,z}\left|\braket{x}{z}\right|^2\ketbra{x}{x}\otimes\bra{z'}\sigma_{Z'B}\ket{z'}\leq c\id_X\otimes\sigma_B.
	\end{align}
The first comparison is a result of the data-processing inequality~(see Proposition~\ref{DPI}) and for the second we maximise $\left|\braket{x}{z}\right|$ over $x$ and $z$.

	Substituting Eq.~\eqref{comp2} into Eq.~\eqref{comp1} yields Eq.~\eqref{initineq}.\\
	
	Let $\alpha\beta\gamma  -2\beta\gamma - \alpha + \beta +\gamma = 0 $ with $\alpha,\gamma \geq \frac{1}{2}$, $\beta>\frac{1}{2}$ and $\frac{1}{\beta} + \frac{1}{\gamma} \geq 2$. Then by Theorem~\ref{chain} we can write for all $\tau_B \in \Den(B)$,
	\begin{equation}
	H^\uparrow_\beta(Z|Z'B)_\rho	\geq  H_\alpha(\mathcal{S}_Z(\rho_{AB})\|\tau_B) - H_\gamma(\tr_Z(\mathcal{S}_Z(\rho_{AB}))\|\tau_B).\label{genchain}
	\end{equation}Substituting Eq.~\eqref{genchain} into Eq.~\eqref{initineq} we have
	\begin{equation}
		H^\uparrow_\beta\left(X|B\right)_\rho\geq H_\alpha(\mathcal{S}_Z(\rho_{AB})\|\tau_B) - H_\gamma(\tr_Z(\mathcal{S}_Z(\rho_{AB}))\|\tau_B)+\qmu.
	\end{equation}
	Using the fact that the marginals on $ZB$ and $Z'B$ of the state $\mathcal{S}_Z(\rho_{AB})$ are equivalent  and that the conditional entropies are invariant under local isometries we obtain Eq.~\eqref{MUlikeIL}.
\end{proof}

The following proof Theorem~\ref{SDGBUR} has the same broad strokes as the previous proof, the main difference being a tighter comparison when applying the measurement in $\mathbb{X}$.
\begingroup
\def\thetheorem{\ref*{SDGBUR}}
\begin{theorem}
	Let $\Map_X\in\CPTP(A,X)$ and $\Map_Z\in \CPTP(A,Z)$ be two incompatible measurement maps, defined by the ONBs $\mathbb{X}$ of $\Hil_X$ and $\mathbb{Z}$ of $\Hil_Z$.
	
	For $\alpha,\gamma\geq\frac{1}{2},\beta>\frac{1}{2},\delta\in \R$ if there exists a $\mu\geq\frac{1}{2}$ such that
	\begin{equation}
		\frac{\alpha-\gamma}{\alpha\gamma-2\gamma + 1} = \frac{\beta-\delta}{\beta\delta-2\delta +1}  = \mu \text{ and }\frac{1}{\delta}\leq2-\frac{1}{\mu}\leq \frac{1}{\gamma},\label{MUconstraintsIL}
	\end{equation}
	then for all $\rho_{AB}\in\Den(AB)$ and $\tau_B\in \Pos(B)$
	\begin{equation}\label{IMUlikeIL}
		H^\uparrow_\beta(X|B)_\rho +H_{\gamma}(\mathcal{M}_Z(\rho_{AB})\|\tau_B)	\geq H_{\alpha}(\rho_{AB}\|\tau_B)  + q_\delta(\rho, \mathbb{X}, \mathbb{Z})
	\end{equation}
	and
	\begin{equation}\label{IMUlike2IL}
		H_{\gamma}(\mathcal{M}_X(\rho_{AB})\|\tau_B) +H^\uparrow_{\beta}(Z|B)_\rho	\geq H_{\alpha}(\rho_{AB}\|\tau_B)  + q_\delta(\rho, \mathbb{Z}, \mathbb{X}),
	\end{equation}
	where $q_\delta(\rho, \mathbb{X}, \mathbb{Z}) = -\log\left(\tr\rho_{X}\sum_x \max_z(c_{x,z})^\frac{1}{\delta'}\ketbra{x}{x}\right)^{\delta'}$ and $c_{x,z} = \left|\braket{x}{z}\right|^2$.
	
	Moreover, if $\gamma>\frac{1}{2}$ and there exists a $\tilde \mu\geq \frac{1}{2}$ such that
	\begin{equation}
		\frac{\alpha-\beta}{\alpha\beta-2\beta + 1} = \frac{\gamma-\delta}{\gamma\delta-2\delta +1}  = \tilde\mu \text{ and }\frac{1}{\delta}\leq2-\frac{1}{\tilde\mu}\leq \frac{1}{\beta},\label{MUconstraints2IL}
	\end{equation}
	then
	\begin{equation}\label{MUlike3IL}
		H^\uparrow_\beta(X|B)_\rho +H^\uparrow_{\gamma}(Z|B)_\rho	\geq H_{\alpha}(\rho_{AB}\|\tau_B)  + q_\delta(\rho),
	\end{equation}
	where $q_\delta(\rho) = \max\{q_\delta(\rho, \mathbb{X}, \mathbb{Z}) , q_\delta(\rho, \mathbb{Z}, \mathbb{X})\}$.
\end{theorem}
\addtocounter{theorem}{-1}
\endgroup
We use a similar argument to the proof of Theorem~\ref{GBUR}, but following the structure given in the proof of~\cite[Theorem 2]{coles14}.
\begin{proof}
	Let $\mathcal{S}_Z\in \CPTP(A, ZZ')$ be a Stinespring dilation of $\Map_Z$ such that
	\begin{equation}
		\mathcal{S}_Z(\rho_A) = \sum_{z, z'}\bra{z}\rho_A\ket{z'}\ketbra{z}{z'}\otimes\ketbra{z}{z'}.
	\end{equation}
	
	By the data-processing inequality:
	\begin{align}
		H^\uparrow_\mu(Z|Z'B)_{\mathcal{S}_Z(\rho)} & \leq -\inf_{\sigma_{Z'B}\in \Den(Z'B)}D_\mu(\rho_{AB}\|\mathcal{S}_Z^{-1}(\id_Z\otimes \sigma_{Z'B}))                                                         \\
		& \leq-\inf_{\sigma_{Z'B}\in \Den(Z'B)}D_\mu\left(\Map_X(\rho_{AB})\|\Map_X\left(\mathcal{S}_Z^{-1}(\id_{Z}\otimes\sigma_{Z'B})\right)\right).
	\end{align}
	We may compare
	\begin{align}
		\Map_X\left(\mathcal{S}_Z^{-1}(\id_{Z}\otimes\sigma_{Z'B})\right) & = \sum_{x,z}\left|\braket{x}{z}\right|^2\ketbra{x}{x}\otimes\bra{z'}\sigma_{Z'B}\ket{z'} \\
		& \leq \underbrace{\sum_x \max_z(c_{x,z})\ketbra{x}{x}}_{\omega_X}\otimes\sigma_B.
	\end{align}
Therefore,
\begin{equation}
	H^\uparrow_\mu(Z|Z'B)_{\mathcal{S}_Z(\rho)}  \leq \sup_{\sigma_{B}\in \Den(B)}-D_\mu(\mathcal{M}_X(\rho_{AB})\|\omega_X\otimes\sigma_B).\label{firstComp}
\end{equation}
	By Theorem~\ref{DivFormIntLem} we conclude for $\beta\geq \frac{1}{2}, \mu> 0$, $\delta\in \R$ with $\beta\mu\delta -2\mu\delta - \beta +\mu +\delta=0$ and $\frac{1}{\mu} + \frac{1}{\delta} \leq 2$,
	\begin{equation}
		-H_\beta(\mathcal{M}_X(\rho_{AB})\|\sigma_B) \leq D_\mu(\mathcal{M}_X(\rho_{AB})\|\omega_X\otimes\sigma_B)+\log\tr\left(\rho_{X}\omega_X^\frac{1}{\delta'}\right)^{\delta'}.
	\end{equation}
	Substituting this into Eq.~\eqref{firstComp} we obtain
	\begin{equation}
		H^\uparrow_\mu(Z|Z'B)_{\mathcal{S}_Z(\rho)}\leq \sup_{\sigma_{B}\in \Den(B)}H_\beta(\mathcal{M}_X(\rho_{AB})\|\sigma_B)+\log\tr\left(\rho_{X}\omega_X^\frac{1}{\delta'}\right)^{\delta'}.
	\end{equation}
	By Theorem~\ref{chain}: for $\alpha, \mu, \gamma \geq \frac{1}{2}$ such that $\alpha\mu\gamma -2\mu\gamma - \alpha + \mu + \gamma = 0$ and $\frac{1}{\mu} + \frac{1}{\gamma}\geq2$
	\begin{equation}
		H_{\alpha}(\mathcal{S}_Z(\rho_{AB})\|\tau_B)-H_{\gamma}(\tr_Z(\mathcal{S}_Z(\rho_{AB}))\|\tau_B)	\leq H^\uparrow_\mu(Z|Z'B)_{\mathcal{S}_Z(\rho)}
	\end{equation}
	for all $\tau_B\in \Den(B)$. Hence
	\begin{equation}
		H_{\alpha}(\mathcal{S}_Z(\rho_{AB})\|\tau_B) - H_{\gamma}(\tr_Z(\mathcal{S}_Z(\rho_{AB})\|\tau_B)	\leq H^\uparrow_\beta(X|B)_\rho+\log\tr\left(\rho_X\omega_X^\frac{1}{\delta'}\right)^{\delta'}.
	\end{equation}
	As in the proof of Theorem~\ref{GBUR}, the marginals on $ZB$ and $ZB'$ of the state $\mathcal{S}_Z(\rho_{AB})$ are equivalent and the conditional entropies are invariant under local isometries, therefore we obtain Eq.~\eqref{IMUlikeIL}.
	
	For the conditions on the parameters note that we can express $\mu$ in terms of $\beta$ and $\delta$:
	\begin{align}
		\beta\mu\delta -2\mu\delta - \beta +\mu +\delta&=0\\
		\implies \mu(\beta\delta - 2\delta + 1) &= \beta - \delta\\
		\implies \mu = \frac{\beta - \delta}{\beta\delta - 2\delta + 1},
	\end{align}
	and similarly for $\alpha$ and $\gamma$.
	
	Moreover the conditions determining the direction of the inequalities can be combined via $\mu$:
	\begin{align}
		\frac{1}{\mu} + \frac{1}{\delta} &\leq 2\\
		\implies \frac{1}{\delta} &\leq 2-\frac{1}{\mu} \leq \frac{1}{\gamma}.
	\end{align}
\end{proof}
This leads us to the state-independent version:
\begingroup
\def\thetheorem{\ref*{SIGBUR}}
\begin{theorem}
	With the same conditions required for Eq.~\eqref{MUlike3IL}, we have
	\begin{equation}\label{MUlikeSIDIL}
		H^\uparrow_\beta(X|B)_\rho +H^\uparrow_{\gamma}(Z|B)_\rho	\geq H_{\alpha}(\rho_{AB}\|\tau_B)  + q_\delta,
	\end{equation}
	where
	\begin{equation}\label{SIDboundIL}
		q_\delta = \min_{\sigma \in \Den(AB)}q_\delta(\sigma) = -\min_{0\leq p\leq 1}\log	\lambda_{\max}\left[\Delta_\delta(p)^{\delta'}\right]
	\end{equation}
	with $\Delta_\delta(p) = p\sum_{x} (\max_{z}c_{x,z})^\frac{1}{\delta'}\ketbra{x}{x}+(1-p)\sum_{z} (\max_{x}c_{x,z})^\frac{1}{\delta'}\ketbra{z}{z}$.
\end{theorem}
\addtocounter{theorem}{-1}
\endgroup

\begin{proof}
	We may rewrite
	\begin{align}
		q_\delta(\rho)                       & =\max_{0\leq p\leq 1}\left(-\log\left[\tr\rho_A\Delta_\delta(p)\right]^{\delta'} \right) \\
		\implies \min_{\sigma \in \Den(AB)}q_\delta(\sigma) & =
		-\log\max_{\sigma \in \Den(AB)}\min_{0\leq p\leq 1}\left[\tr\sigma_A\Delta_\delta(p)\right]^{\delta'}.
	\end{align}
	By the linearity in the arguments we may use the minimax theorem~\citep{v.Neumann1928} to swap the optimisations, hence we obtain Eq.~\eqref{SIDboundIL}.
\end{proof}
We conclude this section with the summary of the $\delta$-dependent bounds:
\begingroup
\def\thetheorem{\ref*{constToOne}}
\begin{proposition}
	With $q_\delta(\rho, \mathbb{X}, \mathbb{Z})$ defined as in Theorem~\ref{SDGBUR} we have
	\begin{align}
		\lim_{\delta \rightarrow 1}q_\delta(\rho, \mathbb{X}, \mathbb{Z}) &= q(\rho, \mathbb{X}, \mathbb{Z}),\label{dto1IL} \\
		\lim_{\delta \rightarrow 0}q_\delta(\rho, \mathbb{X}, \mathbb{Z}) &= \qmu.\label{dto0IL}
	\end{align}
	It follows that
	\begin{align}
		\lim_{\delta \rightarrow 1}q_\delta(\rho) &= q(\rho), \\
		\lim_{\delta \rightarrow 1}q_\delta &= q_{\textsc{CP}}.
\end{align}\end{proposition}
\addtocounter{theorem}{-1}
\endgroup

\begin{proof}
	Let $\{\lambda_x\}_x$ be the eigenvalues of $\rho_X$, i.e.
	\begin{equation}
		\rho_X = \sum_x \lambda_x\ketbra{x}{x}.
	\end{equation}
	Note therefore, that $\rho_X$ and $\omega_X$ are diagonal in the same basis and we can write
	\begin{equation}
		q_\delta(\rho, \mathbb{X}, \mathbb{Z}) = -\log\tr\left(\rho_X\sum_x\left(\max_z c_{x,z}\right)^\frac{1}{\delta'}\ketbra{x}{x}\right)^{\delta'} = -\delta'\log\sum_x \lambda_x\left(\max_z c_{x,z}\right)^\frac{1}{\delta'}.
	\end{equation}
	We use l'H\^opital's rule, choosing
	\begin{gather}
		f(y) = \log\sum_x \lambda_x\left(\max_z c_{x,z}\right)^\frac{1}{y} \implies \lim_{y\rightarrow \infty}f(y) = 0,\\
		g(y) = \frac{1}{y} \implies \lim_{y\rightarrow \infty}g(y) = 0.
	\end{gather}
	We have
	\begin{align}
		f'(y) & = \frac{d}{dy}\left(\log\sum_x \lambda_x\left(\max_z c_{x,z}\right)^\frac{1}{y} \right)                                                      \\
		& = \frac{\sum_x \lambda_x\frac{d}{dy}\left(\left(\max_z c_{x,z}\right)^\frac{1}{y}\right)}{\sum_x \lambda_x\left(\max_z c_{x,z}\right)^\frac{1}{y}} \\
		& = \frac{-\sum_x \lambda_x\left(\max_z c_{x,z}\right)^\frac{1}{y}\log\max_z c_{x,z}}{y^2\sum_x \lambda_x\left(\max_z c_{x,z}\right)^\frac{1}{y}}
	\end{align}
	and $g'(y) = \frac{-1}{y^2}$. Hence
	\begin{equation}
		\lim_{y\rightarrow \infty} \frac{f(y)}{g(y)} = \lim_{y\rightarrow \infty} \frac{f'(y)}{g'(y)} =\lim_{y\rightarrow \infty}\frac{\sum_x \lambda_x\left(\max_z c_{x,z}\right)^\frac{1}{y}\log\max_z c_{x,z}}{\sum_x \lambda_x\left(\max_z c_{x,z}\right)^\frac{1}{y}} = \sum_x \lambda_x\log\max_zc_{x,z}.
	\end{equation}
	Observe that $\delta' \rightarrow \infty \implies \delta \rightarrow 1$, therefore
	\begin{equation}
		\lim_{\delta \rightarrow 1}q_\delta(\rho, \mathbb{X}, \mathbb{Z}) = \sum_x \lambda_x\log(1/\max_zc_{x,z}).
	\end{equation}
	Moreover,
$
		\lim_{y\rightarrow 0}f(y) = \infty$ and $
		\lim_{y\rightarrow 0}g(y) = \infty
$.
	So, again by l'H\^opital's rule,
	\begin{equation}
		\lim_{y\rightarrow 0} \frac{f(y)}{g(y)} =\lim_{y\rightarrow 0}\frac{\sum_x \lambda_x\left(\max_z c_{x,z}\right)^\frac{1}{y}\log\max_z c_{x,z}}{\sum_x \lambda_x\left(\max_z c_{x,z}\right)^\frac{1}{y}} =\log\max_{x,z} c_{x,z} = -\qmu.
	\end{equation}
The last two statements of the proposition are evident from Eq.~\eqref{dto1IL} and the definitions of the relevant quantities.
\end{proof}
\subsection{Proofs of the R\'enyi information exclusion relations}
We begin with the unimproved R\'enyi information exclusion relations, which more directly generalise the Hall relation.
\begingroup
\def\thetheorem{\ref*{result2}}
\begin{theorem}
	Let $\Map_X\in\CPTP(A,X)$ and $\Map_Z\in \CPTP(A,Z)$ be two incompatible measurement maps, defined by the ONBs $\mathbb{X}$ of $\Hil_X$ and $\mathbb{Z}$ of $\Hil_Z$.
	
	Let $\alpha,\gamma\geq \frac{1}{2}$ and $\beta>\frac{1}{2}$, satisfying both $\alpha\beta\gamma  -2\beta\gamma - \alpha + \beta +\gamma = 0 $ and $\frac{1}{\beta} + \frac{1}{\gamma} \geq 2$, then for all $\rho_{AB}\in\Den(AB)$ and $\tau_B\in \Pos(B)$
	\begin{equation}\label{MICRIL}
		I^\downarrow_\beta(X:B)_\rho + I_\gamma(\Map_Z(\rho_{AB})\|\tau_B) \leq \rh - H_\alpha(\rho_{AB}\|\tau_B).
	\end{equation}
\end{theorem}
\addtocounter{theorem}{-1}
\endgroup

\begin{proof}[Proof of Theorem~\ref{result2}]
	Starting with Eq.~\eqref{MUlikeIL}, and choosing parameters which satisfy the conditions, we can write
	\begin{equation}
		H^\uparrow_{\bar\beta}(X|B)_\rho + H_{\bar\gamma}(\Map_Z(\rho_{AB})\|\tau_B) \geq H_\alpha(\rho_{AB}\|\tau_B)+\qmu.
	\end{equation}
	For each conditional entropy on the left-hand side we can derive the following inequalities from Eq.~\eqref{decompEq2}:
	\begin{align}
		H^\uparrow_{\bar\beta}(X|B)_{\rho}&\leq H_{\tilde{\beta}}(\rho_X)-I^\downarrow_\beta(X:B)_\rho,\\
		H_{\bar\gamma}(\Map_Z(\rho_{AB})\|\tau_B)&\leq H_{\tilde\gamma}(\rho_Z)-I_\gamma(\Map_Z(\rho_{AB})\|\tau_B).
	\end{align}
	We can then write
	\begin{align}
		-I^\downarrow_\beta(X:B)_\rho -I_\gamma(\Map_Z(\rho_{AB})\|\tau_B) &\geq H_\alpha(\rho_{AB}\|\tau_B)- H_{\tilde{\beta}}(\rho_X) -  H_{\tilde\gamma}(\rho_Z)+\qmu\\
		\implies I^\downarrow_\beta(X:B)_\rho + I_\gamma(\Map_Z(\rho_{AB})\|\tau_B) &\leq H_{\tilde{\beta}}(\rho_X)+H_{\tilde\gamma}(\rho_Z) - H_\alpha(\rho_{AB}\|\tau_B)-\qmu\\
		&\leq \log(d^2c) - H_\alpha(\rho_{AB}\|\tau_B).
	\end{align}
	The last line is due to $H_{\alpha}(\rho)\leq \log d$ for all $\alpha$.
	
	We can optimise the parameters when $\tilde\beta, \tilde\gamma = 0$, which implies $\bar\beta=\beta$ and $\bar\gamma = \gamma$. We therefore have the familiar condition, $\alpha\beta\gamma  -2\beta\gamma - \alpha + \beta +\gamma = 0 $. Finally, note that for all $0\leq \tilde\beta\leq \frac{1}{2} \implies \frac{1}{\tilde\beta} \geq 2$, hence we also have the condition $\frac{1}{\beta}+\frac{1}{\gamma}\geq 2$.
\end{proof}
The following optimal case then follows:
\begingroup
\def\thetheorem{\ref*{res2c}}
\begin{corollary}
	Given the same conditions as Theorem~\ref{result2}, for $\frac{1}{2}<\alpha<2$, we have
	\begin{equation}\label{optIERIL}
		I^\downarrow_\alpha(X:B)_\rho + I^\uparrow_\frac{1}{\alpha}(B\;;\>\!Z)_\rho \leq \rh-H_{\min}(A|B)_\rho.
	\end{equation}
\end{corollary}
\addtocounter{theorem}{-1}
\endgroup

\begin{proof}
	Taking the limit
	\begin{equation}
		\lim_{\alpha\rightarrow\infty}\beta\gamma  -\frac{2\beta\gamma}{\alpha} - 1 + \frac{\beta}{\alpha} +\frac{\gamma}{\alpha} = \beta\gamma - 1,
	\end{equation}
we may conclude $\beta = \frac{1}{\gamma}$. 
	We then let $\tau_B = \rho_B$ to have a tighter inequality, obtaining Eq.~\eqref{optIERIL}.
\end{proof}
This brings us to our final theorem, determining the improved R\'enyi information exclusion relations.
\begingroup
\def\thetheorem{\ref*{IER}}
\begin{theorem}
	Let $\Map_X\in\CPTP(A,X)$ and $\Map_Z\in \CPTP(A,Z)$ be two incompatible measurement maps, defined by the ONBs $\mathbb{X}$ of $\Hil_X$ and $\mathbb{Z}$ of $\Hil_Z$.
	
	Given $\alpha,\gamma\geq\frac{1}{2}$, $\beta \in [1/2, 2]$ with
	\begin{equation}
		\label{IERconstraintsIL}\frac{\alpha-\gamma}{\alpha\gamma -2\gamma +1} = \frac{1}{2-\beta}\text{ and }\beta\gamma\leq1,
	\end{equation}
	
	then for all $\rho_{AB}\in\Den(AB)$ and $\tau_B\in \Pos(B)$
	\begin{equation}
		I^\downarrow_\beta(X:B)_\rho + I_\gamma(\mathcal{M}_Z(\rho_{AB})\|\tau_B)  \leq  r(\mathbb{X},\mathbb{Z})-H_{\alpha}(\rho_{AB}\|\tau_B)
	\end{equation}
	and
	\begin{equation}
		I_\gamma(\mathcal{M}_X(\rho_{AB})\|\tau_B) + I^\downarrow_\beta(Z:B)_\rho  \leq  r(\mathbb{Z},\mathbb{X})-H_{\alpha}(\rho_{AB}\|\tau_B).
	\end{equation}
	Moreover, if $\gamma\leq 2-\beta$ and
	\begin{align}
		\label{IERconstraints2IL}\frac{2\alpha-\alpha\gamma-1}{\alpha-\gamma} = \frac{1}{2-\beta}, \quad
		\frac{2\alpha-\alpha\beta-1}{\alpha-\beta} = \frac{1}{2-\gamma}
	\end{align}
	then
	\begin{equation}\label{IEREqSymIL}
		I^\downarrow_\beta(X:B)_\rho + I^\downarrow_\gamma(Z:B)_\rho  \leq  \rcp-H^\uparrow_{\alpha}(A|B)_\rho.
	\end{equation}
\end{theorem}
\addtocounter{theorem}{-1}
\endgroup
Before the proof of this result, we first show how to derive a state-independent bound analogous to that in~\cite{coles14}.
\begin{lemma}\label{constCompLemma}
	For $\alpha \geq \frac{1}{2}$, $\delta \in(-\infty, \alpha)$ such that $\frac{\alpha -\delta}{\alpha\delta-2\delta + 1}\geq \frac{1}{2}$ and $\frac{(\alpha-1)(\delta-1)}{\alpha-\delta} + \frac{1}{\delta} \geq 1$,
	\begin{equation}
		H_\alpha(\rho_X)-q_\delta(\rho,\mathbb{X},\mathbb{Z}) \leq \log\sum_x \max_{z} c_{x,z}.\label{constComp}
	\end{equation}
	In particular
	\begin{equation}
		H_\frac{1}{2}(\rho_X)-\qmu \leq \log\sum_x \max_{z} c_{x,z}.\label{constCompOpt}
	\end{equation}
\end{lemma}
\begin{remark}
	The conditions on $\alpha$ and $\delta$ in the above lemma are derived from certain conditions on a $\beta$ parameter that appears in the proof but is not required in the final statement.
	
	Indeed, including this parameter provides an arguably nicer form of the condition: For $\alpha \geq \beta, \delta$ such that $\beta \geq \frac{1}{2}$ and $\frac{1}{\beta} + \frac{1}{\delta}\geq 2$. However, in order to present the lemma as a more self-contained statement, the $\beta$ has been omitted.
\end{remark}
\begin{proof}
	Note that the trace is a CPTP, thus by the data-processing inequality~(see Proposition~\ref{DPI}) for all $\beta\geq\frac{1}{2}$,
	\begin{align}
		D_\beta(\rho_A\|\sigma_A) & \geq D_\beta(\tr\rho_A\|\tr\sigma_A)                                                   \\
		& =\log\left[(\tr\sigma_A)^\frac{-1}{2\beta'}(\tr\sigma_A)^\frac{-1}{2\beta'}\right]^{\beta'} \\
		& =-\log\tr\sigma_A.\label{divComp}
	\end{align}
	
	To derive Eq.~\eqref{constComp} we consider a similar interpolation as in Theorem~\ref{DivFormIntLem}. By choosing $B$ to be trivial and given $\alpha, \beta \geq \frac{1}{2}, \delta \in \R$ such that $\alpha\beta\delta -2\beta\delta -\alpha + \beta +\delta = 0$ and $\frac{1}{\beta} + \frac{1}{\delta}\geq 2$ we may conclude
	\begin{equation}
		H_\alpha(\rho_X) + \log\tr\left(\rho_X\omega_X^\frac{1}{\delta'}\right)^{\delta'}\leq -D_\beta(\rho_X\|\omega_X).
	\end{equation}
	Let $\omega_X = \sum_x \max_z(c_{x,z})\ketbra{x}{x}$, then by Eq.~\eqref{divComp} we can bound the left-hand side
	\begin{equation}
		H_\alpha(\rho_X) + \log\tr\left(\rho_X\omega_X^\frac{1}{\delta'}\right)^{\delta'} \leq \log\tr\omega_X = \log\sum_x \max_z(c_{x,z}),
	\end{equation} 
	hence we obtain Eq.~\eqref{constComp}.
	
	As in the proof of Theorem~\ref{GBUR}, we have $\frac{\alpha -\delta}{\alpha\delta-2\delta + 1} = \beta$. Since we can choose any $\beta \geq \frac{1}{2}$ we have the condition in the lemma. Moreover
	\begin{align}
		\frac{1}{\beta} + \frac{1}{\delta} &\geq 2\\
		\implies \frac{\alpha\delta-2\delta + 1}{\alpha -\delta} + \frac{1}{\delta} &\geq 2\\
		\implies \frac{\alpha\delta - 2\delta + 1 + \alpha - \delta - \alpha + \delta}{\alpha -\delta} + \frac{1}{\delta} &\geq 2\\
		\implies \frac{\alpha\delta - \delta - \alpha + 1 }{\alpha -\delta} + 1 + \frac{1}{\delta} &\geq 2\\
		\implies \frac{(\alpha-1)(\delta-1)}{\alpha - \delta} +\frac{1}{\delta} &\geq 1.
	\end{align}
	
	To achieve the optimal situation in Eq.~\eqref{constCompOpt} we note the R\'enyi divergence is monotonically increasing in its parameter~(see Eq.~\eqref{monot}) hence the choice which provides the tightest inequality in Eq.~\eqref{divComp} is $\beta = \frac{1}{2}$. Eq.~\eqref{aGreat} requires $\alpha \geq \beta$ so, again through monotonicity, $\alpha = \frac{1}{2}$ yields the tightest inequality in Eq.~\eqref{constComp}. Substituting these values into our conditions we have
	\begin{align}
		\delta = \frac{\alpha -\beta}{\alpha\beta-2\beta + 1} = \frac{\frac{1}{2} - \frac{1}{2}}{\frac{1}{4} - 1 + 1} = 0
	\end{align}
	and
	\begin{align}
		\lim_{\delta\rightarrow 0^+}\frac{1}{\beta} + \frac{1}{\delta} = \lim_{\delta\rightarrow 0^+}2 + \frac{1}{\delta} = \infty \geq 2.
	\end{align}
\end{proof}

We now have all the tools required to prove Theorem~\ref{IER}.
\begin{proof}[Proof of Theorem~\ref{IER}]
	We begin with Eq.~\eqref{MUlikeIL} and apply Theorem~\ref{decomp} to the measured conditional entropies:
	\begin{align}
		H^\uparrow_{\bar\beta}(X|B)_\rho +H_{\bar\gamma}(\mathcal{M}_Z(\rho_{AB})\|\tau_B)  & \geq H_{\alpha}(\rho_{AB}\|\tau_B)  + q_\delta(\rho, \mathbb{X}, \mathbb{Z})                                                  \\
		\implies  I^\downarrow_\beta(X:B)_\rho + I_\gamma(\mathcal{M}_Z(\rho_{AB})\|\tau_B) & \leq H_{\tilde \beta}(\rho_X) + H_{\tilde \gamma}(\rho_Z) -H_{\alpha}(\rho_{AB}\|\tau_B)-q_\delta(\rho, \mathbb{X}, \mathbb{Z}),
	\end{align}
	where $\bar\beta, \beta, \bar\gamma, \gamma \geq \frac{1}{2}$ and $\tilde\beta,\tilde\gamma\geq 0$ such that $\bar\beta\beta\tilde\beta -2\beta\tilde\beta -\bar\beta + \beta +\tilde\beta =0$, $\bar\gamma\gamma\tilde\gamma -2\gamma\tilde\gamma -\bar\gamma + \gamma +\tilde\gamma =0$ and $\frac{1}{\beta} + \frac{1}{\tilde\beta}, \frac{1}{\gamma} + \frac{1}{\tilde\gamma}\geq 2$.
	
	We choose $\tilde\beta = \frac{1}{2}, \tilde\gamma \rightarrow 0$ and $\delta \rightarrow 0$. 
	Hence we can use Eq.~\eqref{constCompOpt} and the fact that ${H_{\tilde\gamma}(\rho_Z)\leq \log d}$ for all $\tilde\gamma>0$ to determine
	\begin{align}
		I^\downarrow_\beta(X:B)_\rho + I_\gamma(\mathcal{M}_Z(\rho_{AB})\|\tau_B) & \leq  \log\left(d\sum_x \max_{z} c_{x,z}\right)-H_{\alpha}(\rho_{AB}\|\tau_B).
	\end{align}
	From the above choices we deduce
	\begin{align}
		\bar\beta\beta\tilde\beta -2\beta\tilde\beta -\bar\beta + \beta +\tilde\beta &=0\\
		\frac{\bar\beta\beta}{2} - \beta - \bar\beta + \beta + \frac{1}{2} &= 0\\
		\bar\beta\beta - 2\bar\beta &= -1\\
		\bar\beta &= \frac{1}{2-\beta},
	\end{align}
	\begin{align}
		\lim_{\tilde\gamma\rightarrow 0}\bar\gamma\gamma\tilde\gamma -2\gamma\tilde\gamma -\bar\gamma + \gamma +\tilde\gamma &=0\\
		\bar\gamma - \gamma &= 0\\
		\bar\gamma &= \gamma.
	\end{align}
	Hence we can write
	\begin{align}
		\frac{\alpha-\bar\gamma}{\alpha\bar\gamma-2\bar\gamma + 1} = \frac{\bar\beta-\delta}{\bar\beta\delta-2\delta +1}  &= \mu\\
		\implies \frac{\alpha-\gamma}{\alpha\gamma-2\gamma + 1} = \frac{1}{2-\beta}  &= \mu.
	\end{align}
	Moreover,
	\begin{align}
		\lim_{\delta\rightarrow 0^-}\frac{1}{\delta}\leq2-\frac{1}{\mu}&\leq \frac{1}{\bar\gamma} = \frac{1}{\gamma}\\
		\implies 2-2+\beta &\leq \frac{1}{\gamma}\\
		\implies \beta\gamma &\leq 1.
	\end{align}
	Noting $\mu\geq \frac{1}{2}$ we have the conditions on the first part of the theorem.
	
	For the second part of the theorem, we perform the same procedure but instead start with $q_\delta(\rho, \mathbb{Z}, \mathbb{X})$ and choose both $\tilde\beta, \tilde\gamma = \frac{1}{2}$. This yields
	\begin{align}
		\frac{\alpha-\bar\gamma}{\alpha\bar\gamma-2\bar\gamma + 1}&=\frac{\alpha-\frac{1}{2-\gamma}}{\alpha\frac{1}{2-\gamma}-2\frac{1}{2-\gamma} + 1}\\ &=\frac{2\alpha - \alpha\gamma - 1}{\alpha - 2 + 2 - \gamma}\\
		&= \frac{2\alpha - \alpha\gamma - 1}{\alpha - \gamma}
	\end{align}
	and
	\begin{align}
		2-\frac{1}{\mu}&\leq \frac{1}{\bar\gamma}\\
		\implies \beta &\leq 2-\gamma\\
		\implies \beta + \gamma &\leq 2.
	\end{align}
	Then, by choosing the minimum over the order of the measurements, we obtain Eq.~\eqref{IEREqSymIL}.
\end{proof}
This leaves us with the proof of the optimal version of the above relations.
\begingroup
\def\thetheorem{\ref*{IIERopt}}
\begin{corollary}
	Given $\frac{1}{2}\leq\alpha\leq\frac{3}{2}$, then for all $\rho_{AB}\in\Den(AB)$ and $\tau_B\in \Pos(B)$
	\begin{equation}
		I^\downarrow_\alpha(X:B)_\rho + I_{2-\alpha}(\mathcal{M}_Z(\rho_{AB})\|\tau_B)  \leq  r(\mathbb{X},\mathbb{Z})-H_{\min}(A|B)_\rho.
	\end{equation}
	Specifically,
	\begin{equation}
		I^\downarrow_\frac{1}{2}(X:B)_\rho + I^\downarrow_\frac{3}{2}(Z:B)_\rho  \leq  \rcp-H_{\min}(A|B)_{\rho}.
	\end{equation}
\end{corollary}
\addtocounter{theorem}{-1}
\endgroup
\begin{proof}
	Taking the limit $\alpha \rightarrow \infty$ in Eq.~\eqref{IERconstraintsIL} we obtain
	\begin{align}
		\frac{1}{\gamma} &= \frac{1}{2-\beta} \geq \frac{1}{2}\\
		\implies 2\geq \gamma &=2-\beta \geq 0.
	\end{align}
	To obtain the second inequality we choose $\alpha = \frac{1}{2}$ or $\alpha = \frac{3}{2}$. Noting that, by minimising the $Z$ measurement term, this inequality no longer depends on the order of the measurement and we may take the minimum bound.
\end{proof}
\chapter{Discussion}
\section{Immediate observations}
There are a couple of notions which follow from the main results.
\subsection{Specialising to classical R\'enyi entropies}
We have the option to consider classical states as a density operator whose eigenvectors form an ONB and are considered as a `classical register'~(see Section~\ref{sec:classical-quantum-states}). In this case the eigenvalues of this density operator represent the values of the probability mass function of a classical random variable.

It is known that the quantum R\'enyi entropy $H_\alpha(\rho_X)$ of such a classical state is exactly the classical R\'enyi entropy $H_\alpha(X)$. By choosing the arguments as classical states in the R\'enyi divergence we recover the classical R\'enyi divergence~(see Eq.~\eqref{Rdiv}) and the associated conditional entropy and mutual information would then be strong candidates for classical R\'enyi versions of the well-known Shannon entropic quantities. Moreover, the chain and decomposition rules and any subsequent uncertainty relations would be equally applicable, providing useful tools for research in classical information theory.
\subsection{Monotonicity in $\alpha$}
Each of the possible inequalities given by the families of inequalities in Chapters~\ref{sec:renyi-entropy-divergence-inequalities} and~\ref{sec:generalised-renyi-divergence-uncertainty-relations} additionally allows for a whole spectrum of weaker ancillary inequalities via the monotonicity of the R\'enyi divergence.

Considering this fact, we introduce a great deal more freedom in choosing specific parameters beyond those stipulated in the conditions of each theorem. This significantly widens the applicability of each result and provides an extra level of generality.
\section{Advantages and drawbacks}
Clearly, the divergence equations presented not only give us some insight into the fundamental relationships between these quantities but also provide important tools in the derivation of R\'enyi versions of the routine comparisons for Shannon and von Neumann quantities and beyond. The structure afforded may help in consolidating the theory of generalised R\'enyi quantities and provide some possible candidates for further study. The results themselves are quite general, which is usually desirable given that they are more flexible and therefore applicable in more situations. However, this level of generality has the unwanted consequence of obfuscating the actual utility of the results.
\subsection{Place within the broader theory}
\subsubsection{Smooth entropies}

Smooth entropies are defined as optimisations of R\'enyi conditional entropies, but in particular we are concerned with the min- and max-entropies,
\begin{equation}
	H^\uparrow_\infty(A|B)_\rho = H_{\min}(A|B)_\rho\quad\text{and}\quad H^\uparrow_\frac{1}{2}(A|B)_\rho = H_{\max}(A|B)_\rho.
\end{equation}
These particular entropies are of interest as they can be calculated through semi-definite programs which are generally more efficient than direct computation.
The smooth min- and max-entropies are then considered as optimisations over states $\tau$ that are $\varepsilon$-close in the \textit{purified distance} to the given state $\rho$.

We define the \textit{smoothing ball} $\mathcal{B}^\varepsilon(A, \rho) : = \{\tau\in \Den(A) \,:\, P(\tau, \rho) \leq \varepsilon\}$, where 
\begin{equation}
P(\tau, \rho) = \sqrt{1-\left(\tr\left|\sqrt{\rho}\sqrt{\tau}\right| + \sqrt{(1-\tr \rho)(1-\tr\tau) } \right)^2}.
\end{equation}

Then the smooth entropies are defined
\begin{align}
		H^\varepsilon_{\min}(A|B)_{\rho} = \max_{\bar\rho_{AB}\in \mathcal{B}^\varepsilon(\rho_{AB})} H_{\min}(A|B)_{\bar\rho},\\
		H^\varepsilon_{\max}(A|B)_{\rho} = \min_{\bar\rho_{AB}\in \mathcal{B}^\varepsilon(\rho_{AB})} H_{\max}(A|B)_{\bar\rho}.
\end{align}
These `smoothed' entropies are also calculable by semi-definite programs and exhibit similar duality relations as the R\'enyi entropies they are based on. Of particular note is the asymptotic equipartition property which essentially states for normalised smooth min- or max-entropies on a number of copies of the same state, the limit as the number of copies approaches infinity is the von Neumann entropy of the original state, i.e for $\rho_{AB} \in \Den(AB)$ and $\rho_{A^nB^n} = \rho_{AB}^{\otimes n}= \underbrace{\rho_{AB}\otimes\cdots\otimes \rho_{AB}}_{n \text{ times}}\in \Den(A^nB^n)$, 
\begin{equation}
	\lim_{n\rightarrow\infty} \frac{1}{n}H^\varepsilon_{\min}(A^n|B^n)_{\rho^{\otimes n}} = \lim_{n\rightarrow\infty} \frac{1}{n}H^\varepsilon_{\max}(A^n|B^n)_{\rho^{\otimes n}} = H(A|B)_\rho.
\end{equation}

See~\cite{mythesis} and \cite[Chapter 6]{mybook} for a detailed treatment of these quantities, their properties and applications.

Some chain and decomposition rules similar to those covered in Chapter~\ref{sec:renyi-entropy-divergence-inequalities} have been established for smooth min and max-entropies~(see~\cite{vitanov12, ciganovic13}) but it is unknown whether these rules can be extended to smooth entropies of general R\'enyi order. This work may provide some tools or scaffolding to that end and aid in establishing a more coherent theory of smooth entropies.
\subsubsection{Conditional mutual information}
The relationships between the established divergence inequalities would indicate that there is a more general structure available for tripartite and, in turn, multipartite systems.

The following diagram summarises what we have established and the `gaps' in this structure:
\[
\begin{tikzcd}[cramped, column sep=tiny, row sep=huge]
	\text{General tripartite divergence inequality}\arrow[dashed]{r}{\displaystyle{\inf_{\substack{\sigma_{BC}\in \Den^*(BC)}}}}\arrow[dashed,swap]{d}{C\text{ trivial}}\arrow[dotted]{rd}&I^*_\beta(B;A|C)_\rho \gtreqless H_\gamma(\rho_{BC}\|\tau_C)-H_\alpha(\rho_{ABC}\|\tau_{AC})\arrow[dashed]{d}{\tau_{AC} = \tau_C}\arrow{ldd}{\substack{\quad\\\\\\\!\!\!\!\!\!\!\!\!\!\!\!\!\!\!\!\!\!\!\!\!\!\!\!\!\!C\text{ trivial}}}\\
	-H_\alpha\left(\rho_{AB}\left\| \tau_A \right.\right)\lesseqgtr D_\beta\left(\rho_{AB}\left\|\sigma_B\otimes \tau_B \right.\right)+\log\left(\tr\rho_{B}\sigma_B^\frac{1}{\gamma'} \right)^{\gamma'}\arrow[swap]{d}{\displaystyle{\inf_{\sigma_B\in \Den^*(B)}}}\arrow[dotted]{dr}&H_\alpha(\rho_{ABC}\|\tau_C)\gtreqless H^\uparrow_\beta(A|BC)_\rho + H_\gamma(\rho_{BC}\|\tau_C)\arrow{d}{C\text{ trivial}}\\
	I_\beta(\rho_{AB}\|\tau_A) \gtreqless H_\gamma(\rho_{B})-H_\alpha(\rho_{AB}\|\tau_A)\arrow[swap]{r}{\tau_A = \id_A}&H_\alpha(\rho_{AB})\gtreqless H^\uparrow_\beta(A|B)_\rho + H_\gamma(\rho_{B})
\end{tikzcd}
\]
The dashed lines are proposed relationships and the dotted lines are relationships formed by the composition of the explicit transformations.

In essence, this examination implies that there is a general tripartite comparison of R\'enyi quantities from which all inequalities of this type follow. Perhaps of more immediate interest is the indication that there are decomposition/chain rules for some R\'enyi divergence based conditional mutual information (denoted $I^*_\beta(B;A|C)_\rho$ in the diagram). The (von Neumann) quantum conditional mutual information has found many uses in quantum information theory, hence it is currently an important open problem to derive an operationally significant R\'enyi conditional mutual information which is useful in applications. There are several candidates for a quantity of this type~(see~\cite{bertawilde14} and references therein), but one that is compatible with this pre-existing structure would be a strong choice for further study.
\subsection{Applicability}
\subsubsection{Alternative generalised entropies}
There are some other possible candidates for a generalised divergence from which entropic quantities could be similarly derived. Notable are the Tsallis entropies~\citep{tsallis88} and the Petz divergence~\citep{petz86}. However, these quantities do not have the structure which allows us to derive comparisons via interpolation. That said, the Petz divergence is still closely related to the R\'enyi divergence, sharing many of the same properties and there exist duality relations involving both quantities~\citep{mybook}.

Although the quantum R\'enyi divergence has already been established as one of the forerunners in the possible choices for the basis of a generalised quantum R\'enyi entropy framework, having found many applications and utility in quantum information theory, the structure revealed in this work further reinforces the compatibility and usefulness of this definition.
\subsubsection{Generality vs. usability}
The comparisons in Chapter~\ref{sec:renyi-entropy-divergence-inequalities} are, in isolation, relatively simple expressions. However, in applications they are often used in conjunction with one another, as evidenced by the proofs in Chapter~\ref{sec:generalised-renyi-divergence-uncertainty-relations}.

The results of composing the conditions on the parameters for each individual comparison are rather unwieldy and generally take the form of obscure polynomial expressions. An immediate solution is to make choices which yield the tightest versions of the relevant inequality. This notion has already been applied in this thesis, albeit the final conditions still often being difficult to interpret.

Clearly, there is a trade-off in how general we can formulate these types of statements and their accessibility for a broader audience. For the moment, it seems the possible loss of a small amount of information at the gain of usability and consistency is worth the cost.

Overall, the analysis of the relationship between, and valid ranges of, the parameters given a particular set of comparisons has proved time-consuming and beyond the scope of this work. Moreover, any attempts to do so have been unfruitful in establishing a consistent underlying structure, given one exists.
\section{Future work}
We now suggest some possible directions for research or study which follow from the material covered in this thesis.
\subsection{Numerical simulations}
At the moment we have only the abstract representation of most of the relationships detailed in this thesis. By running numeric simulations and producing visualisations thereof we may gain some insights into the relationships between the R\'enyi orders, dimension, degree of entanglement and the relative tightness of each comparisons. An analysis of these trends and correlations may illuminate possible avenues for further refinements and improvements. These simulations are relatively easy to perform when calculating the divergence of known matrices but the problem becomes more involved when considering the necessary optimisations~--~fortunately there are packages and/or libraries available for most high-level programming languages that allow for efficient optimisation over convex sets. The nature of the models require exponentially more processing power for calculations involving higher dimensions, so obtaining sufficiently broad samples may still be a significant undertaking, requiring dedicated time and resources.
\subsection{Multiple measurements}
An interesting direction to take these general uncertainty relations would be to formulate statements involving more than two measurements. Of course, given uncertainty relations of the form
\begin{equation}
	H_{\alpha_1}(X_1) + H_{\alpha_2}(X_2) \geq q_{1,2}\quad\text{and}\quad	H_{\alpha_3}(X_3) + H_{\alpha_4}(X_4) \geq q_{3,4},
\end{equation}
we may immediately derive
\begin{equation}
	H_{\alpha_1}(X_1) + H_{\alpha_2}(X_2) +	H_{\alpha_3}(X_3) + H_{\alpha_4}(X_4) \geq q_{\max},
\end{equation}
where $q_{\max} = \max_{\{i,j,k,l\}=\{1,2,3,4\}}q_{i,j} + q_{k,l}$. However, this does not yield any advantage compared to the uncertainty relations from which it was derived.

We would instead be investigating inequalities of the form
\begin{equation}
	\sum_n H_{\alpha_n}(X_n|B)_\rho\geq q,
\end{equation}
such that the bound $q\geq q_{\max}$. Given the R\'enyi decomposition rule and general comparisons we could then derive 
\begin{equation}
	\sum_n I_{\alpha_n}(X_n:B)_\rho\leq r
\end{equation} 
with $r$ dependent on $q$.

For a more complete review of multiple measurement uncertainty relations see~\cite[Section~III.G]{coles17} and the references therein.
\subsection{POVMs and tripartite uncertainty relations}
In Chapter~\ref{sec:generalised-renyi-divergence-uncertainty-relations} we only consider measurements on ONBs as it provides a simpler scaffold. Indeed, most of the results in Section~\ref{sec:ED} have at least one generalisation to POVMs. It is currently an open question whether these generalisations extend to general R\'enyi uncertainty relations.

A feature of these POVM-based von Neumann relations is that they are often expressed on tripartite systems. It is then likely any R\'enyi version would also require the consideration of tripartite systems. For a more detailed treatment of these situations see~\cite{coles14} and~\cite{mybook}.
\subsection{Adapting Pisier norms}
As stated in Section~\ref{sec:PN}, the full generality of Pisier's norms has not been used in the interpolation employed in Chapter~\ref{sec:renyi-entropy-divergence-inequalities}. This does not mean that the close relationship between these norms and R\'enyi divergence based quantities does not merit further thought. Indeed, the general framework of Pisier contains some useful results which could provide important insights into the nature and structure of these quantities.

The main issue with this pursuit is the ability to make use of this theory without having to rely too heavily on the unnecessarily abstract mathematics it is based on. There is currently work underway to adapt this theory in a cohesive way, accessible from a quantum information theory perspective, and to investigate the implications of such a framework.
\appendix
\chapter{An alternative proof of the decomposition rules}\label{sec:an-alternative-proof-of-the-decomposition-rule}
We include a different approach\footnote{This approach was used to produce the results in~\cite{mckinlay19} which in turn were applied to the derivation of information exclusion relations. These relations have subsequently been improved in this thesis by the more amenable ranges found in Theorem~\ref{decomp}.} to the decomposition rules which follows the method employed in~\cite{Dup} in establishing Theorem~\ref{chainDup}. In a similar vein to the differences between the two versions of the bipartite chain rule (Theorem \ref{Bchain} and Corollary \ref{BchainCor}), this approach produces slightly different valid ranges for the parameters compared to Theorem~\ref{decomp}. Although for certain applications the ranges on this alternative result may be preferable, for our purposes in the main body of the thesis they are more constrained than desirable.\\

Reproduced from A. McKinlay and M. Tomamichel, ``Decomposition rules for quantum R\'enyi mutual
information with an application to information exclusion relations,'' \textit{Journal of Mathematical
Physics}, vol. 61, no. 7, p. 072202, 2020., with the permission of AIP Publishing.\\

\begin{theorem}\label{decompDup}
	Let $\alpha\beta\gamma  -2\beta\gamma - \alpha + \beta +\gamma = 0$ with $\alpha\geq0$, $\beta> \frac{1}{2}$ and $\gamma\geq \frac{1}{2}$. Then, for $\rho_{AB}\in\Den(AB)$ and $\tau\in\Pos(A)$, if $\frac{1}{\beta} + \frac{1}{\gamma} \leq 2$
	\begin{equation}\label{decompDupEq1}
		I_\beta(\rho_{AB}\|\tau_A) \geq H_\gamma(\rho_{B})-H_\alpha(\rho_{AB}\|\tau_A).
	\end{equation}
	Otherwise, if $\frac{1}{\beta} + \frac{1}{\gamma} \geq 2$
	\begin{equation}\label{decompDupEq2}
		I_\beta(\rho_{AB}\|\tau_A) \leq H_\gamma(\rho_{B})-H_\alpha(\rho_{AB}\|\tau_A).
	\end{equation}
\end{theorem}
This method draws on several components, which we break up into the following sections.
\section{Operator-vector correspondence}\label{sec:OPV}
We may establish a correspondence between vectors on a composite Hilbert space and operators mapping from one subspace to another. Given bases $\{\ket{e_i}\}_i$ and $\{\ket{f_i}\}_i$ for $\Hil_{A}$ and $\Hil_{B}$ respectively, define
\begin{equation}
	\Op_{A\rightarrow B}\left(\ket{e_i}\otimes\ket{f_i}\right)  = \ketbra{f_i}{e_i}.
\end{equation}
Accordingly, for $\ket{\psi}\in \Hil_{AB}$ such that
$
	\ket{\psi} = \sum_{i,j} \lambda_{ij}\ket{e_i}\otimes\ket{f_j}
$
we have
$
	\Op_{A\rightarrow B}\left(\ket{\psi}\right)  = \sum_{i,j} \lambda_{ij}\ketbra{f_j}{e_i}
$.
There are some useful properties of this correspondence, for proofs see~\cite[Section~1.1]{Wat}.
\begin{lemma}\label{OV1}
	Let $\ket{\psi} \in \Hil_{AB}$, $M\in \Lin(A)$ and $N\in \Lin(B)$. Then \begin{equation} \Op_{A\rightarrow B}\left[(M\otimes N) \ket{\psi}\right] = N \Op_{A\rightarrow B}(\ket{\psi})M^\top.\end{equation}
\end{lemma}
\begin{lemma}\label{Isom}
	Let $\ket{\psi}\in\Hil_{AB}$. Then
	\begin{equation}
		\|\ket{\psi}\|_2 = \sqrt{\braket{\psi}{\psi}} = \left(\tr\left[\Op_{A\rightarrow B}(\ket{\psi})^\dagger \Op_{A\rightarrow B}(\ket{\psi})\right]\right)^\frac{1}{2} =\|\Op_{A\rightarrow B}(\ket{\psi})\|_2.
	\end{equation}
\end{lemma}
\begin{lemma}\label{OV3}
	Let $\ket{\psi}\in \Hil_{AB}$, $\rho_{AB} = \ketbra{\psi}{\psi}$. Then
	\begin{equation}\rho_A = \Op_{A\rightarrow B}(\ket{\psi})^\dagger \Op_{A\rightarrow B}(\ket{\psi})\quad\text{and}\quad
		{\rho_B = \Op_{A\rightarrow B}(\ket{\psi})\Op_{A\rightarrow B}(\ket{\psi})^\dagger}
		.\end{equation}
\end{lemma}
\section{Identities in terms of operator-vector correspondence}
We begin the main part of the demonstration by establishing the following identities for the relevant entropic quantities:
\begin{lemma}\label{threeEnt}
	For a pure state $\ket{\varphi} \in \Hil_{ABC}$ with $\ketbra{\varphi}{\varphi} = \rho_{ABC}$, let ${X = X_{B\rightarrow AC} = \Op_{B\rightarrow AC}(\ket{\varphi})}$. Given $(\tau_A\otimes\id_B) \gg \rho_{AB}$ and $\alpha\in (0, 1)\cup (1, \infty)$ we have 
	\begin{align}
		H_\alpha(\rho_{AB}\|\tau_A) &= -\log \sup_{\sigma_C\in \Den(C)}\left\|\left(\tau_A^{-1}\otimes\sigma_C\right)^\frac{1}{2\alpha'}X\right\|_{2}^{2\alpha'},\label{CRE}\\
		H_\alpha(\rho_B) &= -\log\left\|X\right\|_{2\alpha}^{2\alpha'}.\label{RE}
	\end{align}
	If in addition $\alpha> \frac{1}{2}$,
	\begin{align}
		I_\alpha(\rho_{AB}\|\tau_A) &=\log\sup_{\sigma_C\in \Den(C)}\left\|\left(\tau_A^{-1}\otimes\sigma_C\right)^\frac{1}{2\alpha'}X\right\|_{2\hat\alpha}^{2\alpha'}.\label{MI}
	\end{align}
\end{lemma}
The proof of Lemma~\ref{threeEnt} relies on the operator-vector correspondence.

\begin{proof}[Proof of Eq.~\eqref{CRE}]
	From equation (19) in~\cite{lennert13}, we can write
	\begin{align}
		H_\alpha(\rho_{AB}\|\tau_A) = -\log\sup_{\sigma_C \in \Den(C)}\bra{\varphi}\tau_A^{\frac{-1}{\alpha'}}\otimes\id_B\otimes\sigma_C^{\frac{1}{\alpha'}}\ket{\varphi}^{\alpha'}.
	\end{align}
	Also, using Lemma~\ref{OV1} we have
	\begin{align}
		\Op_{B\rightarrow AC}\left(\tau_A^\frac{-1}{2\alpha'}\otimes\id_B\otimes\sigma_C^\frac{1}{2\alpha'}\ket{\varphi}\right) = \left(\tau_A^\frac{-1}{2\alpha'}\otimes\sigma_C^\frac{1}{2\alpha'}\right)X.
	\end{align}
	From this we can deduce, using Lemma~\ref{Isom},
	\begin{align}
		\bra{\varphi}\tau_A^{\frac{-1}{\alpha'}}\otimes\id_B\otimes\sigma_C^{\frac{1}{\alpha'}}\ket{\varphi}^{\alpha'} &= \left\|\tau_A^\frac{-1}{2\alpha'}\otimes\id_B\otimes\sigma_C^\frac{1}{2\alpha'}\ket{\varphi}\right\|_2^{2\alpha'}\\
		&= \left\|\Op_{B\rightarrow AC}\left(\tau_A^\frac{-1}{2\alpha'}\otimes \id_B\otimes\sigma_C^\frac{1}{2\alpha'}\ket{\varphi}\right)\right\|_2^{2\alpha'}\\
		&= \left\|\left(\tau_A^\frac{-1}{2\alpha'}\otimes\sigma_C^\frac{1}{2\alpha'}\right)X\right\|_2^{2\alpha'}.
	\end{align}
\end{proof}
\begin{proof}[Proof of Eq.~\eqref{RE}]
	From Lemma~\ref{OV3} we can see that $\rho_B= X^\dagger X$, hence we have
	\begin{align}
		H_\alpha(B)_{\rho} = -\log\left\|X^\dagger X\right\|_\alpha^{\alpha'}= -\log\left\|X\right\|_{2\alpha}^{2\alpha'},
	\end{align}
	where we have also used the identity, Proposition~\ref{fold}.
\end{proof}
The proof of Eq.~\eqref{MI} relies on the duality result in Proposition~\ref{Dual}.
\begin{proof}[Proof of Eq.~\eqref{MI}]
	Using operator-vector correspondence we can re-express the generalised R\'enyi mutual information using operator norms, i.e. for $X_{C\rightarrow AB} = \Op_{C\rightarrow AB}\left(\ket{\varphi}\right)$,
	\begin{align}
		I_\alpha\left(\rho_{AB}\|\tau_A\right) &= \frac{1}{\alpha-1}\log\inf_{\sigma_B\in \Den(A)}\tr\left(\left[\left(\tau_A\otimes\sigma_B\right)^\frac{-1}{2\alpha'}\rho_{AB}\left(\tau_A\otimes\sigma_B\right)^\frac{-1}{2\alpha'}\right]^\alpha\right)\\
		&=\log\inf_{\sigma_B\in \Den(A)}\left\|\left(\tau_A\otimes\sigma_B\right)^\frac{-1}{2\alpha'}X_{C\rightarrow AB}X_{C\rightarrow AB}^\dagger\left(\tau_A\otimes\sigma_B\right)^\frac{-1}{2\alpha'}\right\|_\alpha^{\alpha'}\\
		&=\log\inf_{\sigma_B\in \Den(A)}\left\|\left(\tau_A\otimes\sigma_B\right)^\frac{-1}{2\alpha'}X_{C\rightarrow AB}\right\|_{2\alpha}^{2\alpha'}.\numberthis\label{MINorm}
	\end{align}
	By the duality of the generalised mutual information and Eq.~\eqref{MINorm}, we can write
	\begin{align}
		I_\alpha\left(\rho_{AB}\|\tau_A\right) &=-I_{\hat\alpha}\left(\rho_{AC}\|\tau_A^{-1}\right)\\
		&=-\log\inf_{\omega_C\in \Den(C)}\left\|\left(\tau_A^\frac{1}{2\hat\alpha'}\otimes\omega_C^\frac{-1}{2\hat\alpha'}\right)X\right\|_{2\hat\alpha}^{2\hat\alpha'}\\
		&=\log\sup_{\omega_C\in \Den(C)}\left\|\left(\tau_A^\frac{-1}{2\alpha'}\otimes\omega_C^\frac{1}{2\alpha'}\right)X\right\|_{2\hat\alpha}^{2\alpha'}.
	\end{align}
	where in the last line we used the fact that $\hat\alpha' = -\alpha'$.
\end{proof}

\section{Interpolating for valid choices of $\theta$}
Theorem~\ref{decompDup} can be proved directly from the following propositions which make use of the above identities.
\begin{proposition}\label{part1}
	Let $\alpha,\beta,\gamma$ be such that ${\alpha'} = {\beta'}+{\gamma'}$, $\rho_{AB} \in \Den(AB)$ and $\tau_A\in \PosS(A)$. Then the following holds:\\
	
	For $\alpha\in (1,2),\beta,\gamma\in(1,\infty)$, we find
	\begin{align}
		I_\beta(\rho_{AB}\|\tau_A) \geq H_\gamma(\rho_{B})-H_\alpha(\rho_{AB}\|\tau_A).\label{p1ugd}
	\end{align}
	
	For $\alpha\in \left[2/3,1\right),\beta\in\left(1/2,1\right),\gamma\in\left[1/2,1\right)$, we find
	\begin{align}
		I_\beta(\rho_{AB}\|\tau_A) \leq H_\gamma(\rho_{B})-H_\alpha(\rho_{AB}\|\tau_A).\label{p1uld}
	\end{align}
\end{proposition}
\begin{proof}
	Choose $F(z) = (\tau_A^{-1}\otimes\sigma_C)^\frac{z}{2\beta'}X,\quad
	\theta = \frac{\beta'}{\alpha'},\quad
	p_0 = 2\gamma,\quad
	p_1 = 2\hat\beta.$
	With these choices we can determine $\theta = \frac{1}{\alpha'}\left({\alpha'}-{\gamma'} \right) = 1 - \frac{\gamma'}{\alpha'}$, hence $1-\theta = \frac{\gamma'}{\alpha'}$.
	
	We can also calculate the appropriate value of $p_\theta$ to use Theorem~\ref{weighted3lines}:
	\begin{align}
		\frac{1}{p_\theta} = \frac{\gamma'}{2\alpha'\gamma} + \frac{\beta'}{2\alpha'\hat\beta}\implies\frac{2\alpha'}{p_\theta} =\gamma'-1 + \beta' + 1= {\gamma'} + {\beta'},
	\end{align}
	thus we can conclude that $p_\theta = 2$.
	
	We can therefore calculate that
	\begin{equation}
		\left\|F(\theta)\right\|_{p_\theta} = \left\|(\tau_A^{-1}\otimes\sigma_C)^\frac{1}{2\alpha'}X\right\|_2.
	\end{equation}
	Additionally,
	\begin{equation}
		\|F(\im t)\|_{p_0} = \left\|(\tau_A^{-1}\otimes\sigma_C)^\frac{\im t}{2\beta'}X\right\|_{2\gamma}
	\end{equation}
	and
	\begin{equation}
		\|F(1+\im t)\|_{p_1} = \left\|(\tau_A^{-1}\otimes\sigma_C)^\frac{1+\im t}{2\beta'}X\right\|_{2\hat\beta}.
	\end{equation}
	Since $(\tau_A^{-1}\otimes\sigma_C)^\frac{\im t}{2\beta'}$ is unitary for all $t\in \R$ we can write
	\begin{equation}
		\mathsf{M}_0 = \left\|X\right\|_{2\gamma}\quad\text{and}\quad \mathsf{M}_1 = \left\|(\tau_A^{-1}\otimes\sigma_C)^\frac{1}{2\beta'}X\right\|_{2\hat\beta}.
	\end{equation}

	Applying Theorem~\ref{weighted3lines} we have
	\begin{align}
		\left\|(\tau_A^{-1}\otimes\sigma_C)^\frac{1}{2\alpha'}X\right\|_2&\leq \left\|X\right\|_{2\gamma}^\frac{\gamma'}{\alpha'}\left\|(\tau_A^{-1}\otimes\sigma_C)^\frac{1}{2\beta'}X\right\|_{2\hat\beta}^\frac{\beta'}{\alpha'}.
	\end{align}
	
	First, consider $\alpha'>0$. Maximising over $\sigma_C$ on both sides we have
	\begin{align}
		\sup_{\sigma_C\in \Den(C)}\left\|(\tau_A^{-1}\otimes\sigma_C)^\frac{1}{2\alpha'}X\right\|_2^{2\alpha'}&\leq \left\|X\right\|_{2\gamma}^{2\gamma'}\sup_{\sigma_C\in \Den(C)}\left\|(\tau_A^{-1}\otimes\sigma_C)^\frac{1}{2\beta'}X\right\|_{2\hat\beta}^{2\beta'}.
	\end{align}
	
	Using Lemma~\ref{threeEnt}, we can rewrite this as
	\begin{align}
		-H_\alpha(\rho_{AB}\|\tau_A) &\leq -H_\gamma(\rho_{B})+I_\beta(\rho_{AB}\|\tau_A)\\
		\implies I_\beta(\rho_{AB}\|\tau_A) &\geq H_\gamma(\rho_{B})-H_\alpha(\rho_{AB}\|\tau_A).
	\end{align}
	
	If instead $\alpha'<0$, we obtain
	\begin{align}
		\sup_{\sigma_C\in \Den(C)}\left\|(\tau_A^{-1}\otimes\sigma_C)^\frac{1}{2\alpha'}X\right\|_2^{2\alpha'}&\geq \left\|X\right\|_{2\gamma}^{2\gamma'}\sup_{\sigma_C\in \Den(C)}\left\|(\tau_A^{-1}\otimes\sigma_C)^\frac{1}{2\beta'}X\right\|_{2\hat\beta}^{2\beta'},
	\end{align}
	giving us Eq.~\eqref{p1uld}. The valid ranges can be determined using Lemma~\ref{ana}.
\end{proof}
\begin{proposition}\label{part2}
	Let $\alpha,\beta,\gamma$ be such that ${\alpha'} = {\beta'} + {\gamma'}$, $\rho_{AB} \in \Den(AB)$ and $\tau_A\in \PosS(A)$. Then the following holds:\\
	
	For $\alpha\in(0,1), \beta \in (1/2,1),\gamma\in(1,\infty)$, we find
	\begin{align}
		I_\beta(\rho_{AB}\|\tau_A) \geq H_\gamma(\rho_{B})-H_\alpha(\rho_{AB}\|\tau_A)\label{p2ugd}.
	\end{align}
	
	For $\gamma\in[1/2,1), \beta \in (1,2),\alpha\in(1,\infty)$, we find
	\begin{align}
		I_\beta(\rho_{AB}\|\tau_A) \leq H_\gamma(\rho_{B})-H_\alpha(\rho_{AB}\|\tau_A)\label{p2uld}.
	\end{align}
\end{proposition}
\begin{proof}
	Choose $F(z) = (\tau_A^{-1}\otimes\sigma_C)^\frac{z}{2\alpha'}X,\quad
	\theta = \frac{\alpha'}{\beta'},\quad
	p_0 = 2\gamma,\quad
	p_1 = 2$.
	We have, as before, $1-\theta = \frac{-\gamma'}{\beta'}$ and through a similar calculation we can conclude that $p_\theta= 2\hat\beta$.
	
	We have
	\begin{align}
		&\|F(\theta)\|_{p_\theta} = \left\|(\tau_A^{-1}\otimes\sigma_C)^\frac{1}{2\beta'}X\right\|_{2\hat\beta}, \quad
		\|F(\im t)\|_{p_0} = \left\|(\tau_A^{-1}\otimes\sigma_C)^\frac{\im t}{2\alpha'}X\right\|_{2\gamma}\\
		&\text{and }\|F(1+\im t)\|_{p_1} = \left\|(\tau_A^{-1}\otimes\sigma_C)^\frac{1+\im t}{2\alpha'}X\right\|_{2},
	\end{align}
	hence $\mathsf{M}_0 = \left\|X\right\|_{2\gamma}\text{ and } \mathsf{M}_1 = \left\|(\tau_A^{-1}\otimes\sigma_C)^\frac{1}{2\alpha'}X\right\|_{2}$.
	
	Applying Theorem~\ref{weighted3lines} we have
	\begin{align}
		\left\|(\tau_A^{-1}\otimes\sigma_C)^\frac{1}{2\beta'}X\right\|_{2\hat\beta}&\leq\left\|X\right\|_{2\gamma}^\frac{-\gamma'}{\beta'}\left\|(\tau_A^{-1}\otimes\sigma_C)^\frac{1}{2\alpha'}X\right\|_{2}^\frac{\alpha'}{\beta'}.
	\end{align}
	
	First, we consider the case where $\beta'>0$. It follows that
	\begin{align} \left\|(\tau_A^{-1}\otimes\sigma_C)^\frac{1}{2\beta'}X\right\|_{2\hat\beta}^{2\beta'}&\leq\left\|X\right\|_{2\gamma}^{-2\gamma'}\left\|(\tau_A^{-1}\otimes\sigma_C)^\frac{1}{2\alpha'}X\right\|_{2}^{2\alpha'}.
	\end{align}
	
	As in Proposition~\ref{part1}, we can maximise over $\sigma_C$ on both sides. Hence we arrive at Eq.~\eqref{p2uld}.
	Repeating the same process with the assumption $\beta'<0$ yields Eq.~\eqref{p2ugd}. We can again refer to Lemma~\ref{ana} to determine the valid ranges.
\end{proof}
\begin{proposition}\label{part3}Let $\alpha,\beta,\gamma$ be such that ${\alpha'} = {\beta'}+{\gamma'}$, $\rho_{AB} \in \Den(AB)$ and $\tau_A\in \PosS(A)$. Then the following holds:\\
	
	For $\alpha\in(0,1), \gamma \in [1/2,1),\beta\in(1,\infty)$, we find
	\begin{align}
		I_\beta(\rho_{AB}\|\tau_A) \geq H_\gamma(\rho_{B})-H_\alpha(\rho_{AB}\|\tau_A)\label{p3ugd}.
	\end{align}
	
	For $\beta\in(1/2,1), \gamma \in (1,2),\alpha\in(1,\infty)$, we find
	\begin{align}
		I_\beta(\rho_{AB}\|\tau_A) \leq H_\gamma(\rho_{B})-H_\alpha(\rho_{AB}\|\tau_A)\label{p3uld}.
	\end{align}
\end{proposition}
\begin{proof}
	Choose $F(z) = (\tau_A^{-1}\otimes\sigma_C)^{\frac{1}{2\beta'}-z\frac{\gamma'}{2\alpha'\beta'}}X,\quad
	\theta = \frac{\alpha'}{\gamma'},\quad
	p_0 = 2\hat\beta,\quad
	p_1 = 2$.
	As above, $1-\theta = \frac{-\beta'}{\gamma'}$ and $p_\theta= 2\gamma$.
	
	We have
	\begin{align}
		&\|F(\theta)\|_{p_\theta} = \left\|X\right\|_{2\gamma},\quad
		\|F(\im t)\|_{p_0} = \left\|(\tau_A^{-1}\otimes\sigma_C)^{\frac{1}{2\beta'}-\frac{\im t\gamma'}{2\alpha'\beta'}}X\right\|_{2\hat\beta}\\
		&\text{and }\|F(1+\im t)\|_{p_1} = \left\|(\tau_A^{-1}\otimes\sigma_C)^{\frac{1}{2\alpha'}-\frac{\im t\gamma'}{2\alpha'\beta'}}X\right\|_{2},
	\end{align}
	hence $\mathsf{M}_0 = \left\|(\tau_A^{-1}\otimes\sigma_C)^{\frac{1}{2\beta'}}X\right\|_{2\hat\beta}\text{ and } \mathsf{M}_1 = \left\|(\tau_A^{-1}\otimes\sigma_C)^\frac{1}{2\alpha'}X\right\|_{2}$.
	Applying Theorem~\ref{weighted3lines} and performing the same procedure as in Propositions~\ref{part1} and~\ref{part2}, for both $\gamma'>0$ and $\gamma'<0$ we obtain Eqs.~\eqref{p3ugd} and ~\eqref{p3uld}. For the valid ranges, we have a similar situation as in Proposition~\ref{part2} but with symmetry in $\beta$ and $\gamma$.
\end{proof}
\section{Consolidation and analysis}
We may now prove Theorem~\ref{decompDup}:
\begin{proof}[Proof of Theorem~\ref{decompDup}]
	We first combine the three above propositions and examine the valid ranges.
	We have from Lemma~\ref{ana} that the propositions cover all possible permutations of the parameters, and hence all valid values of $\alpha, \beta$ and $\gamma$.
	
	For the forward inequality, i.e. Eqs.~\eqref{p1ugd},~\eqref{p2ugd} and~\eqref{p3ugd} we can see that either ($\alpha, \beta, \gamma > 1$), ($\alpha, \gamma<1, \beta > 1$) or ($\alpha, \beta <1, \gamma>1$), which all satisfy $\frac{1}{\beta} + \frac{1}{\gamma}\leq2$.
	
	For the reverse inequality, i.e. Eqs.~\eqref{p1uld},~\eqref{p2uld} and~\eqref{p3uld} we have either (${\alpha, \beta, \gamma < 1}$), ($\alpha, \gamma >1, \beta > 1$) or ($\alpha, \beta>1, \gamma<1$), which all satisfy $\frac{1}{\beta} + \frac{1}{\gamma}\geq2$.
	We again obtain the final conditions and the extension to positive semi-definite matrices in the same manner as Theorem~\ref{DivFormIntLem}.
\end{proof}
\chapter{Other useful results}
\section{Background results}\label{sec:alpha1}
\begin{proposition}\label{ato1}
	\begin{equation}
		\lim_{\alpha\rightarrow 1}H_\alpha(X) = H(X)
	\end{equation}
\end{proposition}
\begin{proof}
	We have
	\begin{equation}
		H_\alpha(X) = \frac{1}{1-\alpha}\log\left(\sum_x p(x)^\alpha\right).
	\end{equation}
We choose:
\begin{equation}
	f(\alpha) = \log\left(\sum_x p(x)^\alpha\right),\quad g(\alpha) = 1-\alpha.
\end{equation}
Observe that
$
	\lim_{\alpha\rightarrow 1}f(\alpha) = \log\sum_x p(x) = \log 1 = 0$ and $\lim_{\alpha\rightarrow 1}g(\alpha) = 0.
$
Hence, we may use l'H\^opital's rule to assert
$
	\lim_{\alpha\rightarrow 1} \frac{f(\alpha)}{g(\alpha)} = \lim_{\alpha\rightarrow 1} \frac{f'(\alpha)}{g'(\alpha)}
$.
We compute
\begin{align}
	f'(\alpha) = \frac{d}{d\alpha}\left[\log\left(\sum_x p(x)^\alpha\right)\right]= \frac{\sum_x \frac{d}{d\alpha}p(x)^\alpha}{\sum_x p(x)^\alpha}= \frac{\sum_x p(x)^\alpha \log p(x)}{\sum_x p(x)^\alpha}.
\end{align}
and $g'(\alpha) = -1$. Therefore we may write
\begin{equation}
	\lim_{\alpha\rightarrow 1} \frac{f(\alpha)}{g(\alpha)} = \lim_{\alpha\rightarrow 1} -\frac{\sum_x p(x)^\alpha \log p(x)}{\sum_x p(x)^\alpha} = -\sum_x p(x) \log p(x).
\end{equation}
\end{proof}
\section{Ancillary results}\label{res}
We make use of a modified version of Lemma 12 from~\cite{lennert13}.
\begin{lemma}\label{newLem12}
	Let $p\in \R^+ \setminus \{0,1\}$ and $p' \in \R\setminus [0,1]$ be such that $\frac{1}{p} + \frac{1}{p'} = 1$. Then for $M\in\Lin(A,B)$ such that $M^\dagger M = X\in\Pos(A)$,
	\begin{equation}\label{newLem12eq}
		\left\|M\right\|_{2p}^{2p'} = \left\|X\right\|_p^{p'} = \sup_{\sigma_A\in\Den^*(A)}\left(\tr X\sigma_A^\frac{1}{p'}\right)^{p'} = \sup_{\sigma_A\in\Den^*(A)}\left\|\Gamma_{\id_B,\sigma_A}^\frac{1}{p'}(M)\right\|_2^{2p'}.
	\end{equation}
\end{lemma}
\begin{proof}
	We restrict to density operators with full support which guarantees $\sigma_A\gg X$, hence by Lemma 12 from~\cite{lennert13}
	\begin{align}
		\left\|X\right\|_p^{p'} = \begin{cases}
			\left(\sup_{\sigma_A\in\Den^*(A)}\tr X\sigma_A^\frac{1}{p'}\right)^{p'}\quad & \text{if }p>1 \\
			\left(\inf_{\sigma_A\in\Den^*(A)}\tr X\sigma_A^\frac{1}{p'}\right)^{p'}\quad & \text{if }p<1 
		\end{cases}.
	\end{align}
	Note that when $p>1\implies p'>1$ and otherwise when $p<1\implies p'<0$. Hence when we take the optimisation outside the exponent we obtain Eq.~\eqref{newLem12eq}. The other equalities are evident from the definition of the Schatten norm~(see Section~\ref{sec:the-schatten-operator-norm}).
\end{proof}
It is prudent to explicitly state the available choices we have for the parameters when performing the interpolation in Theorems~\ref{DivFormIntLem},~\ref{chain} and related results. This is especially important when we are constrained to a particular direction of the given inequality.
\begin{lemma}\label{ana}
	If $\alpha,\beta\geq\frac{1}{2},\gamma\in\R$ and are related by
	\begin{align}
		\label{rel}\frac{\alpha}{\alpha-1} = \frac{\beta}{\beta-1} + \frac{\gamma}{\gamma-1}
	\end{align}
	and assuming, without loss of generality, that $\beta\geq\gamma$, then the following are true and cover all possible cases up to symmetry:
	\begin{align*}
		            & \text{If }0<\frac{\beta'}{\alpha'}<1\text{ then either}                                                                                                            \\
		\numberthis & \qquad\text{\bf Case 1. }\alpha,\beta, \gamma>1,\quad\alpha<\gamma\leq\beta\quad\text{and}\quad\alpha\in(1, 2),\beta,\gamma\in(1,\infty)\label{case1}              \\
		            & \text{or}                                                                                                                                                           \\
		\numberthis & \qquad\text{\bf Case 2. }\alpha,\beta, \gamma<1,\quad \gamma\leq\beta<\alpha\quad\text{and}\quad\alpha, \beta\in[1/2, 1),\gamma\in[0,1)\label{case2}               \\
		            & \text{or}                                                                                                                                                           \\
		\numberthis & \qquad\text{\bf Case 3. }\alpha,\beta>1, \gamma\leq 0,\quad\gamma<\alpha\leq\beta\quad \text{and}\quad\alpha,\beta\in (1,\infty), \gamma\in(-\infty,0].\label{case3} \\
		            & \text{If }0<\frac{\alpha'}{\gamma'}<1\text{ then}                                                                                                                   \\
		\numberthis & \qquad\text{\bf Case 4. }\alpha,\gamma<1,\beta>1,\quad\alpha<\gamma<\beta\quad\text{and}\quad\alpha\in[1/2,1), \gamma \in (2/3,1),\beta\in(1,\infty).\label{case4}  \\
		            & \text{If }0<\frac{\alpha'}{\beta'}<1\text{ then}                                                                                                                    \\
		\numberthis & \qquad\text{\bf Case 5. }\alpha,\beta>1,\gamma<1,\quad\gamma<\beta<\alpha\quad\text{and}\quad\alpha,\beta\in(1,\infty),\gamma\in[0,1) \label{case5}                \\
		            & \text{or}                                                                                                                                                           \\
		\numberthis & \qquad\text{\bf Case 6. }\alpha,\beta<1, \gamma\leq 0,\quad\gamma<\alpha\leq\beta\quad \text{and}\quad\alpha,\beta\in [1/2,1), \gamma\in(-\infty,0].\label{case6}
	\end{align*}
\end{lemma}
\begin{proof}
	We first consider the cases where $\gamma>0$ then move on to $\gamma\leq 0$.
	
	Given $\gamma>0$, we investigate the possible cases or, more specifically, the cases missing from the lemma.
	Given three independent binary options there are 8 possible permutations. Of the four that are missing the following:
	$(\alpha, \gamma>1, \beta<1)$ and
	$(\alpha, \beta <1, \gamma>1)$, contradict the assumption that $\beta>\gamma$.
	The remaining two:
	$(\alpha>1, \beta,\gamma<1)$ and
	$(\alpha<1, \beta, \gamma>1)$,
	never satisfy Eq.~\eqref{rel}. We can now explore the implications of each of the assumptions.
	
	Consider $0<\frac{\beta'}{\alpha'}<1$. It is evident that $(\alpha - 1)(\beta - 1)>0$, a condition which now excludes \textbf{Case 4}. However, we can examine the two situations where this condition is satisfied:
	\begin{align}
		0<\frac{\beta'}{\alpha'}<1\implies \begin{cases}
			\alpha<\beta \quad \text{if}\quad \alpha, \beta > 1 \\
			\alpha>\beta \quad \text{if}\quad \alpha, \beta < 1
		\end{cases}.
	\end{align}
	It is clear that \textbf{Case 5} does not satisfy these implications but that \textbf{Cases 1}  and \textbf{2} do depending on the sign of $\alpha - 1$.
	
	For \textbf{Case 1}, we can calculate that
	$\displaystyle{\lim_{\eta\rightarrow 1^{+}} \frac{\eta}{\eta-1} =\infty}$ and
	$\displaystyle{\lim_{\eta\rightarrow \infty} \frac{\eta}{\eta-1} = 1}$.
	
	Since $\alpha, \beta$ and $\gamma$ are related by Eq.~\eqref{rel}, we have
	\begin{align}
		\alpha & \longrightarrow 1 \implies \beta,\gamma\longrightarrow 1\quad \text{and} \\
		\alpha & \longrightarrow 2 \implies \beta,\gamma\longrightarrow \infty,
	\end{align}
	i.e. $1<\alpha<2$ and $1<\beta,\gamma<\infty$.
	
	Moreover, for \textbf{Case 2}, another simple calculation shows that
	$\displaystyle{\max_{1/2\leq\eta<1}\frac{\eta}{\eta-1} = -1}$ and $\displaystyle{\lim_{\eta\rightarrow 1^{-}} \frac{\eta}{\eta-1} =-\infty}$. Hence, $\alpha = \frac{2}{3} \implies \beta,\gamma = \frac{1}{2}$, i.e. $\frac{2}{3}\leq\alpha<1$ and $\frac{1}{2}\leq \beta,\gamma<1$.
	
	If instead $0<\frac{\alpha'}{\beta'}<1$, we still have the condition $(\alpha - 1)(\beta - 1)>0$ but in the second part of the argument the inequalities are reversed, i.e
	\begin{align}
		0<\frac{\alpha'}{\beta'}<1\implies \begin{cases}
			\alpha>\beta \quad \text{if}\quad \alpha, \beta > 1 \\
			\alpha<\beta \quad \text{if}\quad \alpha, \beta < 1
		\end{cases}.
	\end{align}
	This overall excludes \textbf{Cases 1}, \textbf{2}  and \textbf{4} but satisfies \textbf{Case 5}.
	
	In this situation we again have $\alpha\rightarrow 1 \implies \beta, \gamma\rightarrow 1$ and for fixed $\gamma$ we can write $\displaystyle{\lim_{\alpha\rightarrow\infty}\beta = \frac{1}{\gamma}}$. Given that $\gamma>0$, this implies $1<\beta<\infty$.
	
	Lastly, we have $0<\frac{\alpha'}{\gamma'}<1$, which implies $(\alpha - 1)(\gamma-1)>0$, excluding \textbf{Case 5} Similarly, we have following situations:
	\begin{align}
		0<\frac{\alpha'}{\gamma'}<1\implies \begin{cases}
			\alpha>\gamma \quad \text{if}\quad \alpha, \gamma > 1 \\
			\alpha<\gamma \quad \text{if}\quad \alpha, \gamma < 1
		\end{cases},
	\end{align}
	which exclude \textbf{Cases 1} and \textbf{2}. So \textbf{Case 4} is the only remaining case which is satisfied.
	
	We again have $\alpha\rightarrow 1 \implies \beta, \gamma\rightarrow 1$ and for fixed $\beta$, $\displaystyle{\lim_{\alpha\rightarrow 1/2} \gamma = \frac{2\beta-1}{3\beta-2}}$ and $\displaystyle{\lim_{\beta\rightarrow \infty}\frac{2\beta-1}{3\beta-2}} = \frac{2}{3}$. Hence $2/3<\gamma<1$ and $1<\beta<\infty.$
	
	We now treat the cases where $\gamma\leq 0$.
	Firstly, we eliminate the possibility that $\gamma\leq 0$ when $\frac{\alpha'}{\gamma'}\in(0,1)$. For this to be satisfied we require $\alpha'$ and $\gamma'$ to have the same sign. Since $\gamma\leq0\implies \gamma'\in [0,1)$ we require $\alpha>1\implies \alpha'>1$. The second requirement is $\alpha'\leq \gamma'$, however this is never satisfied given the valid ranges of $\alpha'$ and $\gamma'$.

	We now consider \textbf{Case 3}. When $\frac{\beta'}{\alpha'} \in(0,1)$ we have $\frac{\gamma'}{\alpha'}\in(0,1)$ so, as above, we require the same sign but the second requirement, $\gamma'\leq \alpha'$, is always satisfied. This condition also implies that $\beta>1$ in order to satisfy Eq.~\eqref{rel}. The valid ranges are then evident as we have no extra restriction on $\gamma$ and $\beta$, and when $\gamma = 0\implies\alpha=\beta$.

	Finally, for \textbf{Case 6}, we have $\frac{\alpha'}{\beta'}\in(0,1)$ and require $\alpha'$ and $\beta'$ to have the same sign and $|\alpha'|\leq|\beta'|$. If $\alpha, \beta>1$, Eq.~\eqref{rel} determines that $\alpha'\geq \beta'$ so this case cannot be used. Otherwise, when $\alpha, \beta<1$ we instead have $\alpha'\geq\beta'\implies |\alpha'|\leq|\beta'|$. The ranges can then be similarly determined.
\end{proof}
We may use the following conditions to summarise the above cases.
\begin{corollary}
	Given the assumptions in Lemma~\ref{ana}  we have that
	\begin{align}
		\alpha<\gamma<\beta & \implies (\alpha-1)(\beta-1)(\gamma^2-\gamma)>0\quad\text{and} \\
		\gamma<\beta<\alpha & \implies(\alpha-1)(\beta-1)(\gamma^2-\gamma)<0.
	\end{align}
\end{corollary}
\begin{proof}
	This is evident from examining each case of Lemma~\ref{ana}.
\end{proof}
\begin{corollary}\label{dualcond}
	Given the conditions in Lemma~\ref{ana} we have
	\begin{align}
		\frac{1}{\beta} + \frac{1}{\gamma} & \leq 2 \quad\text{if}\quad \alpha\leq \beta, \\
		\frac{1}{\beta} + \frac{1}{\gamma} & \geq 2 \quad\text{if}\quad \alpha\geq \beta.\label{aGreat}
	\end{align}
\end{corollary}
\begin{proof}
	Assume without loss of generality that $\frac{\alpha'}{\beta'}\geq 0$, then
	\begin{align}
		\frac{\gamma'+\beta'}{\beta'}  = \frac{\gamma'}{\beta'} + 1 & \geq 0 \\
		\implies \frac{1}{\beta'} + \frac{1}{\gamma'}               &
		\begin{cases}
			\geq 0\quad\text{if }\gamma'\geq0 \\
			\leq 0 \quad\text{if }\gamma' \leq0
		\end{cases}\\
		\implies \frac{1}{\beta} + \frac{1}{\gamma}               &
		\begin{cases}
			\leq 2\quad\text{if }\gamma'\geq0 \\
			\geq 2 \quad\text{if }\gamma' \leq0
		\end{cases}.
	\end{align}
	We can perform a symmetric argument where $\frac{\alpha'}{\gamma'}\geq 0$. Comparing these with the cases in Lemma~\ref{ana} we have our comparisons.
\end{proof}
\newpage\addcontentsline
{toc}{chapter}{\protect\numberline {\hfil}Bibliography}\bibliography{library}
\end{document}